%% file: main.tex
\documentclass[a4paper,UKenglish,cleveref, autoref, thm-restate]{lipics-v2021}

\input{preamble}
\input{macros}

\title{Tape Diagrams for Monoidal Monads} 

\author{Filippo Bonchi}{University of Pisa}{}{}{Bonchi is supported by the Ministero dell'Università e della Ricerca of Italy grant PRIN 2022 PNRR No. P2022HXNSC - RAP (Resource Awareness in Programming). This study was carried out within the National Centre on HPC, Big Data and Quantum Computing - SPOKE 10 (Quantum Computing) and received funding from the European Union Next-GenerationEU - National Recovery and Resilience Plan (NRRP) – MISSION 4 COMPONENT 2, INVESTMENT N. 1.4 – CUP N. I53C22000690001.}

\author{Cipriano Junior Cioffo}{University of Pisa}{}{}{}

\author{Alessandro Di Giorgio}{Tallinn University of Technology}{}{}{}

\author{Elena Di Lavore}{University of Oxford}{}{}{}

\authorrunning{F. Bonchi, C.J. Cioffo, A. Di Giorgio, E. Di Lavore} 

\Copyright{Filippo Bonchi and Cipriano Cioffo and Alessandro Di Giorgio and Elena Di Lavore} 

\ccsdesc[500]{Theory of computation~Logic}

\keywords{rig categories, string diagrams, monads, probabilistic control} 

\category{} 

\relatedversion{} 

\funding{This research was partly funded by the Advanced Research + Invention Agency (ARIA) Safeguarded AI Programme.}

\nolinenumbers 

\EventEditors{John Q. Open and Joan R. Access}
\EventNoEds{2}
\EventLongTitle{42nd Conference on Very Important Topics (CVIT 2016)}
\EventShortTitle{CVIT 2016}
\EventAcronym{CVIT}
\EventYear{2016}
\EventDate{December 24--27, 2016}
\EventLocation{Little Whinging, United Kingdom}
\EventLogo{}
\SeriesVolume{42}
\ArticleNo{23}

\begin{document}

\maketitle

\begin{abstract}
Tape diagrams provide a graphical representation for arrows of rig categories, namely categories equipped with two monoidal structures, $\piu$ and $\per$, where $\per$ distributes over $\piu$. However, their applicability is limited to categories where $\piu$ is a biproduct, i.e., both a categorical product and a coproduct. In this work, we 
extend tape diagrams to deal with Kleisli categories of symmetric monoidal monads, presented by algebraic theories.
\end{abstract}

\input{sections/intro}

\input{sections/CD}

\input{sections/rig}

\input{sections/cdfbrig}

\input{sections/sigmaerig}

\input{sections/cdsigmaerig}

\input{sections/sigmaetapes}

\input{sections/conclusion}

\bibliography{references}

\appendix

\input{appendices/coherence}

\input{sections/monoidal}

\input{appendices/rigapp}
 \input{appendices/fccdrigapp}

\input{appendices/appseigmae}

 \input{appendices/sigmaetapespp}

\end{document}

%% file: preamble.tex
\usepackage{array}
\usepackage{adjustbox}
\usepackage{tikz-cd}
\usepackage{tikz}
\usetikzlibrary{cd,arrows}
\usepackage{mathpartir}
\usepackage{caption}
\usepackage{subcaption}
\usepackage{amsfonts}
\usepackage{amsmath}
\usepackage[renew-matrix,renew-dots]{nicematrix}
\usepackage{float}
\usepackage{enumitem}
\usepackage{xcolor}
\usepackage{booktabs}
\usepackage{rotating}
\usepackage{multirow}
\usepackage{comment}
\usepackage{graphicx}
\usepackage{braket}
\usepackage{enumitem}
\usepackage{xparse}
\usepackage{xstring}
\usepackage{mathtools}
\usepackage{xifthen}
\usepackage{makecell}
\usepackage{txfonts}
\usepackage{hyperref}
\usepackage{cleveref}
\hypersetup{
    colorlinks,
    linkcolor={black!50!black},
    citecolor={blue!50!black},
    urlcolor={blue!80!black}
}
\usepackage{mathtools}
\usepackage{changepage}

\usepackage{quiver}

\usepackage{todonotes}

\usepackage{tabularx}

\usepackage{wrapfig}

%% file: sections/intro.tex
\section{Introduction}\label{sec:intro}

Motivated by the growing interest in formal languages that are both diagrammatic and compositional, an increasing number of works~\cite{BaezErbele-CategoriesInControl,DBLP:journals/pacmpl/BonchiHPSZ19,Bonchi2015,coecke2011interacting,Fong2015,DBLP:journals/corr/abs-2009-06836,Ghica2016,DBLP:conf/lics/MuroyaCG18,Piedeleu2021,DBLP:journals/jacm/BonchiGKSZ22} leverage \emph{string diagrams}~\cite{joyal1991geometry,selinger2010survey}. These diagrams formally represent arrows in a strict symmetric monoidal category freely generated by a monoidal signature. A symbol $\gen$ in the signature is depicted as a box $\Cgen{\gen}{}{}$, and arbitrary string diagrams are constructed through horizontal ($;$) and vertical ($\per$) composition of such symbols. 

While string diagrammatic languages are typically expressive enough to represent all relevant systems, they often become cumbersome when dealing with \emph{control}. For example, representing quantum control within the ZX-calculus~\cite{Coecke2008,Coecke2017} or probabilistic control of Boolean circuits~\cite{piedeleu2025boolean} requires ingenuity. Various works~\cite{DBLP:journals/corr/abs-2103-07960,zhao2021analyzing,stollenwerk2022diagrammatic,villoria2024enriching,DBLP:conf/fossacs/BoisseauP22} adopt a hybrid approach combining diagrammatic and algebraic syntax. Others~\cite{duncan2009generalised, james2012information, staton2015algebraic} extend beyond monoidal categories to \emph{rig categories}~\cite{laplaza_coherence_1972,johnson2021bimonoidal}, roughly categories equipped with two monoidal products, $\per$ and $\piu$, where $\per$ distributes over $\piu$.

A key challenge in visualizing arrows in rig categories arises from the need to compose them not only horizontally ($;$) and vertically ($\per$) but also using the additional monoidal product $\piu$. One natural approach, implemented by \emph{sheet diagrams}~\cite{comfort2020sheet}, utilizes three dimensions. An alternative, which remains within two dimensions, is offered by \emph{tape diagrams}~\cite{bonchi2023deconstructing}.

Intuitively, tape diagrams are \emph{string diagrams of string diagrams}: the vertical composition of inner diagrams represents $\per$, whereas the vertical composition of outer diagrams represents $\piu$.
\begin{equation*}\label{eq:example intro}
    
    \InputIfFileExists{zx/control.tikz}{}{\input{./tikz/zx/control.tikz}}
 \qquad\qquad \scalebox{0.8}{
    \InputIfFileExists{addition/additionDef.tikz}{}{\input{./tikz/addition/additionDef.tikz}}
} \qquad {\texttt{while(x>0)\{x:=x-1; y:=y+1\}; return y}}
\end{equation*}

For instance, the leftmost tape diagram above represents a controlled unitary using ZX-circuits inside the tapes. The central tape diagram denotes the imperative program on the right. In both cases, the inner diagrams correspond to data flow, while the outer diagrams capture control flow. The leftmost diagram is interpreted within the category of complex matrices, whereas the central diagram within the category of sets and relations. In both categories, $\piu$ is a \emph{biproduct}, namely a product and a coproduct. This fact guarantees the existence of natural \emph{comonoids} and \emph{monoids} for $\piu$, which provide the crucial linguistic primitives for representing both the branching and the join of tapes:
\begin{center}
$\smash{\Tcomonoid{\!\!\!\!}}$ and $\smash{\Tmonoid{\!\!\!\!}}$
\end{center}
Unfortunately, in many contexts, such as probabilistic systems, $\piu$ cannot be assumed to be a biproduct.

\medskip

In this paper, we generalize the framework of tape diagrams from~\cite{bonchi2023deconstructing} to accommodate categories where $\piu$ is not necessarily a biproduct. Specifically, our proposed tape diagrams represent arrows of rig categories that share the structure of Kleisli categories for \emph{monoidal monads}~\cite{kock1972strong}. Like in~\cite{bonchi2023deconstructing}, we deliberately avoid \emph{monoidal traces}, which --while useful for iteration as illustrated in~\cite{bonchi2024diagrammatic} and in the central tape above-- are not essential for our development.

Our starting observation is that in Kleisli categories, when the monoidal monad is presented by an algebraic theory $\mathbb{T}=(\Sigma,E)$, each symbol $f$ of the signature $\Sigma$ induces a natural transformation $\scalebox{0.8}{\Talgop{\!\!\!\!}{f}}$ that we use in place of the missing $\smash{\Tcomonoid{\!\!\!\!}}$. Conveniently, $\smash{\Tmonoid{\!\!\!\!}}$ always exists in any Kleisli category, as $\piu$ remains a coproduct~\cite{borceux2}. Notably, when $\mathbb{T}$ corresponds to the algebraic theory of commutative monoids, our $\mathbb{T}$-tape diagrams naturally generalize those in~\cite{bonchi2023deconstructing} (see Remark~\ref{remark:Generalization}).

Another fundamental aspect of our work concerns comonoids for the product $\per$. These structures, which reside within the Kleisli category of any commutative monad~\cite{fritz_2020}, have recently gained attention due to their significance in copy-discard~\cite{cho2019disintegration} (aka gs-monoidal~\cite{GadducciCorradini-termgraph}) and Markov categories~\cite{fritz_2020}.

We analyze, in the setting of rig categories, the interaction between $\per$-comonoids and $\piu$-monoids in Kleisli categories of monoidal monads, introducing the notion of \emph{finite coproduct-copy discard rig categories} (Definition~\ref{def: distributive gs-monoidal category}). By further considering the natural transformation arising from the algebraic theory $\mathbb{T}$, we obtain  \emph{$\mathbb{T}$-copy discard rig categories} (Definition \ref{def: T distributive gs-monoidal category}). We demonstrate that any Kleisli category arising from a monoidal monad presented by $\mathbb{T}$ inherently possesses this structure (Proposition~\ref{prop:kleisliTrig}). Furthermore, we establish that our tape diagrammatic language  provides a $\mathbb{T}$-copy discard rig category freely generated by a monoidal signature (Theorem~\ref{th:tapes-free-T-cd}). These results ensure that, whenever the monoidal signature is interpreted within such a Kleisli category, the interpretation extends uniquely to a rig functor, thus providing a compositional semantics for tapes. We conclude by illustrating how interpreting Boolean circuits within the Kleisli category of subdistributions yields a language of tapes that both encodes~\cite{piedeleu2025boolean} and easily manages probabilistic control (Example~\ref{example: probabilistic control}).

\textbf{Synopsis.} Section~\ref{sec:running} illustrates the key structures within the Kleisli category of subdistributions. Section~\ref{sec:CD} recalls the fundamental concepts of copy-discard and finite (co)product categories, alongside monoidal monads. Section~\ref{sec:rigcategories} provides an overview of rig categories, followed by our definition of finite coproduct-copy discard rig categories in Section~\ref{sec:fccdrig}. Section~\ref{sec:Trig} introduces $\mathbb{T}$-rig categories and $\mathbb{T}$-copy discard rig categories, while their corresponding tape diagrams are in Section \ref{sec:tapediagrams}. The appendix includes coherence diagrams and proofs. We assume familiarity with monads, Kleisli categories, and monoidal categories. For a symmetric monoidal category $(\Cat{C}, \perG, \unoG)$, we denote symmetries by $\sigma^\perG_{X,Y}$ and the remaining structural isomorphisms by $\assoc{X}{Y}{Z}$, $\lunit{X}$, and $\runit{X}$. For a natural isomorphism $\phi$, its inverse is denoted by $\phi^-$.

%% file: sections/CD.tex
\section{Motivating Example}\label{sec:running} 
To provide a clearer illustration of our approach, we begin with the Kleisli category \(\Cat{Set}_{\subdistr}\) associated with the monad of finitely supported subdistributions \(\subdistr\) on   \(\Cat{Set}\) (see e.g., \cite{hasuo2007generic}).
%
Objects are sets; morphisms \(f \colon X \to Y\) are functions \(X \to \subdistr(Y)\) mapping each \(x \in X\) into a finitely supported subdistribution over $Y$.
  We write \(f(y \mid x)\) for \(f(x)(y)\)  as this number represents the probability that \(f\) returns \(y\) given the input \(x\).
  Identities \(\id{X} \colon X \to \subdistr(X)\) map each element \(x \in X\) to \(\delta_{x}\), the Dirac distribution at \(x\).
  For two functions \(f \colon X \to \subdistr(Y)\) and \(g \colon Y \to \subdistr(Z)\), their composition in \(\Cat{Set}_{\subdistr}\) is obtained by summing on the middle variable: \(f ; g (z \mid x) \defeq \sum_{y \in Y} f(y \mid x) \cdot g(z \mid y)\).

The category \(\Cat{Set}_{\subdistr}\) carries two symmetric monoidal structures: $ (\Cat{Set}_{\subdistr}, \per, \uno)$ and $(\Cat{Set}_{\subdistr}, \piu, \zero)$. The monoidal product $\per$ is defined on objects as the cartesian product of sets, hereafter denoted by $\times$, with unit the singleton $\uno \defeq \{\bullet\}$; $\piu$ is the disjoint union of sets with unit the empty set $\zero \defeq \{\}$. Hereafter we denote the disjoint union of two sets $X$ and $Y$ by $X+Y \defeq \{(x,0) \mid x\in X\} \cup \{(y,1) \mid y \in Y\}$ where $0$ and $1$ are tags used to distinguish the set of provenance. For arrows  \(f \colon X \to Y\) and \(g \colon X' \to Y'\), $f\per g\colon X\times X' \to Y \times Y'$ and $f \piu g \colon X+X' \to Y+Y'$ are defined as
\begin{equation}\label{ex:products}
f \per g(y,y' \mid x,x') \defeq f(y \mid x) \cdot g(y' \mid x') \quad  f\piu g(v \mid u) \defeq \begin{cases} f(y \mid x) & \text{if } u=(x,0) \text{ and } v=(y,0)\\ g(y' \mid x') & \text{if } u=(x',1) \text{ and } v=(y',1) \\ 0 & \text{otherwise} \end{cases}
\end{equation}
for all $x\in X$, $x'\in X'$, $y\in Y$, $y'\in Y'$, $u\in X+X'$ and $v\in Y+Y'$. 

Within the two monoidal structures, for all objects $X$, there are arrows particularly relevant for us:
%
%
%
\begin{equation}\label{ex:comonoids}
\begin{array}{cccc}
\begin{array}{rcl}
\copier{X} \colon  X & \to & X \times X \\
x & \mapsto & \delta_{(x,x)}
\end{array}
&
\begin{array}{rcl}
\discharger{X}  \colon X & \to & \uno\\
x & \mapsto& \delta_{\bullet}
\end{array}
&
\begin{array}{rcl}
\codiag{X}  \colon  X+X & \to & X\\
(x,i) & \mapsto& \delta_{x}
\end{array}
&
\begin{array}{rcl}
\cobang{X}\colon  \zero &\to& X\\
\text{ }
\end{array}
\end{array}
\end{equation}
Indeed $(\copier{X},\discharger{X})$ provide a \emph{comonoid} for the monoidal category  $ (\Cat{Set}_{\subdistr}, \per, \uno)$, while $(\codiag{X},\cobang{X})$ a \emph{monoid} for $(\Cat{Set}_{\subdistr}, \piu, \zero)$. The following arrows, where $p\in (0,1)$, are also crucial:
\begin{equation}\label{ex:op}
 \Br{+_p}_X \colon X \to X+X \qquad \Br{\star}_X \colon X \to 0\end{equation}
The function $\Br{+_p}_X$ maps each $x\in X$ into the distribution that assigns to $(x,0)$ probability $p$ and to $(x,1)$ probability $1-p$. Instead, $\Br{\star}_X$ is the  function mapping each $x$ into the null subdistribution.

The structures above constitute the fundamental linguistic constructs for the language we are about to introduce. Naturally, they are not arbitrary: 
the structures in \eqref{ex:products} and \eqref{ex:comonoids} exist in the Kleisli category of any \emph{symmetric monoidal monad}, while those in \eqref{ex:op} are specific to the monad \(\subdistr\). However, analogous structures can also be found in monads equipped with an \emph{algebraic presentation} \(\mathbb{T}\).

Before introducing our diagrammatic language in Section \ref{sec:tapediagrams}, we examine how these structures interact: \(\per, \uno,\copier{},\discharger{}\) define a copy-discard category (Definition \ref{def:cd}); \(\piu, \zero,\codiag{},\cobang{}\) form a finite coproduct category (Definition \ref{def:fp}). Together, they constitute what we refer to as a \emph{finite coproduct-copy discard rig category} (Definition \ref{def: distributive gs-monoidal category}). By further incorporating the arrows provided by the algebraic theory \(\mathbb{T}\), we obtain a \emph{\(\mathbb{T}\)-copy discard rig category} (Definition \ref{def:Trig cat}).

%
%
%
%
%

\section{Copy Discard and Finite (Co)Product Categories}\label{sec:CD}

Fox's theorem \cite{fox1976coalgebras} characterises categories with finite products as symmetric monoidal categories equipped with coherent comonoids which are natural.
It is convenient to split the definition in two:

%

\begin{table}
{\small
\begin{equation}\label{eq:comonoid}
\copier{X}; (\copier{X} \perG \id{X}); \assoc{X}{X}{X} = \copier{X} ; (\id{X} \perG \copier{X})\quad
\copier{X}; (\discharger{X} \perG \id{X}) ; \lunit{X} = \id{X} = \copier{X} ; (\id{X} \perG \discharger{X}); \runit{X} \quad
\copier{X}; \sigma_{X,X}^\perG = \copier{X}
\end{equation}
\vspace{-0.5cm}
\begin{equation}\label{eq:comonoidcoherence}
\begin{array}{c}\copier{X\perG Y} = (\copier{X}\perG\copier{Y}) ; \assoc{X}{X}{Y\perG Y} ; (\id{X} \perG \assoc{X}{Y}{Y}^-) ; ( \id{X} \perG (\sigma^{\perG}_{X,Y} \perG\id{Y})) ; (\id{X} \perG \assoc{Y}{X}{Y}) ; \assoc{X}{Y}{X\perG Y}^- \\
\discharger{X\perG Y} = \discharger{X} \perG \discharger{Y} ; \lunit{\unoG} \qquad 
\copier{\unoG} = \lunit{\unoG}^- \qquad
\discharger{\unoG} = \id{\unoG}
\end{array}
\end{equation}
\vspace{-0.5cm}
\begin{equation}\label{eq:naturality}
f; \copier{Y} = \copier{X}; (f \perG f) \qquad f;\discharger{Y} = \discharger{X} \qquad \qquad \text{for all }f\colon X \to Y
\end{equation}
}
\caption{Axioms of finite product categories. 
See also Figure~\ref{fig:comonoidax} and \ref{fig:fpcoherence} in Appendix \ref{app:coherence axioms}.}\label{tab:axiomsfc}
\end{table}

\begin{definition}\label{def:cd}
A \emph{copy discard} (cd, for short) \emph{category} is a symmetric monoidal category $(\Cat C, \perG, \unoG)$ and, for all objects $X$, two arrows $\copier{X}\colon X \to X \perG X$ and $\discharger{X}  \colon X \to \unoG$ such that
\begin{itemize}
\item $\copier{X}$ and $\discharger{X}$ form a comonoid, i.e., the laws in Table \ref{tab:axiomsfc}.\eqref{eq:comonoid} hold;
\item $\copier{X}$ and $\discharger{X}$ satisfy the coherence axioms in Table \ref{tab:axiomsfc}.\eqref{eq:comonoidcoherence}.
\end{itemize}
A \emph{morphism of cd categories} is a strong symmetric monoidal functor preserving comonoids.
\end{definition}

\begin{definition}\label{def:fp}
A \emph{finite product} (fp, for short) \emph{category} is a copy discard category such that
\begin{itemize}
\item $\copier{}$ and $\discharger{}$  are natural transformations, i.e., for all $f\colon X \to Y$ in $\Cat{C}$ the laws in Table \ref{tab:axiomsfc}.\eqref{eq:naturality} hold;
\end{itemize}
A \emph{finite coproduct} (fc) \emph{category} is a symmetric monoidal category with natural and coherent monoids $(\codiag{},\cobang{})$, i.e., $\Cat C$ is a finite coproduct category iff $\Cat{C}^{op}$ is a fp category. A morphism of (fp)fc categories is a strong symmetric monoidal functor preserving (co)monoids.
\end{definition}
Definition \ref{def:fp} is equivalent to require the existence of categorical (co)products and (initial) final objects but, differently from the usual definition, it does not rely on universal properties (see Remark \ref{rem:productcomonoids} in App. \ref{ssec:fbcat} for more details). 
Moreover, since any monoidal category is equivalent to a \emph{strict} one \cite{mac_lane_categories_1978}, namely a monoidal category where the structural isomorphism $\assoc{X}{Y}{Z}\colon (X \perG Y) \perG Z \to X\perG (Y \perG Z)$, $\lunit{X}\colon \unoG \perG X \to X$ and $\runit{X} \colon X \perG \unoG \to X$ are identities, the laws in Table~\ref{tab:axiomsfc} can remarkably simplified. In particular, they acquire an intuitive geometrical meaning when drawn as \emph{string diagrams}, formally arrows of a freely generated strict symmetric monoidal category (see App. \ref{sec:monoidal}).
By depicting, $\copier{X}$ as $\comonoid{X}$, $\discharger{X}$ as $\counit{X}$ and an arbitrary arrow $f\colon X \to Y$ as $\boxTypedCirc{f}{X}{Y}$, the laws in Table~\ref{tab:axiomsfc} appear as follows.
  \[
  \!\!\!\!\begin{array}{c @{} c@{}c@{}c}
  \begin{array}{c@{}c@{}c @{\quad} c@{}c@{}c @{\quad} c@{}c@{}c}
    
    \InputIfFileExists{fbAx/comonoid/assoc/lhs.tikz}{}{\input{./tikz/fbAx/comonoid/assoc/lhs.tikz}}
 & = & 
    \InputIfFileExists{fbAx/comonoid/assoc/rhs.tikz}{}{\input{./tikz/fbAx/comonoid/assoc/rhs.tikz}}

    &
    
    \InputIfFileExists{fbAx/comonoid/unit/lhs.tikz}{}{\input{./tikz/fbAx/comonoid/unit/lhs.tikz}}
 & = & 
    \InputIfFileExists{fbAx/comonoid/unit/rhs.tikz}{}{\input{./tikz/fbAx/comonoid/unit/rhs.tikz}}

    &
    
    \InputIfFileExists{fbAx/comonoid/comm/lhs.tikz}{}{\input{./tikz/fbAx/comonoid/comm/lhs.tikz}}
 &=  & 
    \InputIfFileExists{fbAx/comonoid/comm/rhs.tikz}{}{\input{./tikz/fbAx/comonoid/comm/rhs.tikz}}

  \end{array}
  &
  
    \InputIfFileExists{fbAx/comonoid/nat/copy/lhs.tikz}{}{\input{./tikz/fbAx/comonoid/nat/copy/lhs.tikz}}
 &=& 
    \InputIfFileExists{fbAx/comonoid/nat/copy/rhs.tikz}{}{\input{./tikz/fbAx/comonoid/nat/copy/rhs.tikz}}

  \\
  \begin{array}{c@{\quad}c@{\quad}c@{\quad}c}
    \CBcopier[][I] = 
    \InputIfFileExists{empty.tikz}{}{\input{./tikz/empty.tikz}}
 & 
    \InputIfFileExists{fbAx/coherence/copierLhs.tikz}{}{\input{./tikz/fbAx/coherence/copierLhs.tikz}}
 = 
    \InputIfFileExists{fbAx/coherence/copierRhs.tikz}{}{\input{./tikz/fbAx/coherence/copierRhs.tikz}}
 & \CBdischarger[][I] = 
    \InputIfFileExists{empty.tikz}{}{\input{./tikz/empty.tikz}}
 & 
    \InputIfFileExists{fbAx/coherence/dischargerLhs.tikz}{}{\input{./tikz/fbAx/coherence/dischargerLhs.tikz}}
 = 
    \InputIfFileExists{fbAx/coherence/dischargerRhs.tikz}{}{\input{./tikz/fbAx/coherence/dischargerRhs.tikz}}

  \end{array}
  &
  
    \InputIfFileExists{fbAx/comonoid/nat/discard/lhs.tikz}{}{\input{./tikz/fbAx/comonoid/nat/discard/lhs.tikz}}
 &=& 
    \InputIfFileExists{fbAx/comonoid/nat/discard/rhs.tikz}{}{\input{./tikz/fbAx/comonoid/nat/discard/rhs.tikz}}

  \end{array}
  \]


The archetypical examples of finite product and finite coproduct category are  \((\Cat{Set}, \times, \uno)\) and  \((\Cat{Set}, +, 0)\):  \emph{copy} $\copier{X} \colon X \to X \times X$ is  $\langle \id{X},\id{X}\rangle$;  the \emph{discard} $\discharger{X}\colon X \to \uno$ is the unique function into $\uno$; $\codiag{X}\colon X+X \to X$ is $[\id{X},\id{X}]$ and $\cobang{X}\colon \zero \to X$ is the unique function from $\zero$.  

As anticipated in Section \ref{sec:running}, also $(\Cat{Set}_{\subdistr}, \piu, \zero)$ is a finite coproduct category. More generally, the \emph{Kleisli category} of a monad $T\colon\Cat{C}\to \Cat{C}$, hereafter denoted by $\Cat{C}_T$, is a fc category whenever $\Cat{C}$ is fc~\cite{borceux2}.  Instead if $\Cat{C}$ is a fp category, then $\Cat{C}_T$ is \emph{not} necessarily fp, actually not even monoidal. 

\begin{definition}
Let $(\Cat{C},\perG,\unoG)$ be a symmetric monoidal category. A symmetric monoidal monad  consists of a monad \((T, \eta, \mu)\) over $\Cat{C}$ and a natural transformation \(m_{X,Y} \colon T(X) \perG T(Y) \to T(X \perG Y)\) satisfying the expected laws (see Figure \ref{fig:monoidalmonads} and also \cite{eilenberg1966closed,kock1972strong,kock1970monads}). 
\end{definition}
\begin{example}\label{ex:monoidal monads}
Any monad on  \((\Cat{Set}, +, 0)\) and, more generally, any monad on a fc category is symmetric monoidal: take $m_{X,Y}=[T(\iota_X),T(\iota_Y)]$. 
Instead, not all monads on \((\Cat{Set}, \times, \uno)\)  are symmetric monoidal: for instance, the monad of convex sets of (sub)distributions (see e.g.~\cite{jacobs2008coalgebraic,goy2020combining}), combining probabilistic and non-deterministic choice. The monad \(\subdistr\) is symmetric monoidal~\cite{jacobs2018probability}: $m_{X,Y}$ is defined for all $d_X\in \subdistr(X)$, $d_Y\in \subdistr(Y)$, $x\in X$ and $y\in Y$ as  $m_{X,Y}(d_X,d_Y)(x,y)=d_X(x)\cdot d_Y(y)$.
%
\end{example}
When $T$ is a symmetric monoidal monad, then one can define a monoidal product $\perG_T$ on $\Cat{C}_T$: on objects, $\perG_T$ is defined as $\perG$ in $\Cat{C}$; 
for arrows $f\colon X_1 \to X_2$, $g\colon Y_1 \to Y_2$ in $\Cat{C}_T$, $f\perG_T g \defeq (f\perG g) ; m_{X_2,Y_2}$.
Moreover, if $\Cat{C}$ has coherent monoids or comonoids, then for all objects $X$, one can define
\begin{equation}\label{eq:copierflat}
\copier{X}^\flat \defeq  \copier{X} ;\eta_{X\perG X} \qquad \discharger{X}^\flat \defeq \discharger{X} ; \eta_{\unoG} \qquad \text{ and } \qquad   \codiag{X}^\flat \defeq  \codiag{X} ;\eta_{X} \qquad  \cobang{X}^\flat \defeq \cobang{X} ; \eta_{X}\text{,}
\end{equation}
forming coherent (co)monoids in $\Cat{C}_T$. If $\Cat{C}$ is fc, then the above monoids are natural in $\Cat{C}_T$. Instead, if $\Cat{C}$ is fp, then the comonoids are \emph{not} necessarily natural. In other words, $\Cat{C}_{T}$ may not have finite products, but it is always a \emph{copy-discard} category (Definition \ref{def:cd}).

\begin{proposition}[From \cite{cioffogadduccitrotta}]\label{prop: Kleisli gs e cogs}
  Let $T\colon \Cat{C}\to \Cat{C}$ be a symmetric monoidal monad on a cd category $(\Cat{C},\perG,\unoG, \copier{}, \discharger{})$.
  Then the Kleisli category $(\Cat{C}_T, \perG_T, \unoG, \copier{}^\flat, \discharger{}^\flat)$ is a cd category. 
\end{proposition}

In a cd category $(\Cat{C},\perG,\unoG,\copier{},\discharger{})$, an arrow $f\colon X\to Y$ is called \emph{deterministic} if $f;\copier{Y}= \copier{X};(f\perG f)$ and \emph{total} if
$f;\discharger{Y}= \discharger{X}$. The category of total deterministic morphisms is denoted by $\mathrm{Map}(\Cat{C})$ and it is a finite product category.

Since  \((\Cat{Set}, \times, \uno)\) is a finite product category and \(\subdistr\) is a symmetric monoidal monad on  \((\Cat{Set}, \times, \uno)\) (Example \ref{ex:monoidal monads}) then, by Proposition \ref{prop: Kleisli gs e cogs}, the Kleisli category $\Cat{Set}_{\subdistr}$ is a copy discard category: monoidal product and comonoids are exactly those defined in \eqref{ex:products} and \eqref{ex:comonoids}. In $\Cat{Set}_{\subdistr}$, deterministic and total morphism are simply functions.

%% file: tikz/fbAx/comonoid/assoc/lhs.tikz
\begin{tikzpicture}
	\begin{pgfonlayer}{nodelayer}
		\node [style=label] (120) at (3, 1.25) {$X$};
		\node [style=black] (121) at (0, 0) {};
		\node [style=none] (122) at (0.75, -0.625) {};
		\node [style=none] (123) at (0.75, 0.625) {};
		\node [style=label] (124) at (-1.5, 0) {$X$};
		\node [style=none] (125) at (2.5, -0.625) {};
		\node [style=none] (127) at (-1, 0) {};
		\node [style=black] (128) at (1.5, 0.625) {};
		\node [style=none] (129) at (2.25, 0) {};
		\node [style=none] (130) at (2.25, 1.25) {};
		\node [style=none] (131) at (2.5, 0) {};
		\node [style=none] (132) at (2.5, 1.25) {};
		\node [style=label] (133) at (3, 0) {$X$};
		\node [style=label] (134) at (3, -0.625) {$X$};
	\end{pgfonlayer}
	\begin{pgfonlayer}{edgelayer}
		\draw [bend right] (123.center) to (121);
		\draw [bend right] (121) to (122.center);
		\draw (125.center) to (122.center);
		\draw (127.center) to (121);
		\draw [bend right] (130.center) to (128);
		\draw [bend right] (128) to (129.center);
		\draw (130.center) to (132.center);
		\draw (131.center) to (129.center);
		\draw (123.center) to (128);
	\end{pgfonlayer}
\end{tikzpicture}

%% file: tikz/fbAx/comonoid/assoc/rhs.tikz
\begin{tikzpicture}
	\begin{pgfonlayer}{nodelayer}
		\node [style=label] (120) at (3, -0.625) {$X$};
		\node [style=black] (121) at (0, 0.625) {};
		\node [style=none] (122) at (0.75, 1.25) {};
		\node [style=none] (123) at (0.75, 0) {};
		\node [style=label] (124) at (-1.5, 0.625) {$X$};
		\node [style=none] (125) at (2.5, 1.25) {};
		\node [style=none] (127) at (-1, 0.625) {};
		\node [style=black] (128) at (1.5, 0) {};
		\node [style=none] (129) at (2.25, 0.625) {};
		\node [style=none] (130) at (2.25, -0.625) {};
		\node [style=none] (131) at (2.5, 0.625) {};
		\node [style=none] (132) at (2.5, -0.625) {};
		\node [style=label] (133) at (3, 0.625) {$X$};
		\node [style=label] (134) at (3, 1.25) {$X$};
	\end{pgfonlayer}
	\begin{pgfonlayer}{edgelayer}
		\draw [bend left] (123.center) to (121);
		\draw [bend left] (121) to (122.center);
		\draw (125.center) to (122.center);
		\draw (127.center) to (121);
		\draw [bend left] (130.center) to (128);
		\draw [bend left] (128) to (129.center);
		\draw (130.center) to (132.center);
		\draw (131.center) to (129.center);
		\draw (123.center) to (128);
	\end{pgfonlayer}
\end{tikzpicture}

%% file: tikz/fbAx/comonoid/unit/lhs.tikz
\begin{tikzpicture}
	\begin{pgfonlayer}{nodelayer}
		\node [style=black] (121) at (0, 0) {};
		\node [style=none] (122) at (0.75, 0.625) {};
		\node [style=none] (123) at (0.75, -0.625) {};
		\node [style=label] (124) at (-1.5, 0) {$X$};
		\node [style=black] (125) at (1, 0.625) {};
		\node [style=none] (127) at (-1, 0) {};
		\node [style=none] (135) at (1.25, -0.625) {};
		\node [style=label] (140) at (1.75, -0.625) {$X$};
	\end{pgfonlayer}
	\begin{pgfonlayer}{edgelayer}
		\draw [bend left] (123.center) to (121);
		\draw [bend left] (121) to (122.center);
		\draw (127.center) to (121);
		\draw (122.center) to (125);
		\draw (135.center) to (123.center);
	\end{pgfonlayer}
\end{tikzpicture}

%% file: tikz/fbAx/comonoid/unit/rhs.tikz
\begin{tikzpicture}
	\begin{pgfonlayer}{nodelayer}
		\node [style=none] (123) at (-1, 0) {};
		\node [style=label] (124) at (-1.5, 0) {$X$};
		\node [style=none] (135) at (1, 0) {};
		\node [style=label] (140) at (1.5, 0) {$X$};
	\end{pgfonlayer}
	\begin{pgfonlayer}{edgelayer}
		\draw (135.center) to (123.center);
	\end{pgfonlayer}
\end{tikzpicture}

%% file: tikz/fbAx/comonoid/comm/lhs.tikz
\begin{tikzpicture}
	\begin{pgfonlayer}{nodelayer}
		\node [style=black] (121) at (0, 0) {};
		\node [style=none] (122) at (0.75, 0.625) {};
		\node [style=none] (123) at (0.75, -0.625) {};
		\node [style=label] (124) at (-1.5, 0) {$X$};
		\node [style=none] (125) at (2.5, 0.625) {};
		\node [style=none] (127) at (-1, 0) {};
		\node [style=label] (134) at (3, 0.625) {$X$};
		\node [style=none] (135) at (2.5, -0.625) {};
		\node [style=none] (136) at (1, -0.625) {};
		\node [style=none] (137) at (1, 0.625) {};
		\node [style=none] (138) at (2.25, -0.625) {};
		\node [style=none] (139) at (2.25, 0.625) {};
		\node [style=label] (140) at (3, -0.625) {$X$};
	\end{pgfonlayer}
	\begin{pgfonlayer}{edgelayer}
		\draw [bend left] (123.center) to (121);
		\draw [bend left] (121) to (122.center);
		\draw (127.center) to (121);
		\draw [in=180, out=0, looseness=0.75] (137.center) to (138.center);
		\draw [in=0, out=-180, looseness=0.75] (139.center) to (136.center);
		\draw (122.center) to (137.center);
		\draw (135.center) to (138.center);
		\draw (125.center) to (139.center);
		\draw (136.center) to (123.center);
	\end{pgfonlayer}
\end{tikzpicture}

%% file: tikz/fbAx/comonoid/comm/rhs.tikz
\begin{tikzpicture}
	\begin{pgfonlayer}{nodelayer}
		\node [style=black] (121) at (0, 0) {};
		\node [style=none] (122) at (0.75, 0.625) {};
		\node [style=none] (123) at (0.75, -0.625) {};
		\node [style=label] (124) at (-1.5, 0) {$X$};
		\node [style=none] (125) at (1.25, 0.625) {};
		\node [style=none] (127) at (-1, 0) {};
		\node [style=label] (134) at (1.75, 0.625) {$X$};
		\node [style=none] (135) at (1.25, -0.625) {};
		\node [style=label] (140) at (1.75, -0.625) {$X$};
	\end{pgfonlayer}
	\begin{pgfonlayer}{edgelayer}
		\draw [bend left] (123.center) to (121);
		\draw [bend left] (121) to (122.center);
		\draw (127.center) to (121);
		\draw (122.center) to (125.center);
		\draw (135.center) to (123.center);
	\end{pgfonlayer}
\end{tikzpicture}

%% file: tikz/fbAx/comonoid/nat/copy/lhs.tikz
\begin{tikzpicture}
	\begin{pgfonlayer}{nodelayer}
		\node [style=black] (121) at (0, 0) {};
		\node [style=none] (122) at (0.75, 0.625) {};
		\node [style=none] (123) at (0.75, -0.625) {};
		\node [style=label] (124) at (-3.5, 0) {$X$};
		\node [style=none] (125) at (1.25, 0.625) {};
		\node [style=label] (134) at (1.75, 0.625) {$Y$};
		\node [style=none] (135) at (1.25, -0.625) {};
		\node [style=label] (140) at (1.75, -0.625) {$Y$};
		\node [style=none] (141) at (-3, 0) {};
		\node [style=bbox] (142) at (-1.5, 0) {$f$};
	\end{pgfonlayer}
	\begin{pgfonlayer}{edgelayer}
		\draw [bend left] (123.center) to (121);
		\draw [bend left] (121) to (122.center);
		\draw (122.center) to (125.center);
		\draw (135.center) to (123.center);
		\draw (141.center) to (142);
		\draw (142) to (121);
	\end{pgfonlayer}
\end{tikzpicture}

%% file: tikz/fbAx/comonoid/nat/copy/rhs.tikz
\begin{tikzpicture}
	\begin{pgfonlayer}{nodelayer}
		\node [style=black] (121) at (0.25, 0) {};
		\node [style=none] (122) at (1, 0.625) {};
		\node [style=none] (123) at (1, -0.625) {};
		\node [style=label] (124) at (-1.25, 0) {$X$};
		\node [style=none] (125) at (3, 0.625) {};
		\node [style=label] (134) at (3.5, 0.625) {$Y$};
		\node [style=none] (135) at (3, -0.625) {};
		\node [style=label] (140) at (3.5, -0.625) {$Y$};
		\node [style=none] (141) at (-0.75, 0) {};
		\node [style=bbox, scale=0.8] (142) at (2, 0.625) {$f$};
		\node [style=bbox, scale=0.8] (143) at (2, -0.625) {$f$};
	\end{pgfonlayer}
	\begin{pgfonlayer}{edgelayer}
		\draw [bend left] (123.center) to (121);
		\draw [bend left] (121) to (122.center);
		\draw (122.center) to (125.center);
		\draw (135.center) to (123.center);
		\draw (141.center) to (121);
	\end{pgfonlayer}
\end{tikzpicture}

%% file: tikz/fbAx/coherence/dischargerLhs.tikz
\begin{tikzpicture}
	\begin{pgfonlayer}{nodelayer}
		\node [style=none] (6) at (-1.25, 0) {};
		\node [style=black] (7) at (-0.25, 0) {};
		\node [style=label] (8) at (-2.25, 0) {$X \perG Y$};
	\end{pgfonlayer}
	\begin{pgfonlayer}{edgelayer}
		\draw (6.center) to (7);
	\end{pgfonlayer}
\end{tikzpicture}

%% file: tikz/fbAx/comonoid/nat/discard/lhs.tikz
\begin{tikzpicture}
	\begin{pgfonlayer}{nodelayer}
		\node [style=black] (121) at (0, 0) {};
		\node [style=label] (124) at (-3.5, 0) {$X$};
		\node [style=none] (141) at (-3, 0) {};
		\node [style=bbox] (142) at (-1.5, 0) {$f$};
	\end{pgfonlayer}
	\begin{pgfonlayer}{edgelayer}
		\draw (141.center) to (142);
		\draw (142) to (121);
	\end{pgfonlayer}
\end{tikzpicture}

%% file: tikz/fbAx/comonoid/nat/discard/rhs.tikz
\begin{tikzpicture}
	\begin{pgfonlayer}{nodelayer}
		\node [style=black] (121) at (-1.5, 0) {};
		\node [style=label] (124) at (-3.5, 0) {$X$};
		\node [style=none] (141) at (-3, 0) {};
	\end{pgfonlayer}
	\begin{pgfonlayer}{edgelayer}
		\draw (141.center) to (121);
	\end{pgfonlayer}
\end{tikzpicture}

%% file: sections/rig.tex
\section{Rig Categories}\label{sec:rigcategories}


We have seen that $\Cat{Set}_{\subdistr}$ carries two monoidal categories $(\Cat{Set}_{\subdistr},\per,\uno)$ and $(\Cat{Set}_{\subdistr},\piu, \zero)$. The appropriate categorical setting to study their interactions is provided by rig categories~\cite{laplaza_coherence_1972}. 
An extensive treatment was recently given in~\cite{johnson2021bimonoidal}, from which we borrow most of the notation in this paper.

\begin{definition}\label{def:rig}
    A \emph{rig category} is a category $\Cat{C}$ with 
    two symmetric monoidal structures $(\Cat{C}, \per, \uno)$ and 
    $(\Cat{C}, \piu, \zero)$ and natural isomorphisms 
    \[ \dl{X}{Y}{Z} \colon X \per (Y \piu Z) \to (X \per Y) \piu (X \per Z) \qquad  \annl{X} \colon \zero \per X \to \zero \]
    \[ \dr{X}{Y}{Z} \colon (X \piu Y) \per Z \to (X \per Z) \piu (Y \per Z) \qquad \annr{X} \colon X \per \zero \to \zero \]
satisfying the coherence axioms in Figure~\ref{fig:rigax}.  A rig functor is a functor that is strong symmetric monoidal for both $(\Cat{C}, \per, \uno, \symmt)$ and $(\Cat{C}, \piu, \zero, \symmp)$ and that preserves the four above natural isomorphisms.
\end{definition}



In this work, we focus on those rig categories where $\piu$ is a coproduct. 

\begin{definition}\label{def:fcrig}
A \emph{finite coproduct (fc) rig category} is a rig category $(\Cat{C}, \piu, \zero, \per, \uno)$  such that $(\Cat{C}, \piu, \zero)$ is a finite coproducts category. A \emph{morphism of fc rig categories} is both a rig functor and a morphism of fc categories. 
\end{definition}
By definition, in a finite coproduct rig category, every object $X$ has a natural monoid $\codiag{X}\colon X \piu X \to X $ and $\cobang{X}\colon \zero \to X$. The following result from \cite{bonchi2023deconstructing} illustrates the interactions of $\codiag{}$ and $\cobang{}$ with $\per$.
 \begin{lemma}\label{prop:fcrig}
	In any finite coproduct rig category the following equalities hold. 
\[
	\codiag{X\per Y}=\Idr{X}{Y}{Y};(\codiag{X}\per \id{Y})	=\Idl{X}{Y}{Y};(\id{X}\per \codiag{Y})\qquad \qquad
	\cobang{X\per Y}=\Iannl{Y};(\cobang{X}\per\id{Y})=\Iannr{X};(\id{X}\per\cobang{Y})
	\]
\end{lemma}

We expect the following result to be known but we cannot find a reference for it. 
\begin{proposition}\label{prop: Kleisli of distributive monoidal}
Let $(\Cat{C},\piu,\zero,\per, \uno)$ be a finite coproduct rig category and $T \colon \Cat{C}\to \Cat{C}$ be a symmetric monoidal monad over $(\Cat{C},\per, \uno)$. Then the Kleisli category $\Cat{C}_T$ is a finite coproduct rig category. 
\end{proposition}
The fact that $\Cat{C}_T$ has coproduct easily follows by standard arguments. The hard part of the proof, in Appendix \ref{sec:rigcategories}, consists in proving naturality of the distributors $\delta^l$ and $\delta^r$.
%
%
%
Before combining fc rig categories with cd categories in Section \ref{sec:fccdrig}, it is convenient to fix the following definition.
%
%
\begin{definition}\label{def:fcfprig}
A \emph{finite coproduct-finite product (fc-fp) rig category} is a rig category $(\Cat{C}, \piu, \zero, \per, \uno)$  such that $(\Cat{C}, \piu, \zero)$ is a finite coproduct category and $(\Cat{C}, \per, \uno)$ is a finite product category. 
\end{definition}
\begin{remark}\label{remark: distributive arrows in distriutive categories}
Note that, by definition, any fc-fp rig category has both finite coproducts (given by $\piu$) and products ($\per$). The reader may wonder whether any category with finite products and finite coprodutcs is also a fc-fp rig category. This is not the case, since the canonical morphism $(X\times Y)+ (X\times Z) \to X\times (Y + Z)$ induced by the universal properties of (co)products is not necessarily an iso. Categories with this property are known as \emph{distributive categories} \cite{carboni1993introduction}: note that the only difference with the structures in Definition \ref{def:fcfprig} consists in the fact that in fc-fp rig categories the distributors $\delta^l$ and $\delta^r$ are not forced to be the canonical ones.  
However, as shown in \cite{lacknoncanonical}, the assumption of  natural isomorphisms $\delta^l$ and $\delta^r$ implies that the canonical morphisms are isos. Hence, every fc-fp rig category is also distributive. A similar mismatch arises amongst fc rig categories and \emph{monoidal distributive categories} \cite{annalabella}.
\end{remark}

%% file: sections/cdfbrig.tex
%
%
%

\subsection{Finite Coproduct-Copy Discard Rig categories }\label{sec:fccdrig}
In Section \ref{sec:CD}, we have seen how copy discard categories generalise finite product categories. Here, we introduce finite coproduct-copy discard rig categories that, similarly, generalise finite coproduct-finite product rig categories in Definition \ref{def:fcfprig}.

 \begin{definition}\label{def: distributive gs-monoidal category}
    A \emph{finite coproduct-copy discard (fc-cd) rig category} is a rig category $(\Cat{C}, \piu, \zero, \per, \uno)$ and, for all objects $X$, $\copier{X} \colon X \to X \per X$, $\discharger{X} \colon X \to \uno$, $\codiag{X}\colon X \piu X \to X$ and $\cobang{X} \colon \zero \to X$ such that
     \begin{center}
    (a) $(\Cat{C}, \piu, \zero, \codiag{}, \cobang{})$ is a fc category, \hspace{1cm}  (b) $(\Cat{C}, \per, \uno, \copier{}, \discharger{})$ is a cd category and
    \end{center}        
     (c) monoids and comonoids interact according to the following coherence condition (see also Fig. \ref{fig:cohfccd}). 
\begin{equation}\label{equation: coherence1}
           \discharger{X\oplus Y}=(\discharger{X} \oplus \discharger{Y});\codiag{\uno} \; \;
 \copier{X\oplus Y}=(\Irunitp{X}\piu\Ilunitp{Y} ); ((\copier{X} \oplus \cobang{X\per Y}) \oplus( \cobang{Y \per X} \oplus \copier{Y}));(\Idl{X}{X}{Y}\oplus \Idl{Y}{X}{Y}); (\Idr{X}{Y}{X\oplus Y})
\end{equation}  
A morphism of fc-cd rig categories is a rig functor preserving monoids and comonoids. 
\end{definition}

\begin{figure}
\[
\begin{tikzcd}
	{X\piu Y} && {(X\piu Y) \per (X\piu Y)} \\
	 && {(X \per (X\piu Y))\piu (Y\per (X\piu Y))} \\
	{(X\piu 0 ) \piu (0 \piu Y)} && {((X\per X) \piu (X\per Y)) \piu ((Y\per X) \piu (Y \per Y))}
	\arrow["{\copier{X\piu Y}}", from=1-1, to=1-3]
	\arrow["{\Irunitp{X}\piu\Ilunitp{Y}  }"', from=1-1, to=3-1]
	\arrow["{(\Idr{X}{Y}{X\oplus Y})}"', from=2-3, to=1-3]
	\arrow["(\copier{X} \oplus \cobang{X\per Y}) \oplus( \cobang{Y \per X} \oplus \copier{Y})"',from=3-1, to=3-3]
	\arrow["{(\Idl{X}{X}{Y}\oplus \Idl{Y}{X}{Y})}"', from=3-3, to=2-3]
\end{tikzcd}
\qquad
\begin{tikzcd}
	{X \piu Y} & {\uno \piu \uno} \\
	& \uno
	\arrow["{\discharger{X}\piu \discharger{Y}}", from=1-1, to=1-2]
	\arrow["{\discharger{X\piu Y}}"', from=1-1, to=2-2]
	\arrow["{\codiag{\uno}}", from=1-2, to=2-2]
\end{tikzcd} \]
\caption{Diagrams for the coherence conditons in \eqref{equation: coherence1}.}\label{fig:cohfccd}
\end{figure}

The coherence conditions in \eqref{equation: coherence1} will be fundamental in Section \ref{sc:tapeTcd} for the diagrammatic language that we are going to introduce. Moreover, they guarantee several useful properties.

\begin{proposition}\label{prop:map}
    Let $\Cat{C}$ be a fc-cd rig category. Then, for all objects $X$, the followings hold:
    \begin{enumerate}
    \item $\codiag{X};\copier{X}= (\copier{X}\piu\copier{X});\codiag{X \per X}$,  $\quad \codiag{X};\discharger{X}=(\discharger{X} \piu \discharger{X});\codiag{\uno}$,
    $\quad \cobang{X};\copier{X}= ;\cobang{X \per X} \;$ and $\quad \cobang{X};\discharger{X}=\cobang{\uno}$; 
    \item Both $\codiag{X}$ and $\cobang{X}$ are functional and total, i.e., 
    \[\codiag{X};\copier{X}= \copier{X\piu X};(\codiag{X}\per \codiag{X})\text{,} \quad \codiag{X};\discharger{X}= \discharger{X\piu X}\text{,} \quad \cobang{X};\copier{X}= \copier{\zero};(\cobang{X}\per \cobang{X}) \quad \text{and} \quad \cobang{X};\discharger{X}= \discharger{\zero}\text{;} \]
    \item The category $\mathrm{Map}(\Cat{C})$ is a finite coproduct-finite product rig category.
    \end{enumerate}
  \end{proposition} 
Note that in the last point, $\mathrm{Map}(\Cat{C})$ is defined as for cd categories in Section \ref{sec:fccdrig}. Again, the hard part of the proof consists in showing that the distributors are functional and total.


\begin{lemma}\label{lemma:distributive}
    Every fc-fp rig category is a fc-cd rig category.
\end{lemma}

In Proposition \ref{prop: Kleisli of distributive monoidal}, we have seen that for a commutative monad $T$ over a finite coproduct rig category, its Kleisli category $\Cat{C}_T$ is again a finite coproduct rig category. By Proposition \ref{prop: Kleisli gs e cogs},  $\Cat{C}_T$ is also a cd category. It turns out that it is a fc-cd category.

\begin{proposition}\label{prop:Kleislifccd}
        Let $(\Cat{C},\piu,\zero,\per, \uno)$ be a fc-cd rig category and $T \colon \Cat{C}\to \Cat{C}$ be a symmetric monoidal monad over $(\Cat{C},\per, \uno)$. Then the Kleisli category $\Cat{C}_T$ is fc-cd rig.
\end{proposition}

Since $\Cat{Set}$ is a finite coproduct-finite product rig category, by Lemma \ref{lemma:distributive} and Proposition \ref{prop:Kleislifccd}, the Kleisli category $\Cat{Set}_T$ of any symmetric monoidal monad over $(\Cat{Set}, \times, \uno)$ is a fc-cd rig category.

%% file: sections/sigmaerig.tex
\section{Combining Rig Categories with Algebraic Theories}\label{sec:Trig}

In the previous section, we have illustrated how monoids for $\piu$ and comonoids for $\per$ interact within the Kleisli category of a symmetric monoidal monads. In Section \ref{sec:running}, we have seen that in $\Sets_{\subdistr}$, there are two more interesting arrows $\Br{+_p}_X$ and $\Br{\star}_X$ which arise from an algebraic theory presenting the monad $\subdistr$. In this section we focus on such arrows arising from algebraic theories.

\subsection{Preliminaries on Algebraic Theories}
A \emph{cartesian signature} is a couple $(\Sigma, \ar)$ where $\Sigma$ is a set of symbols and $\ar \colon \Sigma \to \mathbb{N}$ the arity function. 
An \emph{algebraic theory} is a pair $\mathbb{T} = (\Sigma, E)$, where $\Sigma$ is a cartesian signature and $E$ is a set of equations between $\Sigma$-terms, represented as pairs $(t_1, t_2)$. Hereafter we write $T_{\mathbb{T}}$ for the \emph{term monad} on $\Sets$, i.e., the monad assigning to each set $X$ the set of all $E$-equivalence classes of $\Sigma$-terms with variables in $X$, and  $\lawveret$ for the \emph{Lawvere theory} associated to $\mathbb{T}$, i.e., the strict finite product category freely generated by $(\Sigma, E)$. In this category, objects are natural numbers, and a morphism $n \to m$ is a tuple $\langle t_1, \dots, t_m \rangle$ 
of $\Sigma$-terms --modulo the equivalence relation induced by $E$-- where each term $t_i$ has variables in  $x_1, \dots, x_n$. The categorical product in $\lawveret$ is determined by addition on natural numbers. 

Particularly relevant for our exposition is $\aleph_0$, the strict finite coproduct category freely generated by a single object. In $\aleph_0$, objects are finite ordinals --we write $n$ for the ordinal $\{0,\dots , n-1\}$--  and arrows are functions. By freeness, there exists an identity on object functor $d\colon \aleph_0 \to \opcat{\lawveret}$ for any algebraic theory $\mathbb{T}$. 
Moreover, there is a morphism of finite coproduct categories $c\colon \aleph_0 \to \Sets_{T_{\mathbb{T}}}$ that is obtained by  composing the embedding $\aleph_0 \to \Sets$ with the canonical left adjoint $\Sets \to \Sets_{T_{\mathbb{T}}}$. We write $ \Sets_{T_{\mathbb{T}}}^{\aleph_0}$ for the the full subcategory of $\Sets_{T_{\mathbb{T}}}$ having as object only finite ordinals. It is crucial for our development the well known fact \cite{hylandpoweruniversalalgebra} that $\Sets_{T_{\mathbb{T}}}^{\aleph_0}$ and $\opcat{\lawveret}$ are isomorphic as finite coproduct categories. We call $i\colon \opcat{\lawveret} \to  \Sets_{T_{\mathbb{T}}}$ the morphism obtained by postcomposing the isomorphism $\opcat{\lawveret} \to \Sets_{T_{\mathbb{T}}}^{\aleph_0}$ with the embedding $\Sets_{T_{\mathbb{T}}}^{\aleph_0}\to \Sets_{T_{\mathbb{T}}}$. In summary, we are interested in the morphisms of finite coproduct categories illustrated below.
%
\begin{equation}\label{eq:Lawdiag}
\begin{tikzcd}[column sep=2em]
		\opcat{\lawveret} \ar[rr,"i"] &  &  \Sets_{T_{\mathbb{T}}}\\
		 & \aleph_0 \ar[ur,"c"']  \ar[ul,"d"]
\end{tikzcd}
\end{equation}
\begin{example}\label{ex:algebraic-theory-convex}
Consider $\mathbb{PCA}$, the algebraic theory of \emph{pointed convex algebras}: the signature $\Sigma$ consists of a binary operation $+_p$ for all $p\in(0,1)$ and a constant $\star$; the set of equations $E$ contains
\[(x_1+_q x_2)+_p x_3 =  x_1 +_{pq} (x_2+_{\frac{p(1-q)}{1-pq}} x_3) \qquad x_1+_px_2=x_2+_{1-p}x_1 \qquad  x_1+_p x_1 = x_1\]
It is well known, see e.g., \cite{swirszcz1974monadic,semadeni1973monads,jacobs2010convexity,doberkat2006eilenberg}, that the above theory provides a presentation for the monad $\subdistr$, in the sense that $T_\mathbb{PCA}$ and $\subdistr$ are isomorphic as monad. We will take the freedom to tacitly identify their Kleisli categories. 
Consider the morphism $x_1 +_p x_2 \colon 2 \to 1$ and $\star \colon 0 \to 1$ in $\Cat{L}_{\mathbb{PCA}}$. These are mapped, by $i$ into the function $i(+_p) \colon \{0\} \to \subdistr (\{0,1\})$ mapping the only element $0$ into the distribution where $0$ has probability $p$ and $1$ has probability $1-p$ and into $i(\star)\colon \{0\} \to \subdistr(\{\})$ mapping $0$ into the null subdistribution. 
\end{example}

\subsection{$\mathbb{T}$-rig categories}\label{ssec:rig}
We can now combine algebraic theories with finite coproduct rig categories (Definition \ref{def:fcrig}). 

We start by observing that in a rig category, for any natural number $n$, there exists an object $\objn{n}\defeq \PiuL[i=1][n]{\uno}$, where $n$-ary sums are defined as in Table \ref{tab:equationsonobject}.(b). Then, we show that, like for Kleisli categories,  for any  fc rig category $\Cat{C}$, there exists a canonical morphism of fc categories $c\colon \aleph_0 \to \Cat{C}$.
\begin{lemma}\label{lemma:aleph}
Let $\Cat{C}$ be a finite coproduct rig category. Then there exists a canonical morphism of finite coproduct categories $c\colon \aleph_0 \to \Cat{C}$ mapping each object $n$ of $\aleph_0$  into $\objn{n}$.
\end{lemma}

This lemma justifies the following definition.

\begin{definition}\label{def:Trig cat}
Let $\mathbb{T}$ be an algebraic theory. A \emph{$\mathbb{T}$-rig category} is a finite coproduct rig category $\Cat{C}$ together with a morphism of fc categories $i\colon \opcat{\lawveret} \to \Cat{C}$ making the following diagram commute. A morphism of $\mathbb{T}$-rig categories from $i \colon \opcat{\lawveret} \to \Cat{C}$ to $j \colon \opcat{\lawveret} \to \Cat{D}$ is a morphism of fc rig categories $f\colon \Cat{C} \to \Cat{D}$ such that $ i;f = j$. We will write simply $\Cat{C}$ whenever $i\colon \opcat{\lawveret} \to \Cat{C}$ is clear from the context.
\[\begin{tikzcd}[column sep=2em]
		\opcat{\lawveret} \ar[rr,"i"] &  &  \Cat{C}\\
		 & \aleph_0 \ar[ur,"c"']  \ar[ul,"d"]
\end{tikzcd}\]
\end{definition}
Note that commutativity of the above diagrams entails that for all $t\colon n \to 1$ in $\lawveret$, $i(t)$ is an arrow of type $1 \to \objn{n}$ in $\Cat{C}$. This gives rise to a natural transformation:
for all objects $X$,
\begin{equation}\label{eq:natf}
\Br{t}_X \defeq 
\begin{tikzcd}
	X & {1\per X} & { (\PiuL[i=1][n]{1})\per  X} & {\PiuL[i=1][n]{ X}}
    \arrow["{\Ilunitt{X}}", from=1-1, to=1-2]
	\arrow["{ i(t)\per \id{X}}", from=1-2, to=1-3]
	\arrow["{\delta^r_{n,X}}", from=1-3, to=1-4]
\end{tikzcd}
\end{equation}
where ${\delta^r_{n,X}}$ is the arrow inductively defined as $  \delta^r_{0,X}\defeq \annl{X}$ and ${\delta^r_{1,X}}\defeq \lunitt{X}$ and $\delta^r_{n+1,X} \defeq \dr{1}{\objn{n}}{X};(\lunitt{X}\piu \delta^r_{n,X})$.
In Appendix \ref{app:Trig categories}, we show that $t$ is indeed natural (Lemma \ref{lemma: f_X naturale}) and that it allows to monoidally enrich $\Cat{C}$ over the category of models of $\mathbb{T}$, a.k.a. Eilenberg-Moore algebras (Proposition \ref{prop:enrichement}). Moreover, $\Br{t}_X$ will play a crucial role in the tape diagrams in Section \ref{sec:tapediagrams}.

%% file: sections/cdsigmaerig.tex
\subsection{$\mathbb{T}$-Copy Discard Rig Categories}

With no surprise, we can combine fc-cd categories (Definition \ref{def: distributive gs-monoidal category}) with algebraic theories.

 \begin{definition}\label{def: T distributive gs-monoidal category}
 A $\mathbb{T}$-rig category  $i\colon \opcat{\lawveret} \to \Cat{C}$ is a \emph{$\mathbb{T}$-copy discard (shortly $\mathbb{T}$-cd) rig category} if $\Cat{C}$ is a fc-cd rig category.
 A morphism of $\mathbb{T}$-cd rig categories is a both a morphisms of $\mathbb{T}$ rig and fc-cd rig categories.
\end{definition}

Our main example of $\mathbb{T}$-copy discard rig category is provided by $i\colon \opcat{\lawveret} \to \Cat{Set}_{T_{\mathbb{T}}}$ in \eqref{eq:Lawdiag}. 
\begin{proposition}\label{prop:kleisliTrig}
Let $\mathbb{T}$ be an algebraic theory such that $T_{\mathbb{T}}$ is a symmetric monoidal monad. Then $\Sets_{T_{\mathbb{T}}}$ is a $\mathbb{T}$-copy discard rig category.
\end{proposition}
Since $\subdistr$ is a symmetric monoidal monad (Example \ref{ex:monoidal monads}) and $\mathbb{PCA}$ presents $\subdistr$, by Proposition \ref{prop:kleisliTrig}, it holds that the Kleisli category $\Cat{Set}_{\subdistr}$ equipped with the structure in \eqref{ex:products}, \eqref{ex:comonoids} and \eqref{ex:op}, is a $\mathbb{PCA}$-cd rig category. Beyond the examples provided by the above proposition, there are many others.

\begin{example}\label{ex:CM}
Consider for instance the algebraic theory $\mathbb{CM}$ of commutative monoids that presents the finitary multiset monad $M_f$. Any rig category where $\piu$ is a biproduct, i.e., both a product and a coproduct, is a $\mathbb{CM}$ rig category. The copy-discard structure is present not just in $\Sets_{M_f}$:  the category of vector spaces~\cite{johnson2021bimonoidal}, the one of finite dimension Hilbert spaces~\cite{Coecke2012a,harding2008orthomodularity} and the category of sets and relations~\cite{bonchi2023deconstructing,Bruni01somealgebraic} are all $\mathbb{CM}$-copy discard rig categories.  
\end{example}


%% file: sections/sigmaetapes.tex
\section{Tape Diagrams}\label{sec:tapediagrams}
In this section we introduce diagrammatic languages for $\mathbb{T}$-rig categories (Section \ref{sc:tape}) and $\mathbb{T}$-cd rig categories (Section \ref{sc:tapeTcd}). Our main results are  analogous to Theorem 2.3 in~\cite{joyal1991geometry} showing that string diagrams are arrows of a strict symmetric monoidal category freely generated by a monoidal signature. To begin, we need to recall what strict and freely generated means in the context of rig categories.

\subsection{Freely Generated Sesquistrict Rig Categories} 

In Section \ref{sec:CD}, we mentioned that the strictification theorem for monoidal categories allows us to safely forgets about the structural isomorphisms  $\assoc{X}{Y}{Z}$, $\lunit{X}$ and $\runit{X}$ and, thus, drastically simplify definitions and proofs. A strictification result also holds for rig categories (Theorem 5.4.6 in~\cite{johnson2021bimonoidal}): it states that every rig category is equivalent to a \emph{right} strict rig category, defined as follows.

\begin{definition}
A rig category is said to be \emph{right} (respectively \emph{left}) \emph{strict} when both its monoidal structures are 
strict and $\lambda^\bullet, \rho^\bullet$ and $\delta^r$ (respectively $\delta^l$) are all identity 
natural isomorphisms. A \emph{right strict rig functor} is a strict symmetric monoidal functor for both $\per$ and $\piu$ preserving $\delta^l$. 
\end{definition}
Note that only one of the two distributors is forced to be the identity within a strict rig category. The curious reader is referred to \cite[Section 4]{bonchi2023deconstructing} for a brief explanation of this asymmetry. In loc. cit., it is also explained that the above notion of strictness is somewhat unconvenient when studying  freely generated categories: consider a right strict rig category freely generated by a signature $\Gamma$ with sorts $\sort$. The objects of this category are terms generated by the grammar in Table~\ref{table:eq objects fsr} modulo the equations in the first three rows of the same table. These equivalence classes of terms do not come with a very handy form, unlike, for instance, the objects of a strict monoidal category, which are words. 

\begin{table}
	\begin{center}
		\begin{subtable}{0.52\textwidth}
            \scalebox{0.8}{
            \begin{tabular}{c}
				\toprule
				$X \; ::=\; \; A \; \mid \; \uno \; \mid \; \zero \; \mid \;  X \per X \; \mid \;  X \piu X$\\
				\midrule
				\begin{tabular}{ccc}
					$ (X \per Y) \per Z = X \per (Y \per Z)$ & $\uno \per X = X$ &  $X \per \uno = X $\\
					$(X \piu Y) \piu Z = X \piu (Y \piu Z)$ & $\zero \piu X = X$  & $X \piu \zero =X  $\\
					$(X \piu Y) \per Z = (X \per Z) \piu (Y \per Z)$ &  $\!\zero \per X = \zero$ & $X \per \zero=X $\\
				\end{tabular}\\
				$A \per (Y \piu Z) = (A \per Y) \piu (A \per Z)$\\
				\bottomrule
			\end{tabular}
            }
            \caption{}
            \label{table:eq objects fsr}
        \end{subtable}
		\begin{subtable}{0.45\textwidth}
            \scalebox{0.8}{
            \begin{tabular}{c}
				\toprule
				$n$-ary sums and products $\vphantom{\mid}$\\
				\midrule
				\makecell{
					\\[-1pt]
					$\Piu[i=1][0]{X_i} = \zero \; \Piu[i=1][1]{X_i}= X_1 \; \Piu[i=1][n+1]{X_i} = X_1 \piu (\Piu[i=1][n]{X_{i+1}} )$ \\[1em]
					$\Per[i=1][0]{X_i} = \uno \; \Per[i=1][1]{X_i}= X_1 \; \Per[i=1][n+1]{X_i} = X_1 \per (\Per[i=1][n]{X_{i+1}} )$ \\[3pt]
				} \\
				\bottomrule
			\end{tabular}
            }
            \caption{}
            \label{table:n-ary sums and prod}
        \end{subtable}
	\end{center}
	\caption{Equations for the objects of a free sesquistrict rig category}\label{tab:equationsonobject}
\end{table}

An alternative solution is proposed in \cite{bonchi2023deconstructing}: the focus is  on  freely generated rig categories that are \emph{sesquistrict}, i.e.\ right strict but only partially left strict: namely the left distributor  $\dl{X}{Y}{Z} \colon X \per (Y \piu Z) \to (X \per Y) \piu (X \per Z)$ is the identity only when $X$ is a basic sort $A\in \sort$. In terms of the equations to impose on objects, this amounts to the one in the fourth row in Table~\ref{table:eq objects fsr} for each $A\in \sort$.  By orienting from left to right \emph{all} the equations in Table~\ref{table:eq objects fsr}, one obtains a rewriting system that is confluent and terminating and, most importantly, the unique normal forms are exactly polynomials: a term $X$  is in \emph{polynomial} form if there exist $n$, $m_i$ and $A_{i,j}\in \sort$ for $i=1 \dots n$ and $j=1 \dots m_i$ such that $X=\Piu[i=1][n]{\Per[j=1][m_i]{A_{i,j}}}$ (for $n$-ary sums and products as in Table~\ref{table:n-ary sums and prod}).
We will always refer to terms in polynomial form as \emph{polynomials} and, for a polynomial like the aforementioned  $X$, we will call \emph{monomials} of $X$ the $n$ terms $\Per[j=1][m_i]{A_{i,j}}$. For instance the monomials of $(A \per B) \piu \uno$ are $A \per B$ and $1$. Note that, differently from the polynomials we are used to dealing with, here neither $\piu$ nor $\per$ is commutative so, for instance, $(A \per B) \piu \uno$ is different from both $\uno \piu (A \per B)$ and $(B \per A) \piu \uno$. Note that non-commutative polynomials are in one to one correspondence with \emph{words of words} over $\sort$, while monomials are words over $\sort$. 
\begin{notation}
Hereafter, we will denote by $A,B,C\dots$ the sorts in $\sort$, by $U,V,W \dots$ the words in $\sort^\star$ and by $P,Q,R,S \dots$ the words of words in $(\sort^\star)^\star$. Given two words $U,V\in \sort^\star$, we will write $UV$ for their concatenation and $1$ for the empty word. Given two words of words $P,Q\in (\sort^\star)^\star$, we will write $P\piu Q$ for their concatenation and $\zero$ for the empty word of words. Given a word of words $P$, we will write $\pi P$ for the corresponding term in polynomial form, for instance $\pi(A \piu BCD\piu 1 )$ is the term $A \piu ((B \per (C \per D)) \piu \uno)$. Throughout this paper  we  often omit $\pi$, thus we implicitly identify words of words with polynomials.
\end{notation}

Beyond concatenation ($\piu$), one can define a product operation $\per$ on $(\sort^\star)^\star$ by taking the unique normal form of $\pi(P) \per \pi(Q)$ for any $P,Q\in (\sort^\star)^\star$. More explicitly for
$P = \Piu[i]{U_i}$ and $Q = \Piu[j]{V_j}$, 
\begin{equation}\label{def:productPolynomials} P \per Q \defeq \Piu[i]{\Piu[j]{U_iV_j}}.
\end{equation}
Observe that, if both $P$ and $Q$ are monomials, namely, $P=U$ and $Q=V$ for some $U,V\in \sort^\star$, then $P\per Q = UV$. We can thus safely write $PQ$ in place of $P\per Q$ without the risk of any confusion.


%
%
%

\begin{definition}\label{def:sesquistrict rig category}
	A \emph{sesquistrict rig category} is a functor $H \colon \Cat S \to \Cat C$, where $\Cat S$ is a discrete category and $\Cat C$ is a right strict rig category, such that for all $A \in \Cat S$
	\[
	\dl{H(A)}{X}{Y} \colon H(A) \per (X \piu Y) \to (H(A) \per X) \piu (H(A) \per Y)
	\]
	is an identity morphism. We will also say, in this case, that $\Cat C$ is a $\Cat S$-sesquistrict rig category.
	
	Given $H \!\colon\! \Cat S \!\to\! \Cat C$ and $H' \!\colon\! \Cat S' \!\to\! \Cat C'$ two sesquistrict rig categories, a \emph{sesquistrict rig functor} from $H$ to $H'$ is a pair $(\alpha \!\colon\! \Cat S \!\to\! \Cat S', \beta \!\colon\! \Cat C \!\to\! \Cat C')$, with $\alpha$ a functor and $\beta$ a strict rig functor, such that $\alpha; H' = H; \beta$.
\end{definition}
\begin{remark}
In \cite{bonchi2023deconstructing}, it was shown that for any rig category $\Cat{C}$, one can construct its strictification $\overline{\Cat{C}}$ as in \cite{johnson2021bimonoidal} and then consider the obvious embedding from  $ob(\Cat{C})$, the discrete category of the objects of $\Cat{C}$, into $\overline{\Cat{C}}$. The embedding $ob(\Cat{C}) \to \overline{\Cat{C}}$ forms a sesquistrict category and it is equivalent (as a rig category) to the original $\Cat{C}$ \cite[Corollary 4.5]{bonchi2023deconstructing}. Hereafter, when dealing with a rig category $\Cat{C}$, we will often implicitly refer to the equivalent sesquistrict $ob(\Cat{C}) \to \overline{\Cat{C}}$.
\end{remark}

Given a set of sorts $\sort$, a \emph{rig signature} is a tuple $(\sort,\Gamma,\ar,\coar)$ where $\ar$ and $\coar$ assign to each $\gen \in \Gamma$ an arity and a coarity respectively, which are polynomials. A \emph{monoidal signature} is a rig signature where arity and coarity of any $s\in \Gamma$ are monomials. An \emph{interpretation} of a rig signature $(\sort,\Gamma,\ar,\coar)$ in a sesquistrict  rig category $H \colon \Cat M \to \Cat D$ is a pair of functions $(\alpha_{\sort} \colon \sort \to Ob(\Cat M), \alpha_\Gamma \colon \Gamma \to Ar(\Cat D))$ such that, for all $\gen \in \Gamma$, $\alpha_{\Gamma}(s)$ is an arrow having as domain and codomain $(\alpha_{\sort};H)^\sharp(\ar(s))$ and  $(\alpha_{\sort};H)^\sharp(\coar(s))$, where $-^\sharp$ stands for the expected inductive extension.

\begin{definition}\label{def:freesesqui}
Let $(\sort,\Gamma,\ar,\coar)$ (simply $\Gamma$ for short) be a rig signature. A sesquistrict  rig category $H \colon \Cat M \to \Cat D$ is said to be \emph{freely generated} by $\Gamma$ if there is an interpretation $(\alpha_S,\alpha_\Gamma)$ of $\Gamma$ in $H$ such that for every sesquistrict rig category $H' \colon \Cat M' \to \Cat D'$ and every interpretation $(\alpha_\sort' \colon \sort \to Ob(\Cat M'), \alpha_\Gamma' \colon \Gamma \to Ar(\Cat D'))$ there exists a unique sesquistrict rig functor $(\alpha \colon \Cat M \to \Cat M', \beta \colon \Cat D \to \Cat D')$ such that $\alpha_\sort ; \alpha = \alpha_\sort'$ and $\alpha_\Gamma ; \beta = \alpha_{\Gamma}'$. 
\end{definition}
%
\begin{remark}
Theorem 4.9 in \cite{bonchi2023deconstructing} guarantees that, whenever $\piu$ is forced to be a biproduct, every rig signature can be reduced to a monoidal one. 
When $\piu$ is not necessarily a biproduct, as it is the case of $\mathbb{T}$-rig categories, monoidal signatures are less expressive than rig ones. 

The diagrammatic language introduced in the next section is limited to monoidal signatures: on the one hand, this restriction is what  allows having intuitive 2 dimensional diagrams; on the other, the use of the natural transformations in \eqref{eq:natf} makes the language expressive enough to generalise  \cite{bonchi2023deconstructing}.
\end{remark}





\begin{table}[t]
    \scalebox{0.7}{
        \[
        \begin{array}{@{}c@{}}
        \toprule
        \begin{array}{c@{\qquad}c}
            \begin{array}{c@{\qquad}c@{\qquad}c@{\qquad}c@{\qquad}c}
                \id{\uno} \colon \uno \to \uno 
                &
                \id{A} \colon A \to A  
                &
                {id_\zero \colon \zero \to \zero} 
                & 
                {id_U \colon U \to U} 
                & 
                s \colon \ar(s) \to \coar(s)
                \\
                \symmt{A}{B} \colon A \per B \to B \per A
                &
                {\sigma_{U, V}^{\piu} \colon U \piu V \to V \piu U} 
                &
                {\codiag{U}\colon U \piu U \to U}
                &
                { \cobang{U} \colon \zero \to U}   
                &
                \inferrule{c \colon U \to V}{\tapeFunct{c}\colon U \to V} 
            \end{array}    
            &
            \inferrule{\ar{(f)}=n }{\textstyle{\Br{f}_U \colon U \to \bigoplus^n U}}
        \end{array}
        \\
        \begin{array}{c@{\qquad}c@{\qquad}c@{\qquad}c}
            \inferrule{c \colon U \to V \and d \colon V \to W}{c ; d \colon U \to W}
            &
            \inferrule{c \colon U_1 \to V_1 \and d \colon U_2 \to V_2}{c \per d \colon U_1 \per  U_2 \to V_1 \per  V_2}
            &
            \inferrule{\t \colon P \to Q \and \s \colon Q \to R}{\t ; \s \colon P \to R}
            &
            \inferrule{\t \colon P_1 \to Q_1 \and \s \colon P_2 \to Q_2}{\t \piu \s \colon P_1 \piu P_2 \to Q_1 \piu Q_2}       
        \end{array}
        \\
        \midrule
        \midrule
        \begin{array}{@{}l @{\quad}|@{\quad} r}
            \begin{array}{cc}
                (f;g);h=f;(g;h) & id_X;f=f=f;id_Y
                \\
                \multicolumn{2}{c}{(f_1\perG f_2) ; (g_1 \perG g_2) = (f_1;g_1) \perG (f_2;g_2)}
                \\
                id_{\unoG}\perG f = f = f \perG id_{\unoG} & (f \perG g)\, \perG h = f \perG \,(g \perG h)
                \\
                \sigma_{A, B}^{\perG}; \sigma_{B, A}^{\perG}= id_{A \perG B} & (\gen \perG id_Z) ; \sigma_{Y, Z}^{\perG} = \sigma_{X,Z}^{\perG} ; (id_Z \perG \gen)
            \end{array}
            &
            \begin{array}{cc}
                (\id{ U}\piu \codiag{ U}) ; \codiag{ U} = (\codiag{ U}\piu \id{ U}) ; \codiag{ U} & \codiag{U};\tapeFunct{c} =(\tapeFunct{c} \piu \tapeFunct{c}); \codiag{V} \quad \cobang{U};\tapeFunct{c} =\cobang{V} \\ 
                \multicolumn{2}{c}{
                (\cobang{ U}\piu \id{ U}) ; \codiag{ U}  = \id{ U} 
                \quad\;\;
                \codiag{U} ; \Br{f}_U = (\Br{f}_U \piu \Br{f}_U) ; \codiag{\Piu[][n]{U}} 
                \quad\;\;
                \cobang{U};\Br{f}_U = \cobang{\Piu[][n]U}  
                } \\
                \sigma_{ U, U}^\piu;\codiag{ U}=\codiag{ U} & \t ; \Br{f}_Q = \Br{f}_P ; \Piu[][n] \t  \\ 
                \multicolumn{2}{c}{\tape{\id{U}} = \id{U} \qquad \tape{c ; d} = \tape{c} ; \tape{d} \qquad t_U = t'_U, \text{ with } (t,t') \in E}
            \end{array}
        \end{array}
        \\
        \midrule
        \midrule
        \begin{array}{@{}l | l | l}

            \begin{array}{@{}l}
                \codiag{P} \colon P \per P \to P
                \\
                \midrule
                \codiag{\zero} \defeq \id{\zero} 
                \qquad\qquad
                \codiag{U \piu P'} \defeq (\id{U} \piu \symmp{U}{P'} \piu \id{P'}) ; (\codiag{U} \piu \codiag{P'})
                \phantom{\quad}
            \end{array}
        &
            \begin{array}{@{}l}
            \smash{\cobang{P} \colon \zero \to P}
            \\
            \midrule
            \cobang{\zero} \defeq \id{\zero} 
            \qquad\qquad
            \cobang{U \piu P'} \defeq \cobang{U} \piu \cobang{P'}
            \phantom{\qquad}
            \end{array}
        &
        \begin{array}{@{}l}
            \smash{\id{P} \colon P \to P}
            \\
            \midrule
            \id{\zero} \defeq \id{\zero} 
            \qquad
            \id{U \piu P'} \defeq \id{U} \piu \cobang{P'}
        \end{array}
            \\
            \midrule
            \begin{array}{@{}l@{\qquad\qquad}l}
                \multicolumn{2}{@{}l}{
                    \sigma^\piu_{P,Q} \colon P \piu Q \to Q \piu P
                }
                \\
                \midrule
                \sigma^\piu_{\zero,Q} \defeq \id{Q} 
                &
                \sigma^\piu_{U \piu P',Q} \defeq (\id{U} \piu \sigma^\piu_{P', Q}) ; (\sigma^\piu_{U,Q} \piu \id{P'})
                \\
                \sigma^\piu_{P, \zero} \defeq \id{P}
                &
                \sigma^\piu_{P, W \piu Q'} \defeq (\sigma^\piu_{P,W} \piu \id{Q'}) ; (\id{W} \piu \sigma^\piu_{P,Q'})
                \phantom{\quad}
                \end{array}
                &
            \multicolumn{2}{l}{
                    \begin{array}{@{}l@{\qquad\qquad}l}
                    \multicolumn{2}{@{}l}{
                        \sigma_{U,V} \colon U \per V \to V \per U
                    }
                        \\
                        \midrule
                        \sigma_{U, \uno} \defeq \id{U}
                        &
                        \sigma_{A \per U',W} \defeq (\id{A} \per \sigma_{U', W}) ; (\sigma_{A,W} \per \id{U'})
                        \\
                        \sigma{\uno, W} \defeq \id{W}
                        &
                        \sigma_{U, B \per W'} \defeq (\sigma_{U,B} \per \id{W'}) ; (\id{B} \per \sigma_{U,W'})
                        \phantom{\qquad\qquad\qquad}
                    \qquad
                \end{array}
            }
            \\
            \midrule
            \begin{array}{@{}l}
                \dl{P}{Q}{R} \colon P \per (Q\piu R)  \to (P \per Q) \piu (P\per R) \vphantom{\symmt{P}{Q}} \\
                \midrule
                \dl{\zero}{Q}{R} \defeq \id{\zero} \vphantom{\symmt{P}{\zero} \defeq \id{\zero}} \\
                \dl{U \piu P'}{Q}{R} \defeq (\id{U\per (Q \piu R)} \piu \dl{P'}{Q}{R});(\id{U\per Q} \piu \symmp{U\per R}{P'\per Q} \piu \id{P'\per R}) \vphantom{\symmt{P}{V \piu Q'} \defeq \dl{P}{V}{Q'} ; (\Piu[i]{\tapesymm{U_i}{V}} \piu \symmt{P}{Q'})} \\
            \end{array}
            &
            \begin{array}{@{}l}
                \symmt{P}{Q} \colon P\per Q \to Q \per P, \text{ with } P = \Piu[i]{U_i} \\
                \midrule
                \symmt{P}{\zero} \defeq \id{\zero} \\
                \symmt{P}{V \piu Q'} \defeq \dl{P}{V}{Q'} ; (\Piu[i]{\tapesymm{U_i}{V}} \piu \symmt{P}{Q'})
                \phantom{\quad}
            \end{array}
            &
            \begin{array}{@{}l}
                \Br{f}_P \colon P \to \bigoplus^n P \\
                \midrule
                \Br{f}_\zero \defeq \id{\zero} \\
                \Br{f}_{U \piu P'} \defeq (\Br{f}_U \piu \Br{f}_{P'}) ; \Idl{\piu^n \uno}{U}{P'} 
            \end{array}
        \end{array}
        \\
        \bottomrule
        \end{array}
        \]
    }
    \caption{Typing rules (top), axioms (middle) and syntactic sugar (bottom) for $\algt$-tapes.}
    \label{tab:tape typing and axioms}
    \label{table:def syn sugar}
\end{table}

\subsection{Tape Diagrams for $\algt$-Rig Categories}\label{sc:tape}
For an algebraic theory $\algt = (\Sigma, E)$, we want to identify a $\algt$-rig category freely generated from a monoidal signature $(\sort, \Gamma)$. Our work is an adaptation of the tape diagrams in~\cite{bonchi2023deconstructing}.

Consider terms generated by the following two-layer grammar
\begin{equation}\label{tapesGrammar}
    \begin{tabular}{rc ccccccccccccccccccccc}\setlength{\tabcolsep}{0.0pt}
        $c$  & ::= & $\id{A}$ & $\!\!\! \mid \!\!\!$ & $ \id{\uno} $ & $\!\!\! \mid \!\!\!$ & $ \gen $ & $\!\!\! \mid \!\!\!$ & $ \sigma_{A,B} $ & $\!\!\! \mid \!\!\!$ & $   c ; c   $ & $\!\!\! \mid \!\!\!$ & $  c \per c$ & \multicolumn{8}{c}{\;} \\
        $\t$ & ::= & $\id{U}$ & $\!\!\! \mid \!\!\!$ & $ \id{\zero} $ & $\!\!\! \mid \!\!\!$ & $ \tapeFunct{c} $ & $\!\!\! \mid \!\!\!$ & $ \sigma_{U,V}^{\piu} $ & $\!\!\! \mid \!\!\!$ & $   \t ; \t   $ & $\!\!\! \mid \!\!\!$ & $  \t \piu \t  $ & $\!\!\! \mid \!\!\!$  & $\cobang{U}$ & $\!\!\! \mid \!\!\!$ & $\codiag{U}$ & $\!\!\! \mid \!\!\!$ & $\Br{f}_U$
    \end{tabular}
\end{equation}  
where $A,B \in \sort$, $U,V \in \sort^\star$, $s \in \Gamma$ and $f \in \Sigma$. Terms are typed according to the rules at the top of Table~\ref{tab:tape typing and axioms}: each type is an arrow $P \to Q$ where $P,Q\in (\sort^\star)^\star$. As expected, we consider only those terms that are typeable. Constants in~\eqref{tapesGrammar} can be extended to arbitrary $P \in (\sort^\star)^\star$, according to the inductive definitions at the bottom of Table~\ref{table:def syn sugar}. 

Terms in the first layer of the grammar are called \emph{circuits}; those in the second layer are \emph{tapes}. 
Note that (a) circuits can occur inside tapes, (b) circuits have type $U\to V$, for $U,V\in \sort^\star$ and can be composed via $\per$, (c) tapes have type $P\to Q$ for $P,Q\in (\sort^\star)^\star$ and can be composed via $\piu$. 

The key insights of \cite{bonchi2023deconstructing} is that, one can define $\per$ also on tapes: for $\t_1 \colon P \to Q$, $\t_2 \colon R \to S$,  $\t_1 \per \t_2 \defeq \LW{P}{\t_2} ; \RW{S}{\t_1} $ where $\LW{P}{\cdot}$, $\RW{S}{\cdot}$ are the left and right whiskerings, defined as in Table~\ref{tab:producttape}. This stratification is what allows to draw terms in \eqref{tapesGrammar} in just two dimensions, while more direct representation of rig categories require three dimensions (see e.g. \cite{comfort2020sheet}).

The main novelty of tapes for \(\mathbb{T}\) lies in the \( \Br{f}_U \)'s, which allow us to represent the operations of the algebraic theory \(\algt\). By combining the \( \Br{f}_U \)'s with the monoid structure \((\codiag{U}, \cobang{U})\), we can represent every term of the theory. Specifically, for each $\sign$-term $t \colon n \to 1$ and object $U \in \sort^\star$ there is an associated term $\Br{t}_U \colon U \to \Piu[][n]{U}$ constructed inductively on the $\sign$-terms as follows:
\begin{center} 
$
\Br{x_i}_U \defeq \cobang{\Piu[][i-1]{U}} \piu \id{U} \piu \cobang{\Piu[][n-i]{U}}
\qquad\quad
\Br{f(t_1, \ldots, t_m)}_U \defeq \Br{f}_U ; (\Br{t_1}_U \piu \ldots \piu \Br{t_m}_U) ; \codiagg{\Piu[][n]{U}}{m}
$,
\end{center}
where $f\in \sign$ has arity $m$ and $\codiagg{P}{m} \colon \Piu[][m] P \to P$ is defined inductively as 
$\codiagg{P}{0} \defeq \cobang{P}$ and $\codiagg{P}{m+1} \defeq (\id{P} \piu \codiagg{P}{m}) ; \codiag{P}$.

Next, we impose the laws of $\algt$-rig categories on tapes: the axioms of symmetric monoidal categories for both layers; the axioms in $E$ of $\mathbb{T}$, the axioms of fc categories and those stating naturality of the $\Br{f}_U$'s for the second layer.
Formally, for each axiom $l=r$ in the middle left of Table~\ref{tab:tape typing and axioms}, we build a set $\basicR$ containing the pairs $(l,r)$  where  $\perG$ and $\unoG$ are replaced by $\piu$ and $\zero$ and the pairs $(\tape{l},\tape{r})$ where $\perG$ and $\unoG$ are replaced by $\per$ and $\uno$. Moreover, we add a pair $(l,r)$ for each axiom $l = r$ in the middle right of Table~\ref{tab:tape typing and axioms}.
Now we define $\congI$, the smallest congruence relation (w.r.t. $;$, $\piu$ and $\per$) containing $\basicR$, according to the following inductive rules. 
\begin{equation}\label{eq:congr}
        \scalebox{0.85}{$
        \centering
        \begin{array}{c}
        \begin{array}{c@{\qquad\qquad}c@{\qquad\qquad}c@{\qquad\qquad}c}
            \inferrule*[right=($\basicR$)]{\t_1 \mathbin{\basicR} \t_2}{\t_1 \mathrel{\congI} \t_2}
            &
            \inferrule*[right=($\textsc{R}$)]{-}{\t \mathrel{\congI} \t}
            &
            \inferrule*[right=($\textsc{S}$)]{\t_1 \mathrel{\congI} \t_2}{\t_2 \mathrel{\congI} \t_1}
            &    
            \inferrule*[right=($\textsc{T}$)]{\t_1 \mathrel{\congI} \t_2 \quad \t_2 \mathrel{\congI} \t_3}{\t_1 \mathrel{\congI} \t_3}
        \end{array}
        \\[10pt]
        \begin{array}{c@{\qquad}c@{\qquad}c}
            \inferrule*[right=($;$)]{\t_1 \mathrel{\congI} \t_2 \quad \s_1 \mathrel{\congI} \s_2}{\t_1;\s_1 \mathrel{\congI} \t_2;\s_2}
            &
            \inferrule*[right=($\piu$)]{\t_1 \mathrel{\congI} \t_2 \quad \s_1 \mathrel{\congI} \s_2}{\t_1\piu\s_1 \mathrel{\congI} \t_2 \piu \s_2}
            &
            \inferrule*[right=($\per $)]{\t_1 \mathrel{\congI} \t_2 \quad \s_1 \mathrel{\congI} \s_2}{\t_1\per \s_1 \mathrel{\congI} \t_2 \per \s_2}       
        \end{array}
        \end{array}
        $}
\end{equation}
Finally, we construct the category $\CatTapeT[\Gamma][\algt]$ whose objects are polynomials in $(\sort^\star)^\star$ and arrows are $\congI\text{-equivalence}$ classes of typed terms generated from~\eqref{tapesGrammar}. This construction gives rise to a $\algt$-rig category $\Br{\cdot}_1\colon \opcat{\lawveret} \to \CatTapeT[\Gamma][\algt]$. Most importantly, it is the freely generated one.

\begin{theorem}\label{thm:free signametape}
    $\CatTapeT[\Gamma][\algt]$ is the free sesquistrict $\algt$-rig generated by the monoidal signature $(\sort,\Gamma)$. 
\end{theorem}

\begin{table}[t]
    \begin{center}
    {
        \hfill {\tiny
  \[\begin{array}{c}
        \def\arraystretch{1.2}
        \begin{array}{rclrcl|rclrcl}
            \toprule
            \LW U {\id\zero} &\defeq& \id\zero& \LW U {\t_1 \piu \t_2} &\defeq& \LW U {\t_1} \piu \LW U {\t_2} &  \RW U {\id\zero} &\defeq& \id\zero & \RW U {\t_1 \piu \t_2} &\defeq& \RW U {\t_1} \piu \RW U {\t_2} \\
            \LW U {\tape{c}} &\defeq& \tape{\id U \per c} &    \LW U {\t_1 ; \t_2} &\defeq& \LW U {\t_1} ; \LW U {\t_2}&  \RW U {\tape{c}} &\defeq& \tape{c \per \id U} & \RW U {\t_1 ; \t_2} &\defeq& \RW U {\t_1} ; \RW U {\t_2} \\
            \LW U {\symmp{V}{W}} &\defeq& \symmp{UV}{UW}       &       \LW U {\Br{f}_V} &\defeq& \Br{f}_{UV} & \RW U {\symmp{V}{W}} &\defeq& \symmp{VU}{WU}& \RW U {\Br{f}_U} &\defeq& \Br{f}_{VU} \\
            \LW U {\codiag V} &\defeq& \codiag{UV} &  \LW U {\cobang V} &\defeq& \cobang{UV}& \RW U {\codiag V} &\defeq& \codiag{VU} & \RW U {\cobang V} &\defeq& \cobang{VU} \\
            \hline \hline
            \LW{\zero}{\t} &\defeq& \id{\zero} &  \LW{W\piu S'}{\t} &\defeq& \LW{W}{\t} \piu \LW{S'}{\t} & \RW{\zero}{\t} &\defeq& \id{\zero} &
            \RW{W \piu S'}{\t} &\defeq& \dl{P}{W}{S'} ; (\RW{W}{\t} \piu \RW{S'}{\t}) ; \Idl{Q}{W}{S'} \\
            \hline \hline
            \multicolumn{12}{c}{
                \t_1 \per \t_2 \defeq \LW{P}{\t_2} ; \RW{S}{\t_1}   \quad \text{ ( for }\t_1 \colon P \to Q, \t_2 \colon R \to S   \text{ )}
            }
            \\
            \bottomrule
        \end{array}
    \end{array}        
      \]
      }
      \hfill
      \caption{Inductive definition of left and right monomial whiskerings (top); inductive definition of polynomial whiskerings (center); definition of $\per$ (bottom).}\label{tab:producttape}
       }
    \end{center}
\end{table}

\begin{remark}\label{remark:Generalization}
Theorem~\ref{thm:free signametape} generalizes Theorem 5.11 in \cite{bonchi2023deconstructing} when $\algt$ is $\mathbb{CM}$, the theory of commutative monoids (Example \ref{ex:CM}).
\end{remark}

The arrows of $\CatTapeT[\Gamma][\algt]$ enjoy an intuitive 2 dimesional representation in terms of \emph{tape diagrams}. The grammar in~\eqref{tapesGrammar} becomes:
\begin{equation*}\label{tapesDiagGrammar}
    \setlength{\tabcolsep}{2pt}
    \begin{tabular}{rc c@{$\,\mid\,$}c@{$\,\mid\,$}c@{$\,\mid\,$}c@{$\,\mid\,$}c@{$\,\mid\,$}c@{$\,\mid\,$}c@{$\,\mid\,$}c@{$\,\mid\,$}c}
        $c$  & $\Coloneqq$ &  $\wire{A}$ & $ 
    \InputIfFileExists{empty.tikz}{}{\input{./tikz/empty.tikz}}
 $ & $ \Cgen{\gen}{A}{B}  $& $ \Csymm{A}{B} $ & $ 
    \InputIfFileExists{seq_compC.tikz}{}{\input{./tikz/seq_compC.tikz}}
   $ & $  
    \InputIfFileExists{par_compC.tikz}{}{\input{./tikz/par_compC.tikz}}
$ \\
        $\t$ & $\Coloneqq$ & $\Twire{U}$ & $ 
    \InputIfFileExists{empty.tikz}{}{\input{./tikz/empty.tikz}}
 $ & $ \Tcirc{c}{U}{V}  $ & $ \Tsymmp{U}{V} $ & $ 
    \InputIfFileExists{tapes/seq_comp.tikz}{}{\input{./tikz/tapes/seq_comp.tikz}}
  $ & $  
    \InputIfFileExists{tapes/par_comp.tikz}{}{\input{./tikz/tapes/par_comp.tikz}}
$ & $\Tunit{U}$  & $\Tmonoid{U}$ & $\Talgop{U}{f}$
    \end{tabular}
\end{equation*}  
The identity $\id\zero$ is rendered as the empty tape $
    \InputIfFileExists{empty.tikz}{}{\input{./tikz/empty.tikz}}
$, while $\id\uno$ is $
    \InputIfFileExists{tapes/empty.tikz}{}{\input{./tikz/tapes/empty.tikz}}
$: a tape filled with the empty circuit. 
For a monomial $U \!=\! A_1\dots A_n$, $\id U$ is depicted as a tape containing  $n$ wires labelled by $A_i$. For instance, $\id{AB}$ is rendered as $\TRwire{A}{B}$. When clear from the context, we will simply represent it as a single wire  $\Twire{U}$ with the appropriate label.
Similarly, for a polynomial $P = \Piu[i=1][n]{U_i}$, $\id{P}$ is obtained as a vertical composition of tapes, as illustrated below on the left.

\scalebox{0.82}{$ 
    \id{AB \piu \uno \piu C} = \!\!\!\begin{aligned}\begin{gathered} \TRwire{A}{B} \\[-1.8mm] \Twire{\uno} \\[-1.8mm] \Twire{C} \end{gathered}\end{aligned}
    \qquad\;\;
    \codiag{A\piu B \piu C} = \!\!
    \InputIfFileExists{tapes/examples/codiagApBpC.tikz}{}{\input{./tikz/tapes/examples/codiagApBpC.tikz}}

    \qquad\;\;
    \cobang{AB \piu B \piu C} = \!
    \InputIfFileExists{tapes/examples/cobangABpBpC.tikz}{}{\input{./tikz/tapes/examples/cobangABpBpC.tikz}}

    \qquad\;\;
    
    \InputIfFileExists{tapes/examples/algt.tikz}{}{\input{./tikz/tapes/examples/algt.tikz}}

$}

The codiagonal $\codiag{U} \colon  U \piu U \!\to\! U$ is represented as a merging of tapes, while the cobang $\cobang{U} \colon \zero \!\to\! U$ is a tape closed on its left boundary. 
Exploiting the definitions in Table~\ref{table:def syn sugar}, we can construct codiagonals and cobangs for arbitrary polynomials.
For example, $\codiag{A\piu B \piu C}$ and $\cobang{AB \piu B \piu C}$ are depicted as the second and third diagrams above. The last diagram is the tape for 
$\Br{(\star +_p x_1)+_q x_1}_A$ when the algebraic theory is fixed to be $\mathbb{PCA}$ (Example \ref{ex:algebraic-theory-convex}).

For an arbitrary $\t \colon P \to Q$ we write $\Tbox{\t}{P}{Q}$. For a $\sign$-term $t \colon n \to 1$ we write \scalebox{0.8}{$\TSterm{U}{t_U}$}.

The graphical representation embodies several axioms such as those of symmetric monoidal categories. Those axioms that are not implicit in the graphical representation are shown in Figure~\ref{fig:tapesax}. 

\begin{figure}[t]
    \mylabel{ax:symmpinv}{$\symmp$-inv}
    \mylabel{ax:symmpnat}{$\symmp$-nat}
    
    \mylabel{ax:codiagas}{$\codiag{}$-as}
    \mylabel{ax:codiagun}{$\codiag{}$-un}
    \mylabel{ax:codiagco}{$\codiag{}$-co}
    
    \mylabel{ax:codiagnat}{$\codiag{}$-nat}
    \mylabel{ax:cobangnat}{$\cobang{}$-nat}

    \mylabel{ax:fnat}{$f$-nat}

    \begin{center}
        \scalebox{0.9}{
            \begin{tabular}{r@{}c@{}l @{} r@{}c@{}l @{} r@{}c@{}l}
                        
    \InputIfFileExists{tapes/ax/symminv_left.tikz}{}{\input{./tikz/tapes/ax/symminv_left.tikz}}
 &$\stackrel{(\symmp\text{-inv})}{=}$& 
    \InputIfFileExists{tapes/ax/symminv_right.tikz}{}{\input{./tikz/tapes/ax/symminv_right.tikz}}
  
                    & 
                    
    \InputIfFileExists{tapes/ax/symmnat_left.tikz}{}{\input{./tikz/tapes/ax/symmnat_left.tikz}}
  &$\stackrel{(\symmp\text{-nat})}{=}$&  
    \InputIfFileExists{tapes/ax/symmnat_right.tikz}{}{\input{./tikz/tapes/ax/symmnat_right.tikz}}

                &
                \multicolumn{3}{c}{
                    \begin{tabular}{r@{}c@{}l}
                        
    \InputIfFileExists{cb/symm_inv_left.tikz}{}{\input{./tikz/cb/symm_inv_left.tikz}}
 & $\axeq{\sigma\text{-inv}}$ & 
    \InputIfFileExists{cb/symm_inv_right.tikz}{}{\input{./tikz/cb/symm_inv_right.tikz}}
 
                        \\
                        
    \InputIfFileExists{cb/symm_nat_left.tikz}{}{\input{./tikz/cb/symm_nat_left.tikz}}
 & $\axeq{\sigma\text{-nat}}$ & 
    \InputIfFileExists{cb/symm_nat_right.tikz}{}{\input{./tikz/cb/symm_nat_right.tikz}}

                    \end{tabular}
                }
                \\
                        
    \InputIfFileExists{tapes/whiskered_ax/monoid_assoc_left.tikz}{}{\input{./tikz/tapes/whiskered_ax/monoid_assoc_left.tikz}}
 &$\stackrel{(\codiag{}\text{-as})}{=}$& 
    \InputIfFileExists{tapes/whiskered_ax/monoid_assoc_right.tikz}{}{\input{./tikz/tapes/whiskered_ax/monoid_assoc_right.tikz}}
  
                        &
                        
    \InputIfFileExists{tapes/whiskered_ax/monoid_unit_left.tikz}{}{\input{./tikz/tapes/whiskered_ax/monoid_unit_left.tikz}}
 &$\stackrel{(\codiag{}\text{-un})}{=}$& \Twire{U}
                        &
                        
    \InputIfFileExists{tapes/whiskered_ax/monoid_comm_left.tikz}{}{\input{./tikz/tapes/whiskered_ax/monoid_comm_left.tikz}}
 &$\stackrel{(\codiag{}\text{-co})}{=}$& \Tmonoid{U}
                \\
                \multicolumn{9}{c}{
                    \begin{tabular}{r@{}c@{}l @{\qquad} r@{}c@{}l @{\qquad} r@{}c@{}l}
                        
    \InputIfFileExists{tapes/ax/monoidnat_left.tikz}{}{\input{./tikz/tapes/ax/monoidnat_left.tikz}}
 &$\stackrel{(\codiag{}\text{-nat})}{=}$& 
    \InputIfFileExists{tapes/ax/monoidnat_right.tikz}{}{\input{./tikz/tapes/ax/monoidnat_right.tikz}}

                        &
                        
    \InputIfFileExists{tapes/ax/unitnat_left.tikz}{}{\input{./tikz/tapes/ax/unitnat_left.tikz}}
  &$\stackrel{(\cobang{}\text{-nat})}{=}$& 
    \InputIfFileExists{tapes/ax/unitnat_right.tikz}{}{\input{./tikz/tapes/ax/unitnat_right.tikz}}

                        &
                        \scalebox{0.8}{
    \InputIfFileExists{tapes/ax/newax/nat/left.tikz}{}{\input{./tikz/tapes/ax/newax/nat/left.tikz}}
} & $\stackrel{(f\text{-nat})}{=} \; \; $ & \scalebox{0.8}{
    \InputIfFileExists{tapes/ax/newax/nat/right.tikz}{}{\input{./tikz/tapes/ax/newax/nat/right.tikz}}
}
                        \\
                        
    \InputIfFileExists{tapes/ax/newax/copyf_left.tikz}{}{\input{./tikz/tapes/ax/newax/copyf_left.tikz}}
 &$\stackrel{(\codiag{}\text{-nat})}{=}$& \scalebox{0.7}{
    \InputIfFileExists{tapes/ax/newax/copyf_right.tikz}{}{\input{./tikz/tapes/ax/newax/copyf_right.tikz}}
}
                        &
                        
    \InputIfFileExists{tapes/ax/newax/discf_left.tikz}{}{\input{./tikz/tapes/ax/newax/discf_left.tikz}}
 &$\stackrel{(\cobang{}\text{-nat})}{=}$& 
    \InputIfFileExists{tapes/ax/newax/discf_right.tikz}{}{\input{./tikz/tapes/ax/newax/discf_right.tikz}}
       
                        &
                        $\TSterm{U}{t}$ & $=$ & $\TSterm{U}{t'}$
                    \end{tabular}
                }
            \end{tabular}
        }
    \end{center}
    \caption{Axioms for $\algt$-tape diagrams. Here $(t,t') \in E$.}
    \label{fig:tapesax}
\end{figure}

\subsection{Tape Diagrams for $\algt$-cd Rig Categories}\label{sc:tapeTcd}


In this section, we extend $\algt$-tape diagrams to deal with $\algt$-cd rig categories. Obviously, this requires augmenting the grammar in~\eqref{tapesGrammar} with the copy $\copier{A}$ and discard $\discharger{A}$ operations, for every $A \in \sort$.
\begin{equation}
    \begin{tabular}{rc ccccccccccccccccccccc}\setlength{\tabcolsep}{0.0pt}
        $c$  & ::= & $\id{A}$ & $\!\!\! \mid \!\!\!$ & $ \id{\uno} $ & $\!\!\! \mid \!\!\!$ & $ \gen $ & $\!\!\! \mid \!\!\!$ & $ \sigma_{A,B} $ & $\!\!\! \mid \!\!\!$ & $   c ; c   $ & $\!\!\! \mid \!\!\!$ & $  c \per c$ & $\!\!\! \mid \!\!\!$  & $\discharger{A}$ & $\!\!\! \mid \!\!\!$ & $\copier{A}$ & \\
        $\t$ & ::= & $\id{U}$ & $\!\!\! \mid \!\!\!$ & $ \id{\zero} $ & $\!\!\! \mid \!\!\!$ & $ \tapeFunct{c} $ & $\!\!\! \mid \!\!\!$ & $ \sigma_{U,V}^{\piu} $ & $\!\!\! \mid \!\!\!$ & $   \t ; \t   $ & $\!\!\! \mid \!\!\!$ & $  \t \piu \t  $ & $\!\!\! \mid \!\!\!$  & $\cobang{U}$ & $\!\!\! \mid \!\!\!$ & $\codiag{U}$ & $\!\!\! \mid \!\!\!$ & $\Br{f}_U$
    \end{tabular}
\end{equation}  
Next, we construct the category $\CatTapeTCD[\Gamma][\algt]$ as in Section~\ref{sc:tape}, but extending $\basicR$ to include the axioms of cocommutative comonoids for $(\copier{A}, \discharger{A})$. Graphically, we require the following equations.
\begin{center}
    $
    \InputIfFileExists{cb/monoid_assoc_right.tikz}{}{\input{./tikz/cb/monoid_assoc_right.tikz}}
 = 
    \InputIfFileExists{cb/monoid_assoc_left.tikz}{}{\input{./tikz/cb/monoid_assoc_left.tikz}}
 
    \qquad
    
    \InputIfFileExists{cb/monoid_unit_left.tikz}{}{\input{./tikz/cb/monoid_unit_left.tikz}}
  = 
    \InputIfFileExists{cb/monoid_unit_right.tikz}{}{\input{./tikz/cb/monoid_unit_right.tikz}}
 
    \qquad
    
    \InputIfFileExists{cb/monoid_comm_left.tikz}{}{\input{./tikz/cb/monoid_comm_left.tikz}}
  = 
    \InputIfFileExists{cb/monoid_comm_right.tikz}{}{\input{./tikz/cb/monoid_comm_right.tikz}}
$.
\end{center}
%
%
Copier and discard can be defined  for every  $U \in \sort^\star$ via the coherence conditions in~\eqref{eq:comonoidcoherence}:
\[
    \begin{array}{rcl C{0.65cm} rcl}
        \copier{\uno} &\defeq& \id{\uno} && \discharger{\uno} &\defeq& \id{\uno}  \\
        \copier{A \per U'} &\defeq& (\copier{A} \per \copier{U'}) ; (\id{A} \per \sigma^\per_{A, U'} \per \id{U'}) && \discharger{A \piu U'} &\defeq& (\discharger{A} \per \discharger{U'})
    \end{array}
\]
%
However, to conclude that $\CatTapeTCD[\Gamma][\algt]$ is a $\algt$-cd rig category, we need to show that the copy-discard structure exists for every polynomial $P \in (\sort^\star)^\star$.
The conditions in~\eqref{equation: coherence1} provide an inductive definition:
\begin{equation}\label{eq:copierind}
    \begin{array}{rcl C{0.5cm} rcl}
        \copier{\zero} &\defeq& \id{\zero} && \discharger{\zero} &\defeq& \cobang{\uno}  \\
        \copier{U \piu P'} &\defeq& \Tcopier{U} \piu \cobang{UP'} \piu ((\cobang{P'U} \piu \copier{P'}) ; \Idl{P'}{U}{P'}) && \discharger{U \piu P'} &\defeq& (\Tdischarger{U} \piu \discharger{P'}) ; \codiag{\uno}
    \end{array}
\end{equation}
For instance, $\copier{A \piu B} \colon A \piu B \to (A \piu B) \per (A \piu B) = AA \piu AB \piu BA \piu BB$ and $\discharger{A \piu B} \colon A \piu B \to \uno$ are:
\begin{center}
    $\copier{A \piu B} = \!\!\!
    \InputIfFileExists{cb/examples/copierApB.tikz}{}{\input{./tikz/cb/examples/copierApB.tikz}}
 
    \qquad\qquad
    \discharger{A \piu B} = \!\!
    \InputIfFileExists{cb/examples/dischargerApB.tikz}{}{\input{./tikz/cb/examples/dischargerApB.tikz}}
$
\end{center}

\begin{theorem}\label{th:tapes-free-T-cd}
    $\CatTapeTCD[\Gamma][\algt]$ is the free sesquistrict $\algt$-cd rig generated by the monoidal signature $(\sort,\Gamma)$.
\end{theorem}
By virtue of the above theorem, any intepretation $\interpretation$ of the monoidal signature $(\sort,\Gamma)$ into a $\algt$-cd rig category $\Cat{C}$, gives rise uniquely 
to a morphism $\CBdsem{-}_\interpretation \colon \CatTapeTCD[\Gamma][\algt] \to \Cat{C}$ which represents the \emph{compositional} semantics of tapes. Its inductive definition is reported in Figure \ref{eq:SEMANTICA}.

\begin{figure}[t]
    \renewcommand{\arraystretch}{1.5}
\!\!\!\begin{tabular}{l@{\;\;\;\;}l@{\;\;\;\;}l@{\;\;\;\;}l@{\;\;\;\;}l}
$\CBdsem{s}_{\interpretation} \defeq \alpha_{\Gamma}(s) $& $\CBdsem{\, \tapeFunct{c} \,}_{\interpretation}\defeq \CBdsem{c}_{\interpretation} $ & $\CBdsem{\codiag{U}}_{\interpretation}\defeq \codiag{\alpha^\sharp_\sort(U)} $&$ \CBdsem{\cobang{U}}_{\interpretation} \defeq \cobang{\alpha^\sharp_\sort(U)}$   &  $\CBdsem{\Br{f}_{U}}_{\interpretation} \defeq j(f) \per \id{\alpha^\sharp_\sort(U)}$\\
$\CBdsem{\id{A}}_{\interpretation}\defeq \id{\alpha_\sort(A)} $&$ \CBdsem{\id{1}}_{\interpretation}\defeq \id{1} $&$ \CBdsem{\symmt{A}{B}}_{\interpretation} \defeq \symmt{\alpha_\sort(A)}{\alpha_\sort(B)}  $&$ \CBdsem{c;d}_{\interpretation} \defeq \CBdsem{c}_{\interpretation}; \CBdsem{d}_{\interpretation}  $ & $ \CBdsem{c\per d}_{\interpretation} \defeq \CBdsem{c}_{\interpretation} \per \CBdsem{d}_{\interpretation}$\\
$\CBdsem{\id{U}}_{\interpretation}\defeq \id{\alpha^\sharp_\sort(U)} $&$ \CBdsem{\id{\zero}}_{\interpretation} \defeq \id{\zero}  $&$ \CBdsem{\symmp{U}{V}}_{\interpretation}\defeq \symmp{\alpha^\sharp_\sort(U)}{\alpha^\sharp_\sort(V)} $&$ \CBdsem{\s;\t}_{\interpretation} \defeq \CBdsem{\s}_{\interpretation} ; \CBdsem{\t}_{\interpretation} $&$ \CBdsem{\s \piu \t}_{\interpretation} \defeq \CBdsem{\s}_{\interpretation} \piu \CBdsem{\t}_{\interpretation} $ \\
 & & $\CBdsem{\copier{A}} \defeq \copier{\alpha_\sort(A)}$ & $\CBdsem{\discharger{A}} \defeq \discharger{\alpha_\sort(A)}$ & 
\end{tabular}
\caption{The  morphism of $\algt$-cd rig categories $\CBdsem{-}_\interpretation \colon \CatTapeTCD[\Gamma][\algt] \to \Cat{C}$ induced by  an interpretation $\interpretation = (\alpha_\sort, \alpha_\Gamma)$ of the monoidal signature $(\sort,\Gamma)$ in $\Cat{C}$. This is the compositional semantics of tapes.}
\label{eq:SEMANTICA}
\end{figure}

\begin{example}[Probabilistic Boolean circuits]\label{example: probabilistic control}
 Consider a monoidal signature with a single sort, $\sort = \{ A\}$, and  generators $\Gamma = \{ \Andgate , \Orgate , \Notgate , \Flip{0} , \Flip{1}, \CBcocopier \}$. All generators have coarity $A$, while their arities are determined by the the number of ports on the left: for instance $\Andgate$ has arity $AA$, while $\Flip{0}$ has arity $1$. Intuitively, the first five generators represent operations and constants of boolean algebras ($\wedge$, $\vee$, $\neg$, $0$ and $1$); the last generator takes in input two signals on the left and, if they are equal, it outputs the same signal on the right. 
 
To make this formal we define an interpretation $\interpretation$ of $(\sort, \Gamma)$ into \(\Sets_{\subdistr}\).
The sort \(A\) is interpreted as the set of booleans, \(\alpha_{\sort}(A) = 2\), and the generators as the following arrows of $\Sets_{\subdistr}$.
  \[
    \begin{array}{rclrclrcl}
      \alpha_{\Gamma}(\Andgate)\colon 2 \times 2& \to & 2 & \alpha_{\Gamma}(\Orgate)\colon 2 \times 2& \to & 2 & \alpha_{\Gamma}(\Notgate)\colon  2& \to & 2 \\
      (x,y)&\mapsto&\delta_{x\wedge y} & (x,y)&\mapsto&\delta_{x\vee y} & x & \mapsto & \delta_{\neg x}\\[4pt]
      \alpha_{\Gamma}(\Flip{0}) \colon 1 & \to & 2 & \alpha_{\Gamma}(\Flip{1}) \colon 1 & \to & 2 & \alpha_{\Gamma}(\CBcocopier) \colon 2 \times 2 & \to & 2 \\
      \bullet &\mapsto&\delta_{0} & \bullet &\mapsto&\delta_{1} & (x,x) &\mapsto&\delta_{x}
    \end{array}
  \]
Observe that $\CBcocopier$ denotes an arrow in  \(\Sets_{\subdistr}\) that is not total, while all the remaining generators are total. Moreover, \emph{all} generators denote \emph{deterministic} arrows  in \(\Sets_{\subdistr}\). 

To add probabilistic behaviours, one can consider tape diagrams for $\mathbb{PCA}$-cd categories, where $\mathbb{PCA}$ is the algebraic theory in \Cref{ex:algebraic-theory-convex}. Indeed since  \(\Sets_{\subdistr}\) is a $\mathbb{PCA}$-cd category (Proposition \ref{prop:kleisliTrig}), then by Theorem \ref{th:tapes-free-T-cd}, the intepretation $\interpretation$ induces a morphism $\CBdsem{-}_\interpretation \colon \CatTapeTCD[\Gamma][\mathbb{PCA}] \to \Sets_{\subdistr}$ defined as in Figure~\ref{eq:SEMANTICA}. 
For example, $\CBdsem{-}_\interpretation$ assigns to  the following tape on the left
\begin{equation}\label{eq:exprobabilistictapes}
        
    \InputIfFileExists{tapes/examples/ANDpOR.tikz}{}{\input{./tikz/tapes/examples/ANDpOR.tikz}}
 \qquad \qquad 
    \InputIfFileExists{tapes/examples/1p0.tikz}{}{\input{./tikz/tapes/examples/1p0.tikz}}

\end{equation}
the arrow $2 \times 2 \to 2$ of \(\Sets_{\subdistr}\) which maps any pair of booleans $(x,y)$ into $x\wedge y$ with probability $p$ and $x \vee y$ with probability $(1-p)$.
The tape on the right is mapped by $\CBdsem{-}_\interpretation$ into the arrow $1 \to 2$ assigning to $\bullet$ the distribution $1 \mapsto p$, $0 \mapsto (1-p)$. By denoting this tape with $ \Flip{p}[t]$, one can easily see that there is a straightforward encoding of the diagrammatic calculus in \cite{piedeleu2025boolean} into $\CatTapeTCD[\Gamma][\mathbb{PCA}]$. 

While a more detailed comparison with \cite{piedeleu2025boolean} is left for future work, we conclude by highlighting that the use of tapes enables a natural implementation of probabilistic control, as illustrated in the leftmost diagram of \eqref{eq:exprobabilistictapes}. In \cite{piedeleu2025boolean}, control is achieved through the probabilistic multiplexer, which is represented by the tape $\Ifgate[t][A] \defeq 
    \InputIfFileExists{tapes/examples/multiplex.tikz}{}{\input{./tikz/tapes/examples/multiplex.tikz}}
$.

Intuitively, when $ \Flip{p}[t]$, $\Tcirc{c}{\!\!\!\!}{\!\!\!\!}$ and $\Tcirc{d}{\!\!\!\!}{\!\!\!\!}$ are plugged, respectively, into the first, second and third wire of the multiplexer, the output will match that of $c$ with probability \( p \) and the one of $d$ with probability \( 1 - p \).
%
The full picture is given by the term $(\;\Flip{p}[t] \;\per \Tcirc{c}{\!\!\!\!}{\!\!\!\!} \per \Tcirc{d}{\!\!\!\!}{\!\!\!\!}) ; \scalebox{0.7}{\Ifgate[t]}$ which, by definition of $\Flip{p}[t]$ in~\eqref{eq:exprobabilistictapes} and $\per$ in~Table~\ref{tab:producttape}, is the tape on the left below. 
\begin{equation*}
    
    \InputIfFileExists{tapes/examples/multiplexT.tikz}{}{\input{./tikz/tapes/examples/multiplexT.tikz}}
 \qquad \qquad 
    \InputIfFileExists{tapes/examples/pchoice.tikz}{}{\input{./tikz/tapes/examples/pchoice.tikz}}

\end{equation*}
The two diagrams above exhibit similar behavior. However, there is an important difference when  \( d \) (or similarly \( c \)) is 
    \InputIfFileExists{tapes/examples/failure.tikz}{}{\input{./tikz/tapes/examples/failure.tikz}}
\!\!\!  denoting the null subdistribution, intuitively representing failure. The circuit on the left always fails (Lemma 4.4 in \cite{piedeleu2025boolean}); instead, the output of the circuit on the right 
matches the output of $c$ (or $d$) with probability $p$ (or $1-p$). 

Since multiplying ($\per$) by a null behaviour always results in a null behaviour, we believe that issues similar to the one of the multiplexer are unavoidable when having only the monoidal product $\per$. The additional use of $\piu$ provides, in our opinion, a valuable representation of control.

%
%

\end{example}

%% file: tikz/cb/examples/copierApB.tikz
\begin{tikzpicture}
	\begin{pgfonlayer}{nodelayer}
		\node [style=label] (76) at (-2.25, 2.625) {$A$};
		\node [style=none] (77) at (-1.75, 1.875) {};
		\node [style=none] (78) at (-1.75, 3.375) {};
		\node [style=black] (80) at (-0.5, 2.625) {};
		\node [style=none] (82) at (0.225, 3.025) {};
		\node [style=none] (83) at (0.225, 2.225) {};
		\node [style=none] (84) at (-1.75, 2.625) {};
		\node [style=none] (87) at (0.75, 1.875) {};
		\node [style=none] (88) at (0.75, 3.375) {};
		\node [style=label] (146) at (1.25, 3.025) {$A$};
		\node [style=label] (147) at (1.25, 2.225) {$A$};
		\node [style=none] (164) at (0.75, 1.175) {};
		\node [style=none] (165) at (0.75, 0.575) {};
		\node [style=none] (168) at (-0.825, 1.175) {};
		\node [style=none] (169) at (-0.825, 0.575) {};
		\node [style=none] (181) at (0.75, 1.625) {};
		\node [style=none] (182) at (0, 1.625) {};
		\node [style=none] (183) at (0, 0.125) {};
		\node [style=none] (184) at (0.75, 0.125) {};
		\node [style=none] (189) at (0.75, 3.025) {};
		\node [style=none] (190) at (0.75, 2.225) {};
		\node [style=label] (191) at (-2.25, -2.625) {$B$};
		\node [style=none] (192) at (-1.75, -1.875) {};
		\node [style=none] (193) at (-1.75, -3.375) {};
		\node [style=none] (198) at (0.75, -1.875) {};
		\node [style=none] (199) at (0.75, -3.375) {};
		\node [style=label] (200) at (1.25, -3.025) {$B$};
		\node [style=label] (201) at (1.25, -2.225) {$B$};
		\node [style=none] (202) at (0.75, -1.175) {};
		\node [style=none] (203) at (0.75, -0.575) {};
		\node [style=none] (206) at (-0.825, -1.175) {};
		\node [style=none] (207) at (-0.825, -0.575) {};
		\node [style=none] (208) at (0.75, -1.625) {};
		\node [style=none] (209) at (0, -1.625) {};
		\node [style=none] (210) at (0, -0.125) {};
		\node [style=none] (211) at (0.75, -0.125) {};
		\node [style=black] (212) at (-0.5, -2.625) {};
		\node [style=none] (213) at (0.225, -2.225) {};
		\node [style=none] (214) at (0.225, -3.025) {};
		\node [style=none] (215) at (-1.75, -2.625) {};
		\node [style=none] (216) at (0.75, -2.225) {};
		\node [style=none] (217) at (0.75, -3.025) {};
		\node [style=label] (218) at (1.25, 1.275) {$A$};
		\node [style=label] (219) at (1.25, 0.475) {$B$};
		\node [style=label] (220) at (1.25, -0.475) {$B$};
		\node [style=label] (221) at (1.25, -1.275) {$A$};
	\end{pgfonlayer}
	\begin{pgfonlayer}{edgelayer}
		\draw [style=tape] (88.center)
			 to (87.center)
			 to (77.center)
			 to (78.center)
			 to cycle;
		\draw [style=tape] (184.center)
			 to (181.center)
			 to (182.center)
			 to [bend right=90, looseness=2.00] (183.center)
			 to cycle;
		\draw (84.center) to (80);
		\draw [bend left] (80) to (82.center);
		\draw [bend right=330] (83.center) to (80);
		\draw (168.center) to (164.center);
		\draw (165.center) to (169.center);
		\draw (82.center) to (189.center);
		\draw (190.center) to (83.center);
		\draw [style=tape] (199.center)
			 to (198.center)
			 to (192.center)
			 to (193.center)
			 to cycle;
		\draw [style=tape] (211.center)
			 to (208.center)
			 to (209.center)
			 to [bend left=90, looseness=2.00] (210.center)
			 to cycle;
		\draw (206.center) to (202.center);
		\draw (203.center) to (207.center);
		\draw (215.center) to (212);
		\draw [bend left] (212) to (213.center);
		\draw [bend right=330] (214.center) to (212);
		\draw (213.center) to (216.center);
		\draw (217.center) to (214.center);
	\end{pgfonlayer}
\end{tikzpicture}

%% file: tikz/cb/examples/dischargerApB.tikz
\begin{tikzpicture}
	\begin{pgfonlayer}{nodelayer}
		\node [style=none] (224) at (3.5, 0.75) {};
		\node [style=none] (225) at (3.5, -0.75) {};
		\node [style=none] (231) at (2.5, -0.75) {};
		\node [style=none] (232) at (2.5, 0.75) {};
		\node [style=label] (233) at (-1, 1.25) {$A$};
		\node [style=none] (234) at (-0.5, 2) {};
		\node [style=none] (235) at (-0.5, 0.5) {};
		\node [style=none] (236) at (0.75, 0.5) {};
		\node [style=none] (237) at (0.75, 2) {};
		\node [style=label] (238) at (-1, -1.25) {$B$};
		\node [style=none] (239) at (-0.5, -0.5) {};
		\node [style=none] (240) at (-0.5, -2) {};
		\node [style=none] (241) at (0.75, -2) {};
		\node [style=none] (242) at (0.75, -0.5) {};
		\node [style=black] (243) at (0.75, 1.25) {};
		\node [style=none] (244) at (-0.5, 1.25) {};
		\node [style=black] (245) at (0.75, -1.25) {};
		\node [style=none] (246) at (-0.5, -1.25) {};
	\end{pgfonlayer}
	\begin{pgfonlayer}{edgelayer}
		\draw [tape] (237.center)
			 to [bend left] (232.center)
			 to (224.center)
			 to (225.center)
			 to (231.center)
			 to [bend left] (241.center)
			 to (240.center)
			 to (239.center)
			 to (242.center)
			 to [bend right=90, looseness=2.00] (236.center)
			 to (235.center)
			 to (234.center)
			 to cycle;
		\draw (244.center) to (243);
		\draw (246.center) to (245);
	\end{pgfonlayer}
\end{tikzpicture}

%% file: sections/conclusion.tex
\section{Conclusions}

In this work, we have extended the notion of tape diagrams in \cite{bonchi2023deconstructing} to accommodate Kleisli categories of monoidal monads. To achieve this, we introduced fc-cd rig categories (Definition \ref{def: distributive gs-monoidal category}), which govern the interactions between (natural) $\piu$-monoids and $\per$-comonoids. We then demonstrated that these Kleisli categories naturally form fc-cd rig categories (Proposition \ref{prop:Kleislifccd}).
Furthermore, we observed that when a monoidal monad is equipped with an algebraic presentation $\mathbb{T}=(\Sigma, E)$, the operations in $\Sigma$ induce natural transformations within the Kleisli categories (Proposition \ref{prop:kleisliTrig}). These monoids, comonoids, and natural transformations enabled us to define a 2-dimensional diagrammatic language for the morphisms of $\algt$-cd rig categories freely generated by a monoidal signature $\Gamma$ (Theorem \ref{th:tapes-free-T-cd}). Collectively, our results ensure that whenever $\Gamma$ is interpreted within the Kleisli category of a monoidal monad, this interpretation can be uniquely extended to a rig functor from the category of diagrams $\CatTapeTCD[\Gamma][\algt]$. For example, by interpreting Boolean circuits in $\Sets_{\subdistr}$, we obtain a language that captures the expressiveness of \cite{piedeleu2025boolean} while naturally incorporating probabilistic control (Example \ref{example: probabilistic control}).

This work opens several promising research directions. (a) Extending tape diagrams with monoidal traces for $\piu$ could enable the modeling of iteration and, potentially, the development of program logics akin to \cite{bonchi2024diagrammatic}, but tailored for effectful programs. The necessary constraints on monads could be drawn from \cite{jacobs2010coalgebraic,haghverdi2000categorical}. (b) We aim to further explore the connections with probabilistic Boolean circuits in \cite{piedeleu2025boolean}. Notably, an alternative encoding of the calculus in \cite{piedeleu2025boolean} into $\mathbb{PCA}$-tapes maps each wire to the identity tape of $\objn{2}$. We anticipate that this may clarify links to the theory of \emph{effectus} \cite{introductioneffectus}. (c) Our current approach does not accommodate the monad of convex sets of distributions, as it is not monoidal. To incorporate probability and non-determinism, we may leverage the graded monad from \cite{DBLP:journals/pacmpl/LiellCockS25} and extend tape diagrams for graded monads similarly to the string diagrams in \cite{DBLP:journals/corr/abs-2501-18404}. (d) As discussed in Section \ref{ssec:rig}, $\mathbb{T}$-rig categories are enriched over models of $\mathbb{T}$. Investigating whether our tape diagrams provide a principled approach to enriching string diagrams --such as the one proposed in \cite{villoria2024enriching}-- represents another intriguing avenue for future research.

%% file: appendices/coherence.tex
\section{Coherence Axioms}\label{app:coherence axioms}

In this Appendix we collect together various Figures, listing the coherence axioms required by the definition of the algebraic structures we consider in the article.

 \begin{figure}[H]
     \begin{equation} \label{ax:monoidaltriangle}
         \input{tikz-cd/monoidal_triangle.tikz} \tag{M1}
     \end{equation}
     \begin{equation}\label{ax:monoidalpentagone}
         \input{tikz-cd/monoidal_pentagon.tikz} \tag{M2}
     \end{equation}
     \caption{Coherence axioms of monoidal categories}
     \label{fig:moncatax}
 \end{figure}

\begin{figure}[H]
   \begin{minipage}[t]{0.48\textwidth}
       \begin{equation}\label{eq:symmax1}
           \input{tikz-cd/symmax1.tikz} \tag{S1}
       \end{equation}
   \end{minipage}
   \hfill
   \begin{minipage}[t]{0.48\textwidth}
        \begin{equation}\label{eq:symmax2}
            \input{tikz-cd/symmax2.tikz} \tag{S2}
        \end{equation}    
   \end{minipage}
    \begin{equation}\label{eq:symmax3}
        \input{tikz-cd/symmax3.tikz} \tag{S3}
    \end{equation}
    \caption{Coherence axioms of symmetric monoidal categories}
    \label{fig:symmmoncatax}
\end{figure}

\begin{figure}[H]
    \begin{equation}
        \input{tikz-cd/comonoid_assoc.tikz}\tag{Com1}
    \end{equation}
    \begin{minipage}[t]{0.50\textwidth}
        \begin{equation}\tag{Com2}
            \input{tikz-cd/comonoid_unit.tikz}
        \end{equation}
    \end{minipage}
    \hfill
    \begin{minipage}[t]{0.46\textwidth}
        \begin{equation}\tag{Com3}
            \input{tikz-cd/comonoid_comm.tikz}
        \end{equation}
    \end{minipage}
    \caption{Cocommutative comonoid axioms}
    \label{fig:comonoidax}
\end{figure}

\begin{figure}[H]
    \begin{equation}\label{eq:coherence diag}\tag{FP1}
        \begin{tikzcd}[column sep=4.5em,baseline=(current  bounding  box.center)]
            X \perG Y \ar[r,"\copier{X \perG Y}"] \ar[d,"\copier X \perG \copier{Y}"'] & (X \perG Y) \perG (X \perG Y) \\
            (X \perG Y) \perG (Y \perG Y) \ar[d,"\assoc X X {Y \perG Y}"'] \\
            X \perG (X \perG (Y \perG Y)) \ar[d,"\id X \perG \Iassoc X Y Y"'] & X \perG (Y \perG (X \perG Y)) \ar[uu,"\Iassoc X Y {X \perG Y}"'] \\
            X \perG ((X \perG Y) \perG Y) \ar[r,"\id X \perG (\symm{X}{Y}^{\perG} \perG \id Y)"] & X \perG ((Y \perG X) \perG Y) \ar[u,"\id X \perG \assoc Y X Y"'] 
        \end{tikzcd}
    \end{equation}
    \\
	\begin{minipage}[b]{0.33\textwidth}
	 	\begin{equation}\label{eq:coherence bang}\tag{FP2}
            \begin{tikzcd}[baseline=(current  bounding  box.center)]
                X \perG Y \ar[r,"\discharger{X \perG Y}"] \ar[d,"\discharger{X} \perG \discharger{Y}"'] & \unoG \\
                \unoG \perG \unoG \ar[ur,"\lunit \unoG"']
            \end{tikzcd}
	    \end{equation} 
	\end{minipage}
	\hfill
	\begin{minipage}[b]{0.26\textwidth}
		\begin{equation}\tag{FP3}
		\begin{tikzcd}
		I \ar[r,shift left=2,"\copier{I}"] \ar[r,shift right=2,"\Ilunit I"'] & I \perG I
		\end{tikzcd}
		\end{equation}
		\end{minipage}
	\hfill
	\begin{minipage}[b]{0.26\textwidth}
		\begin{equation}\label{eq:bang I = id I}\tag{FP4}
            \begin{tikzcd}
            I \ar[r,shift left=2,"\discharger{I}"] \ar[r,shift right=2,"\id I"'] & I
            \end{tikzcd}
		\end{equation}
	\end{minipage}
 \caption{Coherence axioms for cocommutative comonoid}
\label{fig:fpcoherence}
 \end{figure}

 \begin{figure}[H]
    \begin{equation}\tag{Mon1}
        \input{tikz-cd/monoid_assoc.tikz}
    \end{equation}
    \begin{minipage}[t]{0.50\textwidth}
        \begin{equation}\label{ax:Mon2}\tag{Mon2}
            \input{tikz-cd/monoid_unit.tikz}
        \end{equation}
    \end{minipage}
    \hfill
    \begin{minipage}[t]{0.46\textwidth}
        \begin{equation}\tag{Mon3}
            \input{tikz-cd/monoid_comm.tikz}
        \end{equation}
    \end{minipage}
    \caption{Commutative monoid axioms}
    \label{fig:monoidax}
\end{figure}

 \begin{figure}[H]		
    \begin{equation}\label{eq:coherence codiag}\tag{FC1}
        \begin{tikzcd}[column sep=4.5em,baseline=(current  bounding  box.center)]
        (X \perG Y) \perG (X \perG Y) \ar[dd,"\assoc X Y {X \perG Y}"'] \ar[r,"\codiag{X \perG Y}"] & X \perG Y \\
        & (X \perG X) \perG (Y \perG Y) \ar[u,"\codiag X \perG \codiag Y"']\\
            X \perG (Y \perG (X \perG Y)) \ar[d,"\id X \perG \Iassoc Y X Y"'] & X \perG (X \perG (Y \perG Y)) \ar[u,"\Iassoc X X {Y \perG Y}"']  \\
        X \perG ((Y \perG X) \perG Y) \ar[r,"\id X \perG ( \symm{Y}{X}^{\perG} \perG \id Y)"] & X \perG (( X \perG Y) \perG Y) \ar[u,"\id X \perG \assoc X Y Y"']
        \end{tikzcd}
    \end{equation}
    \\
\begin{minipage}[b]{0.33\textwidth}
	\begin{equation}\label{eq:coherence cobang}\tag{FC2}
        \begin{tikzcd}[baseline=(current  bounding  box.center)]
        \unoG \ar[r,"\cobang{X \perG Y}"] \ar[d,"\Ilunit \unoG"']  & X \perG Y \\
        \unoG \perG \unoG \ar[ur,"\cobang X \perG \cobang Y"']  
        \end{tikzcd}
	\end{equation}
\end{minipage}
\hfill
\begin{minipage}[b]{0.26\textwidth}
	\begin{equation}\tag{FC3}
        \begin{tikzcd}
        I \perG I \ar[r,shift left=2,"\codiag I"] \ar[r,shift right=2,"\lunit I"'] &  I
        \end{tikzcd}
	\end{equation}
\end{minipage}
\hfill
\begin{minipage}[b]{0.26\textwidth}
	\begin{equation}\label{eq:cobang I = id I}\tag{FC4}
        \begin{tikzcd}
        I \ar[r,shift left=2,"\cobang I"] \ar[r,shift right=2,"\id I"'] & I
        \end{tikzcd}
	\end{equation}
\end{minipage}
\caption{Coherence axioms for commutative monoids}
\label{fig:fccoherence}
\end{figure}


\begin{figure}[H]
    \begin{tabular}{c c}
    \begin{tikzcd}[ampersand replacement = \&]
      {T(X) \perG T(Y)} \arrow{r}{m_{X,Y}} \arrow{d}[swap]{\sigma_{T(X), T(Y)}} \& {T(X \perG Y)} \arrow{d}{T(\sigma_{X,Y})} \\
      {T(Y) \perG T(X)} \arrow{r}[swap]{m_{Y,X}} \& {T(Y \perG X)}
    \end{tikzcd} &
    \begin{tikzcd}[ampersand replacement = \&]
      {X \perG Y} \arrow{dr}[swap]{\eta_{X \perG Y}} \arrow{r}{\eta_{X} \perG \eta_{Y}} \& {T(X) \perG T(Y)} \arrow{d}{m_{X,Y}} \\
      \& {T(X \perG Y)}
    \end{tikzcd}\\ \text{($T$-Lax1)} & \text{($\eta$-mon.tr.)}\\ 
  
  \addlinespace[2em]
    \begin{tikzcd}[ampersand replacement = \&]
      {TT(X) \perG TT(Y)} \arrow{r}{\mu_{X} \perG \mu_{Y}} \arrow{d}[swap]{m_{T(X),T(Y)}} \& {T(X) \perG T(Y)} \arrow{dd}{m_{X,Y}} \\
      {T(T(X) \perG T(Y))} \arrow{d}[swap]{T(m_{X,Y})} \& \\
      {TT(X \perG Y)} \arrow{r}[swap]{\mu_{X \perG Y}} \& {T(X \perG Y)}
    \end{tikzcd} &
    \begin{tikzcd}[ampersand replacement = \&]
      {T(X) \perG (T(Y) \perG T(Z))} \arrow{d}[swap]{\id{T(X)} \perG m_{Y,Z}} \arrow{r}{\alpha_{T(X), T(Y), T(Z)}} \& {(T(X) \perG T(Y)) \perG T(Z)} \arrow{d}{m_{X,Y} \perG \id{T(Z)}} \\
      {T(X) \perG T(Y \perG Z)} \arrow{d}[swap]{m_{X, Y \perG Z}} \& {T(X \perG Y) \perG T(Z)} \arrow{d}{m_{X \perG Y, Z}} \\
      {T(X \perG (Y \perG Z))} \arrow{r}[swap]{T(\alpha_{X,Y,Z})} \& {T((X \perG Y) \perG Z)}
    \end{tikzcd}\\ \text{($\mu$-mon.tr.)} & \text{($T$-Lax2)}\\
    \addlinespace[2em]
    \begin{tikzcd}[ampersand replacement = \&]
      {T(X)} \arrow{d}[swap]{\lambda_{T(X)}} \arrow{r}{T(\lambda_{X})} \& {T(I \perG X)}\\
      {I \perG T(X)} \arrow{r}[swap]{\eta_{I} \perG \id{T(X)}} \& {T(I) \perG T(X)} \arrow{u}[swap]{m_{I,X}}
    \end{tikzcd} &
    \begin{tikzcd}[ampersand replacement = \&]
      {T(X)} \arrow{d}[swap]{\rho_{T(X)}} \arrow{r}{T(\rho_{X})} \& {T(X \perG I)}\\
      {T(X) \perG I} \arrow{r}[swap]{\id{T(X)} \perG \eta_{I}} \& {T(X) \perG T(I)} \arrow{u}[swap]{m_{X,I}}
    \end{tikzcd}\\ \text{($T$-Lax3)} & \text{($T$-Lax4)}
    \end{tabular}
    \caption{Axioms for symmetric monoidal monads.}\label{fig:monoidalmonads}
  \end{figure}

\begin{figure}[H]
    \begin{minipage}[t]{0.45\textwidth}
        \begin{equation}
            \label{eq:rigax1}\tag{R1}
            \scalebox{0.9}{$\input{tikz-cd/rigax1.tikz}$}
        \end{equation}    
    \end{minipage}
    \hfill
    \begin{minipage}[t]{0.45\textwidth}
        \begin{equation}
            \label{eq:rigax2}\tag{R2}
            \scalebox{0.9}{$\input{tikz-cd/rigax2.tikz}$}
        \end{equation}
    \end{minipage}
    
    \begin{equation}
        \label{eq:rigax3}\tag{R3}
        \scalebox{0.9}{$\input{tikz-cd/rigax3.tikz}$}
    \end{equation}
    \begin{equation}
        \label{eq:rigax4}\tag{R4}
        \scalebox{0.9}{$\input{tikz-cd/rigax4.tikz}$}
    \end{equation}
    \begin{equation}
        \label{eq:rigax5}\tag{R5}
        \scalebox{0.9}{$\input{tikz-cd/rigax5.tikz}$}
    \end{equation}

    \begin{minipage}[t]{0.25\textwidth}
        \begin{equation}
            \label{eq:rigax6}\tag{R6}
            \scalebox{0.9}{$\input{tikz-cd/rigax6.tikz}$}
        \end{equation}
    \end{minipage}
    \hfill
    \begin{minipage}[t]{0.45\textwidth}
        \begin{equation}
            \label{eq:rigax7}\tag{R7}
            \scalebox{0.9}{$\input{tikz-cd/rigax7.tikz}$}
        \end{equation}
    \end{minipage}
    \hfill
    \begin{minipage}[t]{0.25\textwidth}
        \begin{equation}
            \label{eq:rigax8}\tag{R8}
            \scalebox{0.9}{$\input{tikz-cd/rigax8.tikz}$}
        \end{equation}
    \end{minipage}
    \\
    \begin{minipage}[t]{0.48\textwidth}
        \begin{equation}
            \label{eq:rigax9}\tag{R9}
            \scalebox{0.9}{$\input{tikz-cd/rigax9.tikz}$}
        \end{equation}    
    \end{minipage}
    \hfill
    \begin{minipage}[t]{0.48\textwidth}
        \begin{equation}
            \label{eq:rigax10}\tag{R10}
            \scalebox{0.9}{$\input{tikz-cd/rigax10.tikz}$}
        \end{equation}
    \end{minipage}
    \\
    \begin{minipage}[t]{0.50\textwidth}
        \begin{equation}
            \label{eq:rigax11}\tag{R11}
            \scalebox{0.9}{$\input{tikz-cd/rigax11.tikz}$}
        \end{equation}    
    \end{minipage}
    \hfill
    \begin{minipage}[t]{0.46\textwidth}
        \begin{equation}
            \label{eq:rigax12}\tag{R12}
            \scalebox{0.9}{$\input{tikz-cd/rigax12.tikz}$}
        \end{equation}
    \end{minipage}
    \caption{Coherence Axioms of symmetric rig categories}
    \label{fig:rigax}
\end{figure}

\begin{figure}[H]
    \begin{equation}
        \label{eq:dl1}
        \input{tikz-cd/dl1.tikz}
    \end{equation}
    \begin{equation}
        \label{eq:dl2}
        \input{tikz-cd/dl2.tikz}
    \end{equation}
    \begin{equation}
        \label{eq:dl3}
        \input{tikz-cd/dl3.tikz}
    \end{equation}
    \begin{minipage}[t]{0.44\textwidth}
    \begin{equation}
        \label{eq:dl4}
        \input{tikz-cd/dl4.tikz}
    \end{equation}
    \end{minipage}
    \hfill
    \begin{minipage}[t]{0.54\textwidth}
    \begin{equation}
        \label{eq:dl5}
        \input{tikz-cd/dl5.tikz}
    \end{equation}
    \end{minipage}
    \\
    \begin{minipage}[t]{0.48\textwidth}
        \begin{equation}\label{eq:dl7}
            \input{tikz-cd/dl7.tikz}
        \end{equation}
        \end{minipage}
        \hfill
        \begin{minipage}[t]{0.50\textwidth}
        \begin{equation}
            \label{eq:dl8}
            \input{tikz-cd/dl8.tikz}
        \end{equation}
    \end{minipage}
    \begin{equation}
        \label{eq:dl6}
        \input{tikz-cd/dl6.tikz}
    \end{equation}
    \caption{Derived laws of symmetric rig categories}
    \label{fig:dlaw}
\end{figure}

%% file: sections/monoidal.tex
\section{Monoidal Categories and String Diagrams}\label{sec:monoidal}
In this Appendix we recall the key aspects of the correspondence between monoidal categories and string diagrams \cite{joyal1991geometry,selinger2010survey}. Our starting point is regarding string diagrams as terms of a typed language. Given a set $\sort$ of basic \emph{sorts}, hereafter denoted by $A,B\dots$, types are elements of $\sort^\star$, i.e.\ words over $\sort$. Terms are defined by the following context free grammar
\begin{equation}\label{eq:syntaxsymmetricstrict}
\begin{array}{rcl}
f & ::=& \; \id{A} \; \mid \; \id{I} \; \mid \; \gen \; \mid \; \sigma_{A,B}^{\perG} \; \mid \;   f ; f   \; \mid \;  f \perG f \\
\end{array}
\end{equation}  
where $s$ belongs to a fixed set $\sign$ of \emph{generators} and $I$ is the empty word. Each $s\in \sign$ comes with two types: arity $\ar(s)$ and coarity $\coar(s)$. The tuple $(\sort, \sign, \ar, \coar)$, $\sign$ for short, forms a \emph{monoidal signature}. Amongst the terms generated by  \eqref{eq:syntaxsymmetricstrict}, we consider only those that can be typed according to the inference rules in Table \ref{fig:freestricmmoncatax}. String diagrams are such terms modulo the axioms in Table \ref{fig:freestricmmoncatax} where,  for any $X,Y\in \sort^\star$,  $\id{X}$ and $\sigma_{X,Y}^{\perG}$ can be easily built using $\id I$, $\id{A}$, $\sigma_{A,B}^{\perG}$, $\perG$ and $;$ (see e.g.~\cite{ZanasiThesis}).
\begin{table}[h]
    \centering
    \scalebox{0.8}{
    \begin{tabular}{c c}
        \begin{tabular}{c}
            \toprule
            Objects ($A\in \sort$) \\
            \midrule
            $X \; ::=\; \; A \; \mid \; \unoG \; \mid \;  X \perG X \vphantom{\sigma_{A,B}^{\perG}}$ \\
            \midrule
            \makecell{
                \\[-2pt] $(X\perG Y)\perG Z=X \perG (Y \perG Z)$ \\ 
                $X \perG \unoG = X $ \\ 
                $\unoG \perG X = X$ \\[2pt]
            } \\[13pt]
            \bottomrule
        \end{tabular}
        &
        \begin{tabular}{cc}
            \toprule
            \multicolumn{2}{c}{Arrows ($A\in \sort$, $s\in \sign$)} \\
            \midrule
            \multicolumn{2}{c}{$f \; ::=\; \; \id{A} \; \mid \; \id{\unoG} \; \mid \; \gen  \; \mid \;   f ; f   \; \mid \;  f \perG f \; \mid \; \sigma_{A,B}^{\perG}$} \\
            \midrule
            $(f;g);h=f;(g;h)$ & $id_X;f=f=f;id_Y$\\
            \multicolumn{2}{c}{$(f_1\perG f_2) ; (g_1 \perG g_2) = (f_1;g_1) \perG (f_2;g_2)$} \\
            $id_{\unoG}\perG f = f = f \perG id_{\unoG}$ & $(f \perG g)\, \perG h = f \perG \,(g \perG h)$ \\
            $\sigma_{A, B}^{\perG}; \sigma_{B, A}^{\perG}= id_{A \perG B}$ & $(\gen \perG id_Z) ; \sigma_{Y, Z}^{\perG} = \sigma_{X,Z}^{\perG} ; (id_Z \perG \gen)$ \\
            \bottomrule
        \end{tabular}
    \end{tabular}
    } 

    \vspace{2em} 

    \scalebox{0.75}{
    \begin{tabular}{c}
        \toprule
        Typing rules \\
        \midrule
        $
        {id_A \colon A \!\to\! A} \qquad  {id_\unoG \colon \unoG \!\to\! \unoG} \qquad {\sigma_{A, B}^{\perG} \colon A \perG B \!\to\! B \perG A} \qquad 
            \inferrule{\gen \colon \ar(s) \!\to\! \coar(s) \in \sign}{\gen \colon \ar(s) \!\to\! \coar(s)} \qquad
            \inferrule{f \colon X_1 \!\to\! Y_1 \and g \colon X_2 \!\to\! Y_2}{f \perG g \colon X_1 \perG X_2 \!\to\! Y_1 \perG Y_2}  \qquad
            \inferrule{f \colon X \!\to\! Y \and g \colon Y \!\to\! Z}{f ; g \colon X \!\to\! Z}
        $\\    
        \bottomrule  
    \end{tabular}
    } 

    \caption{Axioms for $\CatString$}
    \label{fig:freestricmmoncatax}
\end{table}

String diagrams enjoy an elegant graphical representation: a generator $\gen$ in $\sign$ with arity $X$ and coarity $Y$ is depicted as  a \emph{box} having \emph{labelled wires} on the left and on the right representing, respectively, the words $X$ and $Y$. For instance $\gen \colon AB \to C$ in $\sign$ is depicted as the leftmost diagram below. Moreover, $\id{A}$ is displayed as one wire,  $id_{\unoG} $ as the empty diagram and $\sigma_{A,B}^{\perG}$ as a crossing:
\begin{center}
    $
    \InputIfFileExists{generator.tikz}{}{\input{./tikz/generator.tikz}}
 \qquad \qquad 
    \InputIfFileExists{id.tikz}{}{\input{./tikz/id.tikz}}
 \qquad  \qquad     
    \InputIfFileExists{empty.tikz}{}{\input{./tikz/empty.tikz}}
 \qquad  \qquad   
    \InputIfFileExists{symm.tikz}{}{\input{./tikz/symm.tikz}}
$
\end{center}
Finally, composition $f;g$ is represented by connecting the right wires 
of $f$ with the left wires of $g$ when their labels match, 
while the monoidal product $f \perG g$ is depicted by stacking the corresponding 
diagrams on top of each other: 
\begin{center}
    $
    \InputIfFileExists{seq_comp.tikz}{}{\input{./tikz/seq_comp.tikz}}
 \qquad \qquad \qquad  
    \InputIfFileExists{par_comp.tikz}{}{\input{./tikz/par_comp.tikz}}
$
\end{center}
The first three rows of axioms for arrows in Table~\ref{fig:freestricmmoncatax}
are implicit in the 
graphical representation while the axioms in the last row  are displayed as 
\begin{center}
    $
    \InputIfFileExists{stringdiag_ax1_left.tikz}{}{\input{./tikz/stringdiag_ax1_left.tikz}}
 = 
    \InputIfFileExists{stringdiag_ax1_right.tikz}{}{\input{./tikz/stringdiag_ax1_right.tikz}}
 \quad\qquad 
    \InputIfFileExists{stringdiag_ax2_left.tikz}{}{\input{./tikz/stringdiag_ax2_left.tikz}}
 = 
    \InputIfFileExists{stringdiag_ax2_right.tikz}{}{\input{./tikz/stringdiag_ax2_right.tikz}}
$
\end{center}

Hereafter, we call  $\CatString$ the category having as objects words in $\sort^\star$ and as arrows string diagrams. Theorem 2.3 in~\cite{joyal1991geometry} states that  $\Cat{C}_\sign$ is a \emph{symmetric strict monoidal category freely generated} by $\sign$.
\begin{definition} A \emph{symmetric monoidal category}  consists 
of a category $\Cat{C}$, a functor $\perG \colon \Cat{C} \times \Cat{C} \to \Cat{C}$,
an object $\unoG$  and natural isomorphisms \[ \alpha_{X, Y, Z} \colon (X \perG Y) \perG Z \to X \perG (Y \perG Z) \qquad \lambda_X \colon \unoG \perG X \to X \qquad \rho_X \colon X \perG \unoG \to X \qquad \sigma_{X, Y}^{\perG} \colon X \perG Y \to Y \perG X \]
satisfying some coherence axioms (in Figures~\ref{fig:moncatax} and~\ref{fig:symmmoncatax}).
A monoidal category is said to be \emph{strict} when $\alpha$, $\lambda$ and $\rho$ are all identity natural isomorphisms. A \emph{strict symmetric monoidal functor} is a functor $F \colon \Cat{C} \to \Cat{D}$  preserving $\perG$, $\unoG$ and $\sigma^{\perG}$. We write $\SMC$ for the category of ssm categories and functors. 

\end{definition}
\begin{remark}\label{rmk:symstrict}
In \emph{strict}  symmetric monoidal (ssm) categories  the symmetry $\sigma$ is not forced to be the identity, since this would raise some problems: for instance, $(f_1;g_1) \perG (f_2;g_2) = (f_1;g_2) \perG (f_2;g_1)$ for all $f_1,f_2\colon A \to B$ and $g_1,g_2\colon B \to C$. 
This fact will make the issue of strictness for rig categories rather subtle.
\end{remark}

To illustrate in which sense $\Cat{C}_\sign$ is freely generated, it is convenient to introduce \emph{interpretations} in a fashion similar to~\cite{selinger2010survey}: an interpretation $\interpretation$ of $\sign$ into an ssm category $\Cat{D}$ consists of two functions $\alpha_{\sort} \colon \sort \to Ob(\Cat{D})$ and $\alpha_{\sign}\colon \sign \to Ar(\Cat{D})$ such that, for all $s\in \sign$, $\alpha_{\sign}(s)$ is an arrow having as domain $\alpha_{\sort}^\sharp(\ar(s))$ and codomain $\alpha_{\sort}^\sharp(\coar(s))$, for $\alpha_{\sort}^\sharp\colon \sort^\star \to Ob(\Cat{D})$ the inductive extension of  $\alpha_{\sort}$.  $\CatString$ is freely generated by $\sign$ in the sense that, for all symmetric strict monoidal categories $\Cat{D}$ and  all interpretations $\interpretation$ of $\sign$ in $\Cat{D}$, there exists a unique ssm-functor $\dsem{-}_{\interpretation}\colon \Cat{C}_\sign \to \Cat{D}$ extending $\interpretation$ (i.e. $\dsem{s}_\interpretation=\alpha_{\sign}(s)$ for all $s\in \sign$).

One can easily extend the notion of interpretation of $\sign$ into a symmetric monoidal category $\Cat D$ that is not necessarily strict. In this case we set $\alpha^\sharp_\sort \colon \sort^\star \to Ob(\Cat D)$ to be the \emph{right bracketing} of the inductive extension of $\alpha_\sort$. For instance, $\alpha^\sharp_\sort(ABC) = \alpha_\sort(A) \perG (\alpha_\sort(B) \perG \alpha_\sort(C))$.


%
%
%
%
%
%
\subsection{Finite Product and Finite Coproduct Categories}\label{ssec:fbcat}
%

Within a symmetric monoidal category it is possible to define certain algebraic structures that can characterise $\perG$ as certain (co)limits: Fox's theorem \cite{fox1976coalgebras} states that if every object of a symmetric monoidal category is equipped with a natural and commutative comonoid structure, then $\perG$ is the categorical product $\times$.

\begin{remark}\label{rem:productcomonoids}
	\begin{enumerate}[leftmargin=*]
\item For any two objects $X_1$, $X_2$ of a  fp-category $\Cat C$ (see Definition \ref{def:fp}), $X_1 \per X_2$ is the categorical product $X_1 \times X_2$: the projections $\pi_1\colon X_1 \per X_2 \to X_1$ and $\pi_2\colon X_1 \ X_2 \to X_2$ are  
\[\begin{tikzcd}[column sep=3em]
X_1 \per X_2 \ar[r,"\id{X_1} \per \, \discharger{X_2}"] &  X_1 \per \unoG \ar[r,"\runit{X_1}"]  & X_1
	\end{tikzcd}
\quad \text{and} \quad
\begin{tikzcd}[column sep=3em]
X_1 \per X_2 \ar[r," \discharger{X_1} \per \id{X_2} "] &  \unoG \per X_2 \ar[r,"\lunit{X_2}"]  & X_2.
	\end{tikzcd}	
\]
The unit $\unoG$ is the terminal object and $\discharger{X}$ the unique morphism of type $X \to \unoG$. 
For $f_1\colon Y\to X_1$, $f_2 \colon Y \to X_2$, their pairing $\pairing{f_1,f_2}\colon Y \to X_1 \per X_2$ is given by 
\[\begin{tikzcd}
Y  \ar[r,"\copier{Y}"] & Y \per Y \ar[r,"f_1 \per f_2"] &  X_1 \per X_2.
	\end{tikzcd}\]
	
\item In a fc category, $X_1 \piu X_2$ is instead a categorical coproduct, with injections $\iota_i \colon X_i \to X_1 \piu X_2$ and copairing $\copairing{f_1,f_2} \colon X_1 \piu X_2 \to Y$, for $f_1 \colon X_1 \to Y$ and $f_2 \colon X_2 \to Y$, in the dual way. Also dually $\unoG$ is the initial object 0 and $\cobang X$ the unique morphism of type $0 \to X$.
\end{enumerate}
\end{remark}

Like for (symmetric) monoidal categories, we are interested in freely generated fp (fc) categories. We will illustrate first the case for fc categories.

Like for symmetries, it is enough to add as generators  $\cobang{A}$, $\codiag{A}$ for all $A \in \sort$ and define $\cobang{X}$, $\codiag{X}$ for all $X \in \sort^*$ inductively as follows:

\begin{equation}\label{eq:codiagind}\begin{array}{rcl|rcl}
\cobang{\unoG} &\defeq& \id{\unoG} & \codiag{\unoG} &\defeq& \id{\unoG} \\
\cobang{A \perG W} &\defeq& \cobang{A} \perG \cobang{W} & \codiag{A \perG W}&\defeq& (\id{A} \perG \sigma_{A,W} \perG \id{W}) ; (\codiag{A}\perG \codiag{W})\\
\end{array}
\end{equation}
With these definitions the coherence axioms in Figures~\ref{fig:fccoherence} are automatically satisfied. In the strict setting also the monoid and comonoid axioms get simplified, as illustrated in the second row of Table~\ref{fig:freestrictfbcat}. Like for symmetries, to obtain naturality of  $\cobang{X}$, $\codiag{X}$ for arbitrary arrows, it is enough to impose naturality just with respect to the generators in $\Sigma$. 

\begin{table}
\begin{center}
\begin{tabular}{c}
$\copier{ A}; (\id{ A}\perG \copier{ A}) = \copier{ A};(\copier{ A}\perG \id{ A}) \qquad \copier{ A} ; (\discharger{ A}\perG \id{ A}) = \id{ A} \qquad \copier{ A};\sigma_{ A, A}=\copier{ A}$ \\
$(\id{ A}\perG \codiag{ A}) ; \codiag{ A} = (\codiag{ A}\perG \id{ A}) ; \codiag{ A} \qquad (\cobang{ A}\perG \id{ A}) ; \codiag{ A}  = \id{ A} \qquad \sigma_{ A, A};\codiag{ A}=\codiag{ A}$ \\


$\gen; \discharger{Y}=\discharger{X} \qquad \gen; \copier{Y}=\copier{X}; (\gen \perG \gen)  $\\
$\cobang{X};\gen =\cobang{Y}\qquad
\codiag{X};\gen =(\gen \perG \gen); \codiag{Y}$ \\
\end{tabular}
\end{center}
\caption{Additional axioms for freely generated strict fp and fc categories.}\label{fig:freestrictfbcat}
\end{table}

In a nutshell, the free strict fc category generated by $\sign$ is defined as follows. Objects are elements of $\sort^\star$; 
   arrows are the $\sign$-terms inductively generated as 
   \begin{equation}
\begin{array}{rcl}
f & ::=& \; \id{A} \; \mid \; \id{I} \; \mid \; \gen  \; \mid \;   f ; f   \; \mid \;  f \perG f \; \mid \; \sigma_{A,B}^{\perG} \; \mid \; \cobang{A}\; \mid \; \codiag{A}\\
\end{array}
\end{equation}  
modulo the axioms in   Tables~\ref{fig:freestricmmoncatax} and those in the second and fourth rows of Figure~\ref{fig:freestrictfbcat}. 

The free strict fp category generated by $\Sigma$ is obtained similarly, considering adding as generators $\copier{A}$ and $\discharger{A}$ for all $A\in \sort$ and defining $\copier{X}$, $\discharger{X}$ for all $X \in \sort^*$ inductively as follows:
\begin{equation}\label{eq:diagind}\begin{array}{rcl|rcl}
\discharger{\unoG} &\defeq& \id{\unoG} & \copier{\unoG} &\defeq& \id{\unoG}  \\
\discharger{A \perG W} &\defeq& \discharger{A} \perG \discharger{W} & \copier{A \perG W}&\defeq& (\copier{A}\perG \copier{W}) ; (\id{A} \perG \sigma_{A,W} \perG \id{W})\\
\end{array}\end{equation}
modulo the axioms in   Tables~\ref{fig:freestricmmoncatax} and those in the first and third rows of Figure~\ref{fig:freestrictfbcat}. 

Also the arrows of free strict fc (fp) categories enjoy an elegant graphical representation in terms of string diagrams.  However we are not going to illustrate it now since this would be redundant with tape diagrams.

%% file: tikz/id.tikz
\begin{tikzpicture}
	\begin{pgfonlayer}{nodelayer}
		\node [style=label] (8) at (1.5, 0) {$A$};
		\node [style=label] (11) at (-1.5, 0) {$A$};
	\end{pgfonlayer}
	\begin{pgfonlayer}{edgelayer}
		\draw (11) to (8);
	\end{pgfonlayer}
\end{tikzpicture}

%% file: tikz/stringdiag_ax1_right.tikz
\begin{tikzpicture}
	\begin{pgfonlayer}{nodelayer}
		\node [style=label] (8) at (1.5, 0.5) {$A$};
		\node [style=label] (11) at (-1.5, 0.5) {$A$};
		\node [style=label] (12) at (1.5, -0.5) {$B$};
		\node [style=label] (13) at (-1.5, -0.5) {$B$};
	\end{pgfonlayer}
	\begin{pgfonlayer}{edgelayer}
		\draw (11) to (8);
		\draw (13) to (12);
	\end{pgfonlayer}
\end{tikzpicture}

%% file: appendices/rigapp.tex
\newcommand{\lstre}{l}
\newcommand{\rstre}{r}
\newcommand{\dstre}{s}

\section{Appendix to Section \ref{sec:rigcategories}}

\begin{proof}[Proof of Proposition \ref{prop: Kleisli of distributive monoidal}]
Since $T$ is a symmetric monoidal monad over  $(\Cat{C},\per, \uno)$, the monoidal product $\per$ induces a monoidal product $\per_T$ and this provides $\Cat{C}_T$ with the structure of a symmetric monoidal category (see \cite{guitart1980tenseurs}).

Moreover, finite coproducts in $\Cat{C}$ are preserved by  the canonical functor $\mathcal{K}\colon \Cat{C}\to \Cat{C}_T$ (which preserves colimits because it is left adjoint to the forgetful functor $\Cat{C}_T\to \Cat{C}$), and then $\Cat{C}_T$ is a fc category. 

To see that $\Cat{C}_T$ is a rig category, it is enough to observe that $\mathcal{K}\colon \Cat{C}\to \Cat{C}_T$  preserves  coherence axioms of the rig structure of $\Cat{C}$ and the isomorphisms $\Idl{X}{Y}{Z},\Idr{X}{Y}{Z},\annl{X}$ and $\annr{X}$. Their naturality is proven in Lemma \ref{lemma:nat dist Kleisli 3}.
 \end{proof}
 
 Before providing Lemma \ref{lemma:nat dist Kleisli 3}, we need to prove a preliminary result. To do that, denote with $k\colon T(X)\piu T(Y)\to T(X\piu Y)$ the canonical natural arrow $[T(\linj{X}),T(\rinj{Y})]$, and recall that it is given by
\[
\begin{tikzcd}
	{T(X)\piu T(Y)} && {T(X\piu Y)} \\
	& {T(X\piu Y)\piu T(X\piu Y)}
	\arrow["k", from=1-1, to=1-3]
	\arrow["{T(\linj{X})\piu T(\rinj{Y})}"', from=1-1, to=2-2]
	\arrow["{\codiag{T(X\piu Y)}}"', from=2-2, to=1-3]
\end{tikzcd}\]
where 
\begin{equation*}
\linj{X}\defeq\Irunitp{X};(\id{X}\piu \cobang{Y})\qquad \rinj{Y}\defeq\Ilunitp{Y};(\cobang{X}\piu \id{Y})
\end{equation*}
Moreover, we denote with $\lstre\colon X\per T(Y)\to T(X\per Y)$ and $\rstre\colon T(X)\per Y\to T(X\per Y)$ respectively the left and right strength of $T$. Denote also with $\dstre\colon T(X)\per T(Y)\to T(X\per Y)$ one of the two equal composition $r;T(l);\mu_{X\per Y}=l;T(r);\mu_{X\per Y}$. The naturality of the strengths and $\mu$ imply that $\dstre$ is natural: for every $f\colon X\to X'$ and $g\colon Y\to Y'$ the following diagram commutes
\[
\begin{tikzcd}
	{T(X)\per T(Y)} & {T(X')\per T(Y')} \\
	{T(X\per Y)} & {T(X'\per Y')}
	\arrow["{T(f)\per T(g)}", from=1-1, to=1-2]
	\arrow["\dstre"', from=1-1, to=2-1]
	\arrow["\dstre", from=1-2, to=2-2]
	\arrow["{T(f\per g)}"', from=2-1, to=2-2]
\end{tikzcd}\]
  In addition, the following diagram commutes 
  \begin{equation}\label{diagram: Tcodiag}
  \begin{tikzcd}
  	{T(X)\piu T(X)} & {T(X)} \\
  	& {T(X\piu X)}
  	\arrow["{\codiag{T(X)}}", from=1-1, to=1-2]
  	\arrow["k"', from=1-1, to=2-2]
  	\arrow["{T(\codiag{X})}"', from=2-2, to=1-2]
  \end{tikzcd}\end{equation}
thoruogh the following computation:
\begin{align}
k;T(\codiag{X}) &= T(\Irunitp{X};(\id{X}\piu \cobang{Y}))\piu T(\Ilunitp{Y};(\cobang{X}\piu \id{Y}));\codiag{T(X\piu X)};T(\codiag{X})\notag\\
&= T(\Irunitp{X};(\id{X}\piu \cobang{Y}))\piu T(\Ilunitp{Y};(\cobang{X}\piu \id{Y}));(T(\codiag{X})\piu T(\codiag{X}));\codiag{T(X)} \tag{Nat. $\codiag{}$}\\
&=\codiag{T(X)} \tag{\ref{ax:Mon2}}
\end{align}

 \begin{lemma}\label{lemma: nat dist Kleisli 2}
    If $\Cat{C}$ is a fc rig category, then the following diagram commutes:
    \[
    \begin{tikzcd}
        {(T(X)\piu T(Y))\per T(Z)} & {(T(X)\per T(Z))\piu(T(Y)\per T(Z))} \\
        {T(X\piu Y)\per T(Z)} & {T(X\per Z)\piu T(Y\per Z)} \\
        {T((X\piu Y)\per Z)} & {T((X\per Z)\piu (Y\per Z))}
        \arrow["{\dr{T(X)}{T(Y)}{T(Z)}}", from=1-1, to=1-2]
        \arrow["{k\per \id{T(Z)}}"', from=1-1, to=2-1]
        \arrow["{\dstre\piu \dstre}", from=1-2, to=2-2]
        \arrow["\dstre"', from=2-1, to=3-1]
        \arrow["k", from=2-2, to=3-2]
        \arrow["{T\dr{X}{Y}{Z}}"', from=3-1, to=3-2]
    \end{tikzcd}\]
 \end{lemma}
 \begin{proof}
	First observe that
	\begin{align}
	k\per \id{T(Z)}&=(T(\linj{X})\piu T(\rinj{Y}))\per \id{T(Z)} \notag \\
	&= (T(\Irunitp{X};(\id{X}\piu \cobang{Y}))\piu T(\Ilunitp{Y};(\cobang{X}\piu \id{Y})))\per \id{T(Z)};{\codiag{T(X\piu Y)}\per \id{T(Z)}} \notag\\
	&= T(\Irunitp{X};(\id{X}\piu \cobang{Y}))\piu T(\Ilunitp{Y};(\cobang{X}\piu \id{Y}))\per \id{T(Z)};\dr{T(X\piu Y)}{T(X\piu Y)}{T(Z)};\codiag{T(X\piu Y)\per T(Z)} \tag{Lemma \ref{prop:fcrig}}\\
	&=\dr{T(X)}{T(Y)}{T(Z)};(T(\Irunitp{X};(\id{X}\piu\cobang{Y}))\per \id{T(Z)})\piu (T(\Ilunitp{Y};(\cobang{X}\per \id{Y}))\per \id{T(Z)}) ;\codiag{T(X\piu Y)\per T(Z)} \tag{Nat. $\delta^r$}\\
	&=\dr{T(X)}{T(Y)}{T(Z)};(T(\linj{X})\per \id{T(Z)})\piu(T(\rinj{Y})\per \id{T(Z)});\codiag{T(X\piu Y)\per T(Z)}\notag
	\end{align}
	Hence, the statement is equivalent to the commutativity of the following diagram:
	\[
	\begin{tikzcd}
		{(T(X)\per T(Z))\piu(T(Y)\per T(Z))} & {T(X\per Z)\piu T(Y\per Z)} \\
		{(T(X\piu Y)\per T(Z))\piu (T(X\piu Y)\per T(Z))} \\
		{T(X\piu Y)\per T(Z)} \\
		{T((X\piu Y)\per Z)} & {T((X\per Z)\piu (Y\per Z))}
		\arrow["{\dstre\piu \dstre}", from=1-1, to=1-2]
		\arrow["{(T(\linj{X})\per \id{T(Z)})\piu(T(\rinj{Y})\per \id{T(Z)})}"', from=1-1, to=2-1]
		\arrow["k", from=1-2, to=4-2]
		\arrow["{\codiag{T(X\piu Y)\per T(Z)}}"', from=2-1, to=3-1]
		\arrow["\dstre"', from=3-1, to=4-1]
		\arrow["{T\dr{X}{Y}{Z}}"', from=4-1, to=4-2]
	\end{tikzcd}\]
	Now using the definition of $k$ and (\ref{diagram: Tcodiag}) for the right vertical arrow, and using the naturality of $\codiag{}$ and that of $\dstre$ for the left vertical composition, we obtain the equivalent diagram
	\[
	\begin{tikzcd}[column sep=large]
		{(T(X)\per T(Z))\piu(T(Y)\per T(Z))} & {T(X\per Z)\piu T(Y\per Z)} \\
		{T(X\per Z)\piu T(Y\per Z)} & {T((X\per Z)\piu (Y\per Z))\piu T((X\per Z)\piu (Y\per Z))} \\
		{T((X\piu Y)\per Z)\piu T((X\piu Y)\per Z)} & {T(((X\per Z)\piu (Y\per Z))\piu ((X\per Z)\piu (Y\per Z)))} \\
		{T((X\piu Y)\per Z)} & {T((X\per Z)\piu (Y\per Z))} 
		\arrow["{\dstre\piu \dstre}", from=1-1, to=1-2]
		\arrow["{\dstre\piu \dstre}"', from=1-1, to=2-1]
		\arrow["{T(\linj{X\per Z})\piu T(\rinj{Y\per Z})}", from=1-2, to=2-2]
		\arrow["{T(\linj{X}\per \id{Z})\piu T(\rinj{Y}\per \id{Z})}"', from=2-1, to=3-1]
		\arrow["k", from=2-2, to=3-2]
		\arrow["{\codiag{T((X\piu Y)\per Z)}}"', from=3-1, to=4-1]
		\arrow["{T(\codiag{(X\per Z)\piu (Y\per Z)})}", from=3-2, to=4-2]
		\arrow["{T\dr{X}{Y}{Z}}"', from=4-1, to=4-2]
	\end{tikzcd}\]
	Using (\ref{diagram: Tcodiag}) again, we obtain the commutative diagram
	\[
	\begin{tikzcd}
		{T((X\piu Y)\per Z)\piu T((X\piu Y)\per Z)} & {T((X\per Z)\piu (Y\per Z))\piu T((X\per Z)\piu (Y\per Z))} \\
		{T((X\piu Y)\per Z)\piu(X\piu Y)\per Z))} & {T(((X\per Z)\piu (Y\per Z))\piu ((X\per Z)\piu (Y\per Z)))} \\
		{T((X\piu Y)\per Z)} & {T((X\per Z)\piu (Y\per Z))}
		\arrow["{T(\dr{X}{Y}{Z})\piu T(\dr{X}{Y}{Z})}", from=1-1, to=1-2]
		\arrow["k"', from=1-1, to=2-1]
		\arrow["k", from=1-2, to=2-2]
		\arrow["{T(\dr{X}{Y}{Z}\piu \dr{X}{Y}{Z})}", from=2-1, to=2-2]
		\arrow["{T\codiag{(X\piu Y)\per Z}}"', from=2-1, to=3-1]
		\arrow["{T(\codiag{(X\per Z)\piu (Y\per Z)})}", from=2-2, to=3-2]
		\arrow["{T\dr{X}{Y}{Z}}"', from=3-1, to=3-2]
	\end{tikzcd}\]
	Hence, it only remains to prove that the following triangle commutes:
	\[
	\begin{tikzcd}
		{T(X\per Z)\piu T(Y\per Z)} \\
		{T((X\piu Y)\per Z)\piu T((X\piu Y)\per Z)} & {T((X\per Z)\piu (Y\per Z))\piu T((X\per Z)\piu (Y\per Z))}
		\arrow["{T(\linj{X}\per \id{Z})\piu T(\rinj{Y}\per \id{Z})}"', from=1-1, to=2-1]
		\arrow["{T(\linj{X\per Z})\piu T(\rinj{Y\per Z})}", from=1-1, to=2-2]
		\arrow["{T(\dr{X}{Y}{Z})\piu T(\dr{X}{Y}{Z})}"', from=2-1, to=2-2]
	\end{tikzcd}\]
	which follows from the commutative diagram below obtained through naturality of $\delta^r$ and axioms \ref{eq:dl8} and Lemma \ref{prop:fcrig}
	\[
	\begin{tikzcd}[column sep=large]
		{X\per Z} \\
		{(X\piu 0)\per Z} & {(X\per Z)\piu (0\per Z)} & {(X\per Z)\piu 0} \\
		{(X\piu Y)\per Z} && {(X\per Z)\piu (Y\per Z)}
		\arrow["{\Irunitp{X}\per \id{Z}}"', from=1-1, to=2-1]
		\arrow["{\Irunitp{X\per Z}}",bend left=10 pt ,from=1-1, to=2-3]
		\arrow["{\dr{X}{0}{Y}}"', from=2-1, to=2-2]
		\arrow["{(\id{X}\piu \cobang{Y})\per \id{Z}}"', from=2-1, to=3-1]
		\arrow["{\id{X\per Z}\piu\annl{Z}}", from=2-2, to=2-3]
		\arrow["{(\id{X}\per \id{Z})\piu(\cobang{Y}\per \id{Z})}"{description}, from=2-2, to=3-3]
		\arrow["{\id{X\per Z}\piu \cobang{Y\per Z}}", from=2-3, to=3-3]
		\arrow["{\dr{X}{Y}{Z}}"', from=3-1, to=3-3]
	\end{tikzcd}\] \end{proof}
 \begin{lemma}\label{lemma:nat dist Kleisli 3}
    For a commutative monad $T$ on a fc rig category, the image through $\mathcal{K}\colon C\to C_T$ of the distributor arrows $\delta^r,\delta^s$ is natural in $C_T$.
 \end{lemma}
 \begin{proof}
    We prove the naturality of $\mathcal{K}(\dr{X}{Y}{Z})$
in the $Z$-component, the other cases follow similarly also for $\dl{X}{Y}{Z}$. Consider an arrow $f'\colon Z\to Z'$ in $C_T$, that corresponds to an arrow $f\colon Z\to T(Z')$ in $C$. The statement in $C_T$ of the naturality:
\[
\begin{tikzcd}
	{(X\piu Y)\per Z} & {(X\per Z)\piu (Y\per Z)} \\
	{(X\piu Y)\per Z'} & {(X\per Z')\piu (Y\per Z')}
	\arrow["{\mathcal{K}(\dr{X}{Y}{Z})}", from=1-1, to=1-2]
	\arrow["{(\id{X}\piu\id{Y})\per f}"', from=1-1, to=2-1]
	\arrow["{(\id{X}\per f)\piu (\id{Y}\per f)}", from=1-2, to=2-2]
	\arrow["{\mathcal{K}(\dr{X}{Y}{Z'})}"', from=2-1, to=2-2]
\end{tikzcd}\]

is equivalent to the commutativity of the following outer square diagram in $C$, which follows by naturality of the distributor and $\eta$, and thanks to Lemma \ref{lemma: nat dist Kleisli 2}.

\[
\begin{tikzcd}
	{(X\piu Y)\per Z} & {(X\per Z)\piu (Y\per Z)} & {T((X\per Z)\piu (Y\per Z))} \\
	{(T(X)\piu T(Y))\per T(Z)} & {(T(X)\per T(Z))\piu(T(Y)\per T(Z))} & {T((T(X)\per T(Z))\piu(T(Y)\per T(Z)))} \\
	{T(X\piu Y)\per T(Z)} & {T(X\per Z)\piu T(Y\per Z)} & {T(T(X\per Z)\piu T(Y\per Z))} \\
	{T((X\piu Y)\per Z)} & {T((X\per Z)\piu (Y\per Z))} & {TT((X\per Z)\piu (Y\per Z))} \\
	& {TT((X\per Z)\piu (Y\per Z))} & {T((X\per Z)\piu (Y\per Z))}
	\arrow["{\dr{X}{Y}{Z}}", from=1-1, to=1-2]
	\arrow["{(\eta_X\piu\eta_Y)\per f}"', from=1-1, to=2-1]
	\arrow["\eta", from=1-2, to=1-3]
	\arrow["{(\eta_X\per \id{Z})\piu(\eta_Y\per f)}", from=1-2, to=2-2]
	\arrow["{T((\eta_X\per \id{Z})\piu(\eta_Y\per f))}", from=1-3, to=2-3]
	\arrow["{\dr{T(X)}{T(Y)}{T(Z)}}", from=2-1, to=2-2]
	\arrow["{k\per \id{T(Z)}}"', from=2-1, to=3-1]
	\arrow["\eta", from=2-2, to=2-3]
	\arrow["{\dstre\piu \dstre}", from=2-2, to=3-2]
	\arrow["{T(\dstre\piu \dstre)}", from=2-3, to=3-3]
	\arrow["\dstre"', from=3-1, to=4-1]
	\arrow["\eta", from=3-2, to=3-3]
	\arrow["k", from=3-2, to=4-2]
	\arrow["{T(k)}", from=3-3, to=4-3]
	\arrow["{T\dr{X}{Y}{Z}}"', from=4-1, to=4-2]
	\arrow["\eta", from=4-2, to=4-3]
	\arrow["\eta"', from=4-2, to=5-2]
	\arrow[equals, from=4-2, to=5-3]
	\arrow["\mu", from=4-3, to=5-3]
	\arrow["\mu"', from=5-2, to=5-3]
\end{tikzcd}\]
 \end{proof}

%% file: appendices/fccdrigapp.tex
\newcommand{\twist}{t}

\section{Appendix to Section \ref{sec:fccdrig}}

\begin{proof}[Proof of Proposition \ref{prop:map}]
The first point follows immediately by naturality of $\codiag{X}$ and $\cobang{X}$.

        For the second point, observe that since $0$ is initial, then $\cobang{X}$ is functional and total. To prove functionality of  $\codiag{X}$, which is $\codiag{X};\copier{X}= \copier{X\piu X};(\codiag{X}\per \codiag{X})$, observe that
        \[(\Idl{X}{X}{X}\oplus \Idl{X}{X}{X}); (\Idr{X}{X}{X\oplus X});(\codiag{X}\per \codiag{X})=(\codiag{X\per X}\piu \codiag{X\per X});\codiag{X\per X}\]
        through the following computation:
      \begin{align}
        (\Idl{X}{X}{X}\oplus \Idl{X}{X}{X}); (\Idr{X}{X}{X\oplus X});(\codiag{X}\per \codiag{X})&=(\Idl{X}{X}{X}\oplus \Idl{X}{X}{X}); (\Idr{X}{X}{X\oplus X});(\codiag{X}\per\id{X});(\id{X} \per\codiag{X})\notag\\
        &=(\Idl{X}{X}{X}\oplus \Idl{X}{X}{X});\codiag{X\per (X\piu X)};(\id{X} \per\codiag{X}) \tag{Prop. \ref{prop:fcrig}}\\
        &=\Idl{X}{X}{X}\oplus \Idl{X}{X}{X}; (\id{X} \per\codiag{X})\piu(\id{X} \per\codiag{X});\codiag{X\per X} \tag{Nat. $\codiag{}$}\\
        &=(\codiag{X\per X}\piu \codiag{X\per X});\codiag{X\per X} \tag{Prop. \ref{prop:fcrig}}
      \end{align}
         
          Hence, using the coherence axiom \eqref{equation: coherence1}, it follows that $\copier{X\piu X};(\codiag{X}\per \codiag{X})$ is equal to 
          \[(\Irunitp{X}\piu\Ilunitp{Y} );((\copier{X} \oplus \cobang{X\per X}) \oplus( \cobang{X \per X} \oplus \copier{X}));(\codiag{X\per X}\piu \codiag{X\per X});\codiag{X\per X}\]
          Now, since $(\copier{X} \oplus \cobang{X\per X});\codiag{X\per X}=\lunitp{X};\copier{X}$ and $(  \cobang{X\per X}\oplus\copier{X});\codiag{X\per X}=\runitp{X};\copier{X}$, and $\codiag{X\per X}$ is natural, we obtain
          \[\copier{X\piu X};(\codiag{X}\per \codiag{X})=(\Irunitp{X}\piu\Ilunitp{X}  );(\runitp{X}\piu\lunitp{X} );\codiag{X};\copier{X}=\codiag{X};\copier{X}\]

        Totality of $\codiag{X}$ is $\codiag{X};\discharger{X}= \discharger{X\piu X}$ and coherence axiom \eqref{equation: coherence1} implies that $\discharger{X\piu X}=(\discharger{X}\piu \discharger{X});\codiag{\unoG}$. Hence, totality follows from naturality of $\codiag{\unoG}$.
        
  For the third point observe that since $\codiag{X}$ and $\cobang{X}$ are functional and total $(\mathrm{Map}(\Cat{C}), \piu, \zero)$ has commutative monoids, that are natural and coherent. It remains to prove that the arrows $\dr{X}{Y}{Z},\dl{X}{Y}{Z}$ and $\annr{X},\annl{X}$ are maps. The most difficult part to prove is the functionality of $\dr{X}{Y}{Z}$. For $\annl{X}$, functionality and totality correspond to \[\annl{X};\copier{0}=\copier{0\per X};(\annl{X}\per \annl{X})\qquad \annl{X};\discharger{0}=\discharger{0\per X}\] 
which follow from the fact that $\annl{X}$ is an isomorphism and $0$ is initial. The same holds for $\annr{X}$. Totality of $\dr{X}{Y}{Z}$ corresponds to 
\[\dr{X}{Y}{Z};\discharger{(X\per Z)\piu (Y\per Z)}=\discharger{(X\piu Y)\per Z}\]
which is proved as follows.
Naturality of the distributor and coherence axiom for $\discharger{-\piu -}$ imply the commutativity of the following diagrams
\[
\begin{tikzcd}
	{(X\piu Y)\per Z} & {(X\per Z)\piu (Y\per Z)} &&& {(X\per Z)\piu (Y\per Z)} && 1 \\
	{(1\piu 1)\per 1} & {(1\per 1)\piu (1\per 1)} &&& {1\piu 1}
	\arrow["{\dr{X}{Y}{Z}}", from=1-1, to=1-2]
	\arrow["{(\discharger{X}\piu\discharger{Y})\per \discharger{Z}}"', from=1-1, to=2-1]
	\arrow["{(\discharger{X}\per\discharger{Z})\piu (\discharger{Y}\per\discharger{Z})}", from=1-2, to=2-2]
	\arrow["{\discharger{(X\per Z)\piu(Y\per Z)}}", from=1-5, to=1-7]
	\arrow["{\discharger{X\per Z}\piu\discharger{Y\per Z}}"', from=1-5, to=2-5]
	\arrow["{\dr{1}{1}{1}}"', from=2-1, to=2-2]
	\arrow["{\codiag{1}}"', from=2-5, to=1-7]
\end{tikzcd}\]
Now, since $(\discharger{X}\per\discharger{Z})\piu (\discharger{Y}\per\discharger{Z});\runitt{1}\piu \runitt{1}=\discharger{X\per Z}\piu\discharger{Y\per Z}$ we obtain that 
\begin{align}
\dr{X}{Y}{Z};\discharger{(X\per Z)\piu(Y\per Z)}&=\dr{X}{Y}{Z};(\discharger{X\per Z}\piu\discharger{Y\per Z});\codiag{1}\notag\\
&=\dr{X}{Y}{Z};(\discharger{X}\per\discharger{Z})\piu (\discharger{Y}\per\discharger{Z});(\runitt{1}\piu \runitt{1});\codiag{1}\notag\\
&=\dr{X}{Y}{Z};(\discharger{X}\per\discharger{Z})\piu (\discharger{Y}\per\discharger{Z});\codiag{1\per 1};\runitt{1} \tag{Nat. $\codiag{}$}\\
&=(\discharger{X}\piu\discharger{Y})\per \discharger{Z};\dr{1}{1}{1};\codiag{1\per 1};\runitt{1} \tag{Nat. $\dr{X}{Y}{Z}$}\\
&=(\discharger{X}\piu\discharger{Y})\per \discharger{Z};(\codiag{1}\per \id{1});\runitt{1} \tag{Prop. \ref{prop:fcrig}}\\
&=(\discharger{X\piu Y}\per \discharger{Z});\runitt{1} \tag{Coh. \eqref{equation: coherence1}}\\
&=\discharger{(X\piu Y)\per Z} \notag
\end{align}
It remains to prove that $\dr{X}{Y}{Z}$ is functional, which is 
\[\dr{X}{Y}{Z};\copier{(X\per Z)\piu (Y\per Z)}=\copier{(X\piu Y)\per Z};(\dr{X}{Y}{Z}\per \dr{X}{Y}{Z})\]
Thanks to \eqref{eq:coherence codiag}, $\copier{(X\piu Y)\per Z}= (\copier{(X\piu Y)}\per\copier{Z});\twist$, where $\twist$ is the arrow that switches the two inner components of a product of two products, it is obtained as the composition of the arrows in \eqref{eq:coherence codiag}.
	Hence, the statement is equivalent to the commutativity of the following diagram:
  \[
\begin{tikzcd}
	{(X\piu Y) \per Z} & {((X\piu Y)\per (X\piu Y))\per (Z\per Z)} \\
	& {((X+Y)\per Z)\per ((X\piu Y)\per Z)} \\
	{(X\per Z)+(Y\per Z)} & {((X\per Z)+(Y\per Z))\per ((X\per Z)\piu(Y\per Z))}
	\arrow["{\copier{X\piu Y} \per \copier{Z}}", from=1-1, to=1-2]
	\arrow["{\dr{X}{Y}{Z}}"', from=1-1, to=3-1]
	\arrow["\twist", from=1-2, to=2-2]
	\arrow["{\dr{X}{Y}{Z}\per \dr{X}{Y}{Z}}", from=2-2, to=3-2]
	\arrow["{\copier{(X\per Z)\piu (Y\per Z)}}"', from=3-1, to=3-2]
\end{tikzcd}\]

Coherence axiom \ref{equation: coherence1} implies that $\copier{(X\per Y)\piu (Y\per Z)}$ is equal to

$$(\Irunitp{X\per Z}\piu\Ilunitp{Y\per Z} ); ((\copier{X\per Z} \oplus \cobang{(X\per Z)\per(Y\per Z)}) \oplus( \cobang{(Y\per Z)\per(X\per Z)} \oplus \copier{Y\per Z}));(\Idl{X\per Z}{X\per Z}{Y\per Z}\oplus \Idl{Y\per Z}{X\per Z}{Y\per Z}); (\Idr{X\per Z}{Y\per Z}{(X\per Z)\oplus (Y\per Z)})$$
and that 
\[\copier{X\oplus Y}=(\Irunitp{Y}\piu\Ilunitp{X} ); ((\copier{X} \oplus \cobang{X\per Y}) \oplus( \cobang{Y \per X} \oplus \copier{Y}));(\Idl{X}{X}{Y}\oplus \Idl{Y}{X}{Y}); (\Idr{X}{Y}{X\oplus Y})\]
which implies the commutativity of the following diagram:

\[
\adjustbox{scale=0.9,center}{\begin{tikzcd}
	{(X\piu Y) \per Z} &&& {((X\piu Y)\per (X\piu Y))\per (Z\per Z)} \\
	& {((X\piu 0)\piu(0\piu Y))\per Z} && {((X\piu Y)\per (X\piu Y))\per Z} \\
	&&& {((X\per(X\piu Y))\piu(Y\per(X\piu Y)))\per Z} \\
	&&& {(((X\per X)\piu (X\per Y))\piu ((Y\per X)\piu (Y\per Y)))\per Z}
	\arrow["{\copier{X\piu Y} \per \copier{Z}}", from=1-1, to=1-4]
	\arrow["{(\Irunitp{X}\piu\Ilunitp{Y})\per \id{Z}}"', from=1-1, to=2-2]
	\arrow["{((\copier{X} \oplus \cobang{X\per Y}) \oplus( \cobang{Y \per X} \oplus \copier{Y}))\per \id{Z}}"', from=2-2, to=4-4]
	\arrow["{\id{}\per \copier{Z}}"', from=2-4, to=1-4]
	\arrow["{(\Idr{X}{Y}{X\oplus Y})\per \id{Z}}"', from=3-4, to=2-4]
	\arrow["{(\Idl{X}{X}{Y}\oplus \Idl{Y}{X}{Y})\per \id{Z}}"', from=4-4, to=3-4]
\end{tikzcd}}\]
Now observe that the following diagram commutes (from now on we suppress the symbol $\per$ for the sake of readability):
\let\oldper\per
\renewcommand{\per}{}
\[
\adjustbox{scale=0.8,center}{\begin{tikzcd}
	{(X\piu Y) \per Z} & {((X\piu 0)+(0\piu Y))\per Z} & {(((X\per X)\piu (X\per Y))\piu ((Y\per X)\piu (Y\per Y)))\per Z} \\
	{(X\per Z) \piu (Y\per Z)} && {(((X\per X)\piu (X\per Y))\per Z)\piu (((Y\per X)\piu (Y\per Y))\per Z)} \\
	{((X\per Z)\piu 0) \piu ((Y\per Z)+ 0)} && {(((X\per X)\per Z)\piu((X\per Y)\per Z))\piu (((Y\per X)\per Z)\piu((Y\per Y)\per Z))}
	\arrow["{(\Irunitp{X}\piu\Ilunitp{Y})\per \id{Z}}", from=1-1, to=1-2]
	\arrow["{\dr{X}{Y}{Z}}"', from=1-1, to=2-1]
	\arrow["{((\copier{X} \oplus \cobang{X\per Y}) \oplus( \cobang{Y \per X} \oplus \copier{Y}))\per \id{Z}}", from=1-2, to=1-3]
	\arrow["{\dr{(X\per X)\piu (X\per Y)}{(Y\per X)\piu (Y\per Y)}{Z}}", from=1-3, to=2-3]
	\arrow["{((\Irunitp{X};(\copier{X}\piu \cobang{X\per Y}))\per \id{Z})\piu ((\Ilunitp{X};(\cobang{Y\per X}\piu\copier{Y}))\per \id{Z}) }"{description}, from=2-1, to=2-3]
	\arrow["{(\Irunitp{X\per Z}\piu\Ilunitp{Y\per Z})}"', from=2-1, to=3-1]
	\arrow["{\dr{X\per X}{X\per Y}{Z}\piu\dr{Y\per X}{Y\per Y}{Z}}", from=2-3, to=3-3]
	\arrow["{((\copier{X}\per \id{Z})\piu \cobang{(X\per Y)\per Z})\piu (\cobang{(Y\per X)\per Z}\piu(\copier{Y}\per \id{Z}))}"', from=3-1, to=3-3]
\end{tikzcd}}\]
Indeed, the top square commutes by naturality of the distributor. For the bottom square, it can be factorized as follows:
\[
\adjustbox{scale=0.8,center}{\begin{tikzcd}
	{(X\per Z) \piu (Y\per Z)} & {((X\piu 0)\per Z)\piu ((0\piu Y)\per Z)} & {(((X\per X)\piu (X\per Y))\per Z)\piu (((Y\per X)\piu (Y\per Y))\per Z)} \\
	\\
	{((X\per Z)\piu 0) \piu (0\piu (Y\per Z))} & {((X\per Z)\piu (0\per Z))\piu ((0\per Z)\piu (Y\per Z))} & {(((X\per X)\per Z)\piu((X\per Y)\per Z))\piu (((Y\per X)\per Z)\piu((Y\per Y)\per Z))}
	\arrow["{(\Irunitp{X}\per \id{Z})\piu (\Ilunitp{Y}\per \id{Z})}", from=1-1, to=1-2]
	\arrow["{(\Irunitp{X\per Z}\piu\Ilunitp{Y\per Z})}"', from=1-1, to=3-1]
	\arrow["{((\copier{X}\piu \cobang{X\per Y})\per \id{Z})\piu((\cobang{Y\per X}\piu \copier{Y})\per \id{Z})}", from=1-2, to=1-3]
	\arrow["{\dr{X}{0}{Z}\piu\dr{0}{Y}{Z}}"{description}, from=1-2, to=3-2]
	\arrow["{(\dr{X\per X}{X\per Y}{Z}\piu\dr{Y\per X}{Y\per Y}{Z})}", from=1-3, to=3-3]
	\arrow["{(\id{X\per Z}\piu \Iannl{Z})\piu (\Iannl{Z}\piu \id{Y\per Z} )}"', from=3-1, to=3-2]
	\arrow["{((\copier{X}\per \id{Z})\piu (\cobang{X\per Y}\per\id{Z}))\piu ((\cobang{Y\per X}\per \id{Z})\piu (\copier{Y}\per\id{Z}))}"', from=3-2, to=3-3]
\end{tikzcd}}\]
Looking at the two components of $\piu$, the right squares commute by naturality of the distributor, while the left ones commute by \ref{eq:dl8}.

Now, since the following diagram commutes through axioms \ref{eq:rigax1} and \ref{eq:rigax4}
\[
\adjustbox{scale=0.8,center}{\begin{tikzcd}
	{((X\piu Y)\per (X\piu Y))\per (Z\per Z)} & {((X\piu Y)\per Z)\per ((X\piu Y)\per Z)} \\
	{((X\per(X\piu Y))\piu(Y\per(X\piu Y)))\per (Z\per Z)} \\
	{(((X\per X)\piu (X\per Y))\piu ((Y\per X)\piu (Y\per Y)))\per (Z\per Z)} & {{((X\per Z)+(Y\per Z))\per ((X\per Z)+(Y\per Z))}} \\
	{{(((X\per X)\piu (X\per Y))\per (Z\per Z))\piu (((Y\per X)\piu (Y\per Y))\per (Z\per Z))}} & {((X\per Z)\per  ((X\per Z)\piu (Y\per Z)))\piu ((Y\per Z)\per  ((X\per Z)\piu (Y\per Z)))} \\
	{{(((X\per X)\per (Z\per Z))\piu((X\per Y)\per (Z\per Z)))\piu (((Y\per X)\per (Z\per Z))\piu((Y\per Y)\per (Z\per Z)))}} & {(((X\per Z)\per (X\per Z))\piu ((X\per Z)\per (Y\per Z)))\piu (((Y\per Z)\per (X\per Z))\piu ((Y\per Z)\per (Y\per Z)))}
	\arrow["\twist", from=1-1, to=1-2]
	\arrow["{\dr{X}{Y}{Z}\per \dr{X}{Y}{Z}}", from=1-2, to=3-2]
	\arrow["{(\Idr{X}{Y}{X\oplus Y})\per \id{Z\per Z}}", from=2-1, to=1-1]
	\arrow["{(\Idl{X}{X}{Y}\oplus \Idl{Y}{X}{Y})\per \id{Z\per Z}}", from=3-1, to=2-1]
	\arrow["{\dr{(X\per X)\piu (X\per Y)}{(Y\per X)\piu (Y\per Y)}{ZZ}}"', from=3-1, to=4-1]
	\arrow["{\dr{X\per X}{X\per Y}{ZZ}\piu\dr{Y\per X}{Y\per Y}{ZZ}}"', from=4-1, to=5-1]
	\arrow["{\Idr{X\per Z}{Y\per Z}{(X\per Z)\piu(Y\per Z)}}"', from=4-2, to=3-2]
	\arrow["{(\twist\piu \twist)\piu(\twist\piu \twist)}"', from=5-1, to=5-2]
	\arrow["{\Idl{X\per Z}{X\per Z}{Y\per Z}\piu \Idl{YZ}{XZ}{YZ}}"', from=5-2, to=4-2]
\end{tikzcd}}\]

The naturality of the distributor implies that the claim is equivalent to the commutativity of the following diagram

\[
\begin{tikzcd}
	{((X\per Z)\piu 0) \piu ((Y\per Z)+ 0)} & {(((X\per X)\per Z)\piu((X\per Y)\per Z))\piu (((Y\per X)\per Z)\piu((Y\per Y)\per Z))} \\
	& {{(((X\per X)\per (Z\per Z))\piu((X\per Y)\per (Z\per Z)))\piu (((Y\per X)\per (Z\per Z))\piu((Y\per Y)\per (Z\per Z)))}} \\
	& {(((X\per Z)\per (X\per Z))\piu ((X\per Z)\per (Y\per Z)))\piu (((Y\per Z)\per (X\per Z))\piu ((Y\per Z)\per (Y\per Z)))}
	\arrow["{((\copier{X}\per \id{Z})\piu \cobang{(X\per Y)\per Z})\piu (\cobang{(Y\per X)\per Z}\piu(\copier{Y}\per \id{Z}))}", from=1-1, to=1-2]
	\arrow["{ ((\copier{X\per Z} \oplus \cobang{(X\per Z)\per(Y\per Z)}) \oplus( \cobang{(Y\per Z)\per(X\per Z)} \oplus \copier{Y\per Z}))}"',bend right=5pt, shift right=5, shorten >=64pt, from=1-1, to=3-2]
	\arrow["{((\id{X\per X}\per \copier{Z})\piu (\id{X\per Y}\per \copier{Z}))\piu((\id{Y\per X}\per \copier{Z})\piu(\id{Y\per Y}\per \copier{Z}) )}"{description}, from=1-2, to=2-2]
	\arrow["{(s\piu s)\piu(s\piu s)}"', from=2-2, to=3-2]
\end{tikzcd}\]
which follows by naturality of $\cobang{}$ and the definition of $\copier{X\per Z}$.
\let\per\oldper
    \end{proof}

%
%
%
%
\begin{proof}[Proof of Lemma \ref{lemma:distributive}]
     Let $\Cat{C}$ be a fc-fp rig category. Since $(\Cat{C},\per ,\uno)$ is a fp category, then it is also cd, with $\copier{X}$ given by the diagonal arrow $\langle \id{X},\id{X}\rangle$ and $\discharger{X}$ given by the unique arrow from $X$ to the terminal object $1$. It only remain to proove coherence conditions in\eqref{equation: coherence1}.
For $\copier{X\oplus Y}=(\Irunitp{X}\piu\Ilunitp{Y} ); ((\copier{X} \oplus \cobang{X\per Y}) \oplus( \cobang{Y \per X} \oplus \copier{Y}));(\Idl{X}{X}{Y}\oplus \Idl{Y}{X}{Y}); (\Idr{X}{Y}{X\oplus Y})$
We first provide the following equality:
\begin{align}
\Irunitp{X};(\copier{X} \oplus \cobang{X\per Y});(\Idl{X}{X}{Y})=&\Irunitp{X};(\copier{X}\piu \id{0});(\id{X}\piu\Iannr{X});(\id{X\per X}\piu (\id{X}\per \cobang{Y}));\Idl{X}{X}{Y} \tag{Lemma \ref{prop:fcrig}}\\
=&\Irunitp{X};(\copier{X}\piu \id{0});(\id{X}\piu\Iannr{X});\Idl{X}{X}{0};(\id{X}\per(\id{X}\piu\cobang{Y})) \tag{Nat. $\delta^r$}\\
=& \copier{X};\Irunitp{X\per X};(\id{X}\piu\Iannr{X});\Idl{X}{X}{0};(\id{X}\per(\id{X}\piu\cobang{Y})) \tag{Nat. $\Irunitp{X}$}\\
=& \copier{X};(\id{X}\per \runitp{X});(\id{X}\per(\id{X}\piu\cobang{Y}))\tag{Ax. \ref{eq:dl7}}
\end{align}
 
 And similarly we can prove that $\Ilunitp{Y};(\cobang{Y \per X} \oplus \copier{Y});\Idl{Y}{X}{Y}=\copier{Y};(\id{Y}\per \Ilunitp{Y});(\id{Y}\per(\cobang{X}\piu \id{Y}))$.
 
 Hence, the statement is equivalent to the equality:
 \begin{equation}\label{eq:lemma 19 proof}
 \copier{X\piu Y};\dr{X}{Y}{X\piu Y}= (\copier{X}\piu\copier{Y});(\id{X}\per \runitp{X})\piu(\id{Y}\per \Ilunitp{Y});(\id{X}\per(\id{X}\piu\cobang{Y}))\piu(\id{Y}\per(\cobang{X}\piu \id{Y}))
 \end{equation}
 To prove this, we observe that $\id{X\piu Y}= (\Irunitp{X}\piu \Ilunitp{Y});(\id{X}\piu \cobang{Y})\piu (\cobang{X}\piu \id{Y});\codiag{X\piu Y}$ and prove the above equality precomposing it with $\id{X\piu Y}$. For the left term we obtain that $(\Irunitp{X}\piu \Ilunitp{Y});(\id{X}\piu \cobang{Y})\piu (\cobang{X}\piu \id{Y});\codiag{X\piu Y};\copier{X\piu Y};\dr{X}{Y}{X\piu Y}$ is equal to 
 \begin{align}
 (\Irunitp{X}\piu \Ilunitp{Y});(\id{X}\piu \cobang{Y})\piu (\cobang{X}\piu \id{Y});(\copier{X\piu Y}\piu \copier{X\piu Y});(\dr{X}{Y}{X\piu Y}\piu \dr{X}{Y}{X\piu Y});\codiag{(X\per (X\piu Y))\piu (Y\per (X\piu Y))} \notag
 \end{align}
 by naturality of $\codiag{}$. Hence, working on the separate components of $\piu$ we obtain 
 \begin{align*}
 &\Irunitp{X};(\id{X}\piu \cobang{Y});\copier{X\piu Y};\dr{X}{Y}{X\piu Y} \\ 
=&\copier{X};(\Irunitp{X}\piu \Irunitp{X});((\id{X}\piu \cobang{Y})\piu (\id{X}\piu \cobang{Y}));\dr{X}{Y}{X\piu Y} \tag{Nat. $\copier{}$}\\
=&\copier{X};(\Irunitp{X}\piu \Irunitp{X});\dr{X}{0}{X\piu 0};(\id{X}\per (\id{X}\piu \cobang{Y}))\piu (\cobang{Y}\per (\id{x}\piu\cobang{Y})) \tag{Nat. $\delta^r$}\\
=&\copier{X};(\id{X}\per \Irunitp{X});\Irunitp{X\per (X\piu 0)};(\id{X\per (X\piu 0)}\piu \Iannl{X\piu 0});(\id{X}\per (\id{X}\piu \cobang{Y}))\piu (\cobang{Y}\per (\id{x}\piu\cobang{Y})) \tag{Ax. \ref{eq:dl8}}\\
=&\copier{X};(\id{X}\per \Irunitp{X});\Irunitp{X\per (X\piu 0)};(\id{X}\per (\id{X}\piu \cobang{Y})\piu \cobang{Y\per (X\piu Y)}) \tag{Nat. $\cobang{}$}
 \end{align*}
Similarly one obtains $\Ilunitp{Y}; (\cobang{X}\piu \id{Y});\copier{X\piu Y};\dr{X}{Y}{X\piu Y}=\copier{Y};(\id{Y}\per \Ilunitp{Y});\Ilunitp{Y\per (0\piu Y)};(\cobang{X\per(X\piu Y)}\piu (\id{Y}\per (\cobang{X}\piu \id{Y})))$.

For the right term, after precomposing with $(\Irunitp{X}\piu \Ilunitp{Y});(\id{X}\piu \cobang{Y})\piu (\cobang{X}\piu \id{Y});\codiag{X\piu Y}$ and applying the naturality of $\codiag{}$, we consider separately the components of $\piu$ and obtain
\begin{align}
\Irunitp{X};(\id{X}\piu \cobang{Y});\copier{X}\piu\copier{Y};(\id{X}\per \runitp{X})\piu(\id{Y}\per \Ilunitp{Y});(\id{X}\per(\id{X}\piu\cobang{Y}))\piu(\id{Y}\per(\cobang{X}\piu \id{Y}))&=\tag{Nat. $\cobang{}$}\\
\Irunitp{X};(\copier{X}\piu \id{0});(\id{X\per X}\piu \cobang{Y\per Y});(\id{X}\per \Irunitp{X})\piu(\id{Y}\per \Ilunitp{Y});(\id{X}\per(\id{X}\piu\cobang{Y}))\piu(\id{Y}\per(\cobang{X}\piu \id{Y}))&= \tag{Nat. $\Irunitp{}$}\\
\copier{X};\Irunitp{(X\per X)};(\id{X\per X}\piu \cobang{Y\per Y});(\id{X}\per \Irunitp{X})\piu(\id{Y}\per \Ilunitp{Y});(\id{X}\per(\id{X}\piu\cobang{Y}))\piu(\id{Y}\per(\cobang{X}\piu \id{Y}))&=\tag{Nat. $\Irunitp{}$}\\
\copier{X};(\id{X}\per \Irunitp{X});\Irunitp{X\per (X\piu 0)};(\id{X\per (X\piu 0)}\piu \cobang{Y\per(0\piu Y)});(\id{X}\per(\id{X}\piu\cobang{Y}))\piu(\id{Y}\per(\cobang{X}\piu \id{Y}))&= \tag{Nat. $\Irunitp{}$}\\
\copier{X};(\id{X}\per \Irunitp{X});\Irunitp{X\per (X\piu 0)};(\id{X}\per (\id{X}\piu \cobang{Y})\piu \cobang{Y\per (X\piu Y)}) \tag{Nat. $\cobang{}$}
\end{align}
Similarly one obtains that $\Ilunitp{Y}; (\cobang{X}\piu \id{Y});\copier{X}\piu\copier{Y};(\id{X}\per \Irunitp{X})\piu(\id{Y}\per \Ilunitp{Y});(\id{X}\per(\id{X}\piu\cobang{Y}))\piu(\id{Y}\per(\cobang{X}\piu \id{Y}))$ is equal to $\copier{Y};(\id{Y}\per \Ilunitp{Y});\Ilunitp{Y\per (0\piu Y)};(\cobang{X\per(X\piu Y)}\piu (\id{Y}\per (\cobang{X}\piu \id{Y})))$. Hence the equality (\ref{eq:lemma 19 proof}) is proved.

The coherence condition in \eqref{equation: coherence1} for $\discharger{X\piu Y}$ follows immediately from the fact that $\discharger{}$ is natural.

    \end{proof}

	\begin{proof}[Proof of Proposition \ref{prop:Kleislifccd}]
	Since $\Cat{C}$ is a cd category, then by Propositions \ref{prop: Kleisli gs e cogs}, $\Cat{C}_T$ is a cd category.
	Moreover, because $\Cat{C}$ is a fc rig category, then Proposition \ref{prop: Kleisli of distributive monoidal} guarantees that  $\Cat{C}_T$ is also a fc rig category. Finally, since $\mathcal{K}\colon \Cat{C}\to \Cat{C}_T$ maps $\copier{}$ and $\discharger{}$ into $\copier{}^\flat$ and $\discharger{}^\flat$ as defined in \eqref{eq:copierflat},  and then it 
	preserves the coherence conditions in \eqref{equation: coherence1}. Hence,  $\Cat{C}_T$ is a fc-cd rig category.
	\end{proof}

%% file: appendices/appseigmae.tex
\section{Appendix to Section \ref{sec:Trig}}\label{app:Trig}\label{app:Trig categories}
%
\begin{proof}[Proof of Lemma \ref{lemma:aleph}]
Consider $\Cat{C}$ and forgets about the monoidal structure given by $\per$. By Mac Lane's strictification theorem for monoidal categories \cite[XI.3 Theorem 1]{mac_lane_categories_1978}, there exists a category $\Cat{C_S}$ which is monoidally equivalent to $(\Cat{C},\piu, \zero)$. Recall that objects of $\Cat{C_S}$ are lists of objects of $\Cat{C}$, in symbols  $[X_1,X_2, \dots, X_n]$ such that $X_i \in Ob(\Cat{C})$, and the monoidal product $\piu$ is defined as list concatenation. The monoidal functor witnessing the equivalence
$s\colon \Cat{C_S} \to \Cat{C}$ is defined on object as 
\[[X_1,X_2, \dots, X_n] \mapsto \PiuL[i=1][n]{X_i}\text{.}\]
Note that, in particular
\begin{equation}\label{eq:inproof1}
[1,1, \dots, 1] \text{ of length } n\mapsto \PiuL[i=1][n]{1}\text{.}
\end{equation}
Recall that $\aleph_0$ is the strict finite coproduct category freely generated by a single object $\uno$. 
Thus, there exists a unique strict finite coproduct morphism $t \colon \aleph_0 \to \Cat{C_S}$ mapping $\uno$ into $[1]$.
Since $t$ is a strict monoidal functor it maps 
\begin{equation}\label{eq:inproof2}
n \in \mathbb{N}  \mapsto [1,1\dots, 1] \text{ of length }n \text{.}
\end{equation}
Now, it is enough to take $c\defeq t;s \colon  \aleph_0 \to \Cat{C}$, and observe that by \eqref{eq:inproof1} and \eqref{eq:inproof2}, each $n \in \mathbb{N}$ is mapped into $\PiuL[i=1][n]{\uno}$. It preserve finite coproducts, since $t$ preserves them and $s$ is an equivalence of monoidal categories.
\end{proof}

\begin{lemma}\label{lemma: f_X naturale}
    Let $i:\opcat{\lawveret}\to\Cat{C}$ be a $(\Sigma, E)$-rig category. For every object $X$ of $\Cat{C}$ and every $f\colon n\to 1$ in $\Sigma$, there is an arrow $\Br{f}_X\colon X\to\PiuL[i=1][n]{X}$ natural in $X$.
    \begin{proof}
        Take $\Br{f}_X$ to be the following composition:
\[
\begin{tikzcd}
	X & {1\per X} & { (\PiuL[i=1][n]{1})\per  X} & {\PiuL[i=1][n]{ X}}
    \arrow["{\Ilunitt{X}}", from=1-1, to=1-2]
	\arrow["{ i(f)\per \id{X}}", from=1-2, to=1-3]
	\arrow["{\delta^r_{n,X}}", from=1-3, to=1-4]
\end{tikzcd}\]
and consider an arrow $h:X\to Y$, then $h;\Br{f}_Y= \Br{f}_X;\PiuL[i=1][n]{h}$ since each of the following square diagrams commute by naturality of $\lunitt{X},\Ilunitt{X},\annl{X}$ and $\dr{X}{Y}{Z}$:
\[
\begin{tikzcd}[row sep=large]
	X & {1\per X} & { (\PiuL[i=1][n]{1})\per  X} & {\PiuL[i=1][n]{ X}} \\
	Y & {1\per Y} & { (\PiuL[i=1][n]{1})\per  Y} & {\PiuL[i=1][n]{ Y}}
	\arrow["{\Ilunitt{X}}", from=1-1, to=1-2]
	\arrow["h"', from=1-1, to=2-1]
	\arrow["{ i(f)\per \id{X}}", from=1-2, to=1-3]
	\arrow["{\id{1}\per h}"{description}, from=1-2, to=2-2]
	\arrow["{\delta^r_{n,X}}", from=1-3, to=1-4]
	\arrow["{\id{n} \per h}"{description}, from=1-3, to=2-3]
	\arrow["{\PiuL[i=1][n]{h}}"{description}, from=1-4, to=2-4]
	\arrow["{\Ilunitt{Y}}"', from=2-1, to=2-2]
	\arrow["{ i(f)\per \id{Y}}"', from=2-2, to=2-3]
	\arrow["{\delta^r_{n,Y}}"', from=2-3, to=2-4]
\end{tikzcd}\]
    \end{proof}
\end{lemma}

Hereafter, we write $\delta^r_{nY,X}:(\PiuL[i=1][n]{Y})\otimes X \to \PiuL[i=1][n]{(Y\per X)}$
for the arrow inductively defined as follows:
\begin{equation}
\delta^r_{0Y,X}\colon=\annl{X} \qquad \delta^r_{1Y,X}\colon=\id{Y\per X}\qquad \delta^r_{(n+1)Y,X}\colon=\dr{Y}{nY}{X};(\id{Y\per X}\piu \delta^r_{nY,X})\text{.}
\end{equation}
where $mY$ stands for the object $ \PiuL[i=1][m]{Y}$. A simple proof by induction shows that $\delta^r_{nY,X}$ is natural in $X,Y$. The arrow $\delta^l_{X,nY} \colon X\otimes \PiuL[i=1][n]{Y} \to \PiuL[i=1][n]{(X\per Y)}$ is defined similarly.
\begin{lemma}\label{lemma: f_P x idR= f_PR}
    For every $X,Y\in \Cat{C}$ it holds that
    \[(\Br{f}_X\per \id{Y});\delta^r_{nX,Y}=\Br{f}_{X\per Y}\qquad (\id{Y}\per \Br{f}_X);\delta^l_{Y,nX}= \Br{f}_{Y\per X}\]
    \begin{proof}
        The left equation follows from the commutativity of the following diagram:
        \[
        \begin{tikzcd}
            {X\per Y} & {(1\per X)\per Y} & {((\PiuL[i=1][n]{1})\per X)\per Y} && {(\PiuL[i=1][n]{X})\per Y} \\
            & {1\per (X\per Y)} & {(\PiuL[i=1][n]{1})\per (X\per Y)} && {\PiuL[i=1][n]{(X\per Y)}}
            \arrow["{\Ilunitt{X}\per \id{Y}}", from=1-1, to=1-2]
            \arrow["{\Ilunitt{X\per Y}}"', from=1-1, to=2-2]
            \arrow["{(i(f)\per \id{X})\per \id{Y}}", from=1-2, to=1-3]
            \arrow["{\delta^r_{n,X}\per \id{Y}}", from=1-3, to=1-5]
            \arrow["{\delta^r_{nX,Y}}", from=1-5, to=2-5]
            \arrow["{\Iassoct{1}{X}{Y}}"{description}, from=2-2, to=1-2]
            \arrow["{i(f)\per (X\per Y)}"', from=2-2, to=2-3]
            \arrow["{\Iassoct{\objn{n}}{X}{Y}}"{description}, from=2-3, to=1-3]
            \arrow["{\delta^r_{n,X\per Y}}"', from=2-3, to=2-5]
        \end{tikzcd}\]
        where the left triangle and the middle square commute by associator properties. The right square is proved to commute by induction. The case $n=0$, namely $\Iassoct{0}{X}{Y};(\delta^r_{0,X}\per \id{Y});\delta^r_{0X,Y}=\delta^r_{0,X\per Y}$ follows from \ref{eq:rigax9} + \ref{eq:rigax10}. The case $n=1$, namely $\Iassoct{1}{X}{Y};(\lunitt{X}\per \id{Y})=\lunitt{X\per Y}$ follows from coherence axioms of symmetric monoidal categories.
         The case $n+1$ follows from the fact that 
        \[\Iassoct{\objn{n+1}}{X}{Y};(\delta^r_{n+1,X}\per \id{Y});\delta^r_{(n+1)X,Y}=\delta^r_{(n+1)X,Y}\]
        is equivalent to the commutativity of the following diagram:
        \[
        \begin{tikzcd}[column sep=large,row sep=large]
            \bullet & \bullet & \bullet && \bullet \\
            \bullet && \bullet && \bullet
            \arrow["{(\dr{1}{\objn{n}}{X}\per \id{Y})}", from=1-1, to=1-2]
            \arrow["{\dr{1\per X}{\objn{n}\per X}{Y}}", from=1-2, to=1-3]
            \arrow["{(\lunitt{X}\per \id{Y})\piu (\delta^r_{n,p}\per \id{Y})}", from=1-3, to=1-5]
            \arrow["{\id{X\per Y}\piu \delta^r_{nX,Y}}", from=1-5, to=2-5]
            \arrow["{\Iassoct{\objn{n+1}}{X}{Y}}", from=2-1, to=1-1]
            \arrow["{\dr{\objn{n+1}}{X}{Y}}"', from=2-1, to=2-3]
            \arrow["{\Iassoct{1}{X}{Y}\piu\Iassoct{\objn{n}}{X}{Y}}"{description}, from=2-3, to=1-3]
            \arrow["{\lunitt{X\per Y}\piu\delta^r_{n,X\per Y}}"', from=2-3, to=2-5]
        \end{tikzcd}\]
        with the left square communting by \ref{eq:rigax11} , and the right square commuting by \ref{ax:monoidaltriangle} + \ref{eq:symmax2} and induction hypotesis. The right equation is proved similarly.
    \end{proof}
\end{lemma}

{\bf Enrichment.}
Hereafter, we will $\codiag{X}^n \colon \PiuL[i=1][n]{X} \to X $ for the arrows inductively defined as follows:
\begin{equation}
\begin{array}{cccc}
    \codiag{X}^0 \defeq  \cobang{X} \;&\; \codiag{X}^1 \defeq \id{X}\;&\;\codiag{X}^{n+1} \defeq  (\id{X}\piu \codiag{X}^n);\codiag{X}
    \end{array}
\end{equation}
A simple inductive argument confirms that $\codiag{X}^n$ is natural in $X$.

\begin{lemma}\label{lemma: delta n per identita}
    For every $X,Y\in \Cat{C}$ it holds that 
    \[\codiag{X}^n\per \id{Y}=\delta^r_{nX,Y};\codiag{X\per Y}^n\qquad \id{Y}\per \codiag{X}^n= \delta^l_{Y,nX};\codiag{Y\per X}^n \]
    \begin{proof}
        The proof is obtained by indiction. The case $n=0$, that corresponds to $\bang{X}\per \id{Y}=\annl{Y};\bang{X\per Y}$ follows from Proposition \ref{prop:fcrig}. The case $n=1$ is trivial. The case $n+1$ follows is obtained as follows:
        \begin{align}
            \codiag{X}^{n+1}\per \id{Y} &= ((\id{X}\piu \codiag{X}^n);\codiag{X})\per \id{Y} \tag{def. $\codiag{X}^{n+1}$}\\&= (\id{X}\piu \codiag{X}^n)\per \id{Y};\codiag{X}\per \id{Y} \tag{Funct. of $\per$}\\
            &=(\id{X}\piu \codiag{X}^n)\per \id{Y};\dr{X}{Y}{Y};\codiag{X\per Y} \tag{Prop. \ref{prop:fcrig}}\\
            &=\dr{X}{nX}{Y};(\id{X\per Y}\piu (\codiag{X}^n\per \id{Y}));\codiag{X\per Y} \tag{Nat. of $\dr{X}{Y}{Y}$}\\
            &=\dr{X}{nX}{Y};(\id{X\per Y}\piu(\delta^r_{nX,Y};\codiag{X\per Y}^n));\codiag{X\per Y}. \tag{Ind. hyp.}\\
            &=\dr{X}{nX}{Y};(\id{X\per Y}\piu\delta^r_{nX,Y});(\id{X\per Y}\piu\codiag{X\per Y}^n);\codiag{X\per Y} \tag{Funct. of $\piu$}\\
            &=\delta^r_{(n+1)X,Y};\codiag{X\per Y}^{n+1}. \tag{Def. of $\delta^r_{(n+1)X,Y}$ and $\codiag{X}^{n+1}$}
        \end{align}
        The right equation is proved similarly.
    \end{proof}
\end{lemma}

By means of $\Br{t}_X$ and  $\codiag{Y}^n$, we define $\enriched{t}(h_1,\dots, h_n)\colon X\to Y$ for all arrows $h_1,\dots,h_n \colon X\to Y$ as
\begin{equation}\label{eq: definizione f(h_1,...,h_n)}
\enriched{t}(h_1,\dots, h_n) \defeq \Br{t}_X;\PiuL[i=1][n]{h_i};\codiag{Y}^n \colon X\to Y
\end{equation}
Note that for each $f$ in the signature $\Sigma$ with arity $n$, \eqref{eq: definizione f(h_1,...,h_n)} provide a function $\enriched{f}\colon\Cat{C}[X,Y]^n \to \Cat{C}[X,Y]$. Indeed, it is easy to see that each homset $\Cat{C}[X,Y]$ is a model for $\mathbb{T}$. Actually $\Cat{C}$ is monoidally enriched over the category of models of $\mathbb{T}$:
\begin{proposition}\label{prop:enrichement}
    Let $h_1,\dots,h_n \colon X\to Y$ and $g$ of the appropriate type. It holds that
        \begin{alignat*}{2}
            1. \quad & \enriched{t}(h_1,\dots, h_n) ; g = \enriched{t}((h_1 ; g),\dots,(h_n ; g)) \quad & \quad 
            3. \quad & \enriched{t}(h_1\,\dots, h_n) \per g = \enriched{t}((h_1 \per g),\dots,(h_n \per g)) \notag \\
            2. \quad &   g ; \enriched{t}(h_1\,\dots, h_n) = \enriched{t}((g; h_1 ),\dots,(g; h_n )) \quad & \quad 
            4. \quad & g \per \enriched{t}(h_1\,\dots, h_n) = \enriched{t}((g\per h_1 ),\dots,(g\per h_n )) \notag \\
        \end{alignat*}
\end{proposition}    
        \begin{proof}
            Equation 1 follows from the naturality of $\codiag{Y}^n$, while equation 2 follows from the naturality of $t_X$ (see Lemma \ref{lemma: f_X naturale}). For equation 3, consider $g \colon Z\to U$, then
            \begin{align}
                \enriched{t}(h_1\,\dots, h_n) \per g &= (\Br{t}_X;\PiuL[i=1][n]{h_i});\codiag{Y}^n\per g \tag{\ref{eq: definizione f(h_1,...,h_n)}}\\
                &=(\Br{t}_X \per \id{Z});(\PiuL[i=1][n]{h_i})\per g;\codiag{Y}^n \per \id{U} \notag\\
                &=(\Br{t}_X \per \id{Z});(\PiuL[i=1][n]{h_i})\per g;\delta^r_{nY,U};\codiag{Y\per U}^n \tag{Lemma \ref{lemma: delta n per identita}}\\
                &=(\Br{t}_X \per \id{Z});\delta^r_{nX,Z};(\PiuL[i=1][n]{h_i}\per g);\codiag{Y\per U}^n \tag{$\delta^r_{nY,U}$ nat.}\\
                &=\Br{t}_{X\per Z};(\PiuL[i=1][n]{h_i}\per g);\codiag{Y\per U}^n \tag{Lemma \ref{lemma: f_P x idR= f_PR}}\\
                &=\enriched{t}((h_1 \per g),\dots,(h_n \per g)). \notag
            \end{align}
            Equation 4 is proved similarly.
        \end{proof}
\begin{proof}[Proof of Proposition \ref{prop:kleisliTrig}]
By the discussion above \eqref{eq:Lawdiag}, there exists a  morphism of finite coproduct categories $i\colon \opcat{\lawveret} \to T_{\mathbb{T}}$.
By Proposition \ref{prop:Kleislifccd},  $\Sets_{T_{\mathbb{T}}}$ is a finite coproduct-copy discard rig category.
\end{proof}

%% file: appendices/sigmaetapespp.tex
\section{Appendix to Section~\ref{sec:tapediagrams}}\label{app:tape}

\begin{table}[t]
    \centering
    \tiny{\begin{tabular}{cc cc}
        \toprule
        \multicolumn{2}{c}{$\LW S {\id{P}} = \id{SP}$} & $\RW S {\id{P}} = \id{PS}$ & (\newtag{W1}{eq:whisk:id})\\[0.3em]
        \multicolumn{2}{c}{$\LW{S}{\t ; \s} = \LW{S}{\t} ; \LW{S}{\s}$} & $\RW{S}{\t ; \s} = \RW{S}{\t} ; \RW{S}{\s}$ & (\newtag{W2}{eq:whisk:funct})\\[0.3em]
        \multicolumn{2}{c}{$\LW{\uno}{\t} = \t$} & $\RW{\uno}{\t} = \t$ & (\newtag{W3}{eq:whisk:uno}) \\[0.3em]
        \multicolumn{2}{c}{$\LW{\zero}{\t} = \id{\zero}$} & $\RW{\zero}{\t} = \id{\zero}$ & (\newtag{W4}{eq:whisk:zero}) \\[0.3em]
        \multicolumn{2}{c}{$\LW{S}{\t_1 \piu \t_2} = \dl{S}{P_1}{P_2} ; (\LW{S}{\t_1} \piu \LW{S}{\t_2}) ; \Idl{S}{Q_1}{Q_2}$}  & $\RW{S}{\t_1 \piu \t_2} = \RW{S}{\t_1} \piu \RW{S}{\t_2}$ & (\newtag{W5}{eq:whisk:funct piu}) \\[0.3em]
        \multicolumn{2}{c}{$\LW{S \piu T}{\t} = \LW{S}{\t} \piu \LW{T}{\t}$}  & $\RW{S \piu T}{\t} = \dl{P}{S}{T} ; ( \RW{S}{\t} \piu \RW{T}{\t} ) ; \Idl{Q}{S}{T}$ & (\newtag{W6}{eq:whisk:sum}) \\[0.3em]
        \multicolumn{3}{c}{$\LW{P_1}{\t_2} ; \RW{Q_2}{\t_1} = \RW{P_2}{\t_1} ; \LW{Q_1}{\t_2}$} & (\newtag{W7}{eq:tape:LexchangeR}) \\[0.3em]
        $\RW S {\codiag{U}} = \codiag{US}$ & (\newtag{W8}{eq:whisk:diag}) &  $\RW S {\cobang{U}} = \cobang{US}$ & (\newtag{W9}{eq:whisk:bang}) \\[0.3em]
        $\RW S {\symmp{P}{Q}} = \symmp{PS}{QS}$ & (\newtag{W10}{eq:whisk:symmp}) & $\symmt{PQ}{S} = \LW{P}{ \symmt{Q}{S}} ; \RW{Q}{\symmt{P}{S}}$ & (\newtag{W11}{eq:symmper}) \\[0.3em]
        $\RW{S}{\t} ; \symmt{Q}{S} = \symmt{P}{S} ; \LW{S}{\t}$ & (\newtag{W12}{eq:LRnatsym}) & $\LW{S}{\RW{T}{\t}} = \RW{T}{\LW{S}{\t}}$ & (\newtag{W13}{eq:tape:LR}) \\[0.3em]
        $\LW{ST}{\t} = \LW{S}{\LW{T}{\t}}$ & (\newtag{W14}{eq:tape:LL}) & $\RW{TS}{\t} = \RW{S}{\RW{T}{\t}}$ & (\newtag{W15}{eq:tape:RR}) \\[0.3em]
        $\RW S {\dl{P}{Q}{R}} = \dl{P}{QS}{RS}$ & (\newtag{W16}{eq:whisk:dl}) & $\LW S {\dl{P}{Q}{R}} = \dl{SP}{Q}{R} ; \Idl{S}{PQ}{PR}$ & (\newtag{W17}{eq:whisk:Ldl}) \\
        \multicolumn{3}{c}{$\RW S {\Br{f}_U} =\Br{f}_{US}$} & (\newtag{W18}{eq:whisk:fU}) \\
        \bottomrule
    \end{tabular}}
    \caption{The algebra of whiskerings}  
    \label{table:whisk}
\end{table}

\begin{lemma}\label{lem:algwhisk}
    The laws in Table~\ref{table:whisk} hold.
\end{lemma}
\begin{proof}
    The proof is a straightforward adaptation of the inductive proofs for~\cite[Lemma 5.9]{bonchi2023deconstructing}. In particular, we do not have to prove the laws involving $\diag{U}$ and $\bang{U}$, as these are not part of the syntax of $\algt$-tapes, and we do not have to consider them as base cases in the remaining proofs.

    The additional law, \eqref{eq:whisk:fU}, is proved by induction on $S$: for $S = \zero$ we have that $\RW \zero {f_U} = \id{\zero}$, by definition of polynomial whiskering in Table~\ref{tab:producttape}. By definition of $\Br{f}_P$ in Table~\ref{table:def syn sugar}, we have that $\Br{f}_\zero = \id{\zero}$ and thus $\RW{\zero}{\Br{f}_U} = \Br{f}_{\zero} = \Br{f}_{U \per \zero}$. 
    
    Now for $S = W \piu S'$, we have that:
    \begin{align*}
        \RW{W \piu S'}{\Br{f}_U} &= \dl{U}{W}{S'} ; ( \RW{W}{\Br{f}_U} \piu \RW{S'}{\Br{f}_U} ) ; \Idl{\PiuL[][n]{U}}{W}{S'} \tag{\ref{eq:whisk:sum}} \\
        &=  ( \RW{W}{\Br{f}_U} \piu \RW{S'}{\Br{f}_U} ) ; \Idl{\PiuL[][n]{U}}{W}{S'} \tag{\cite[Lemma E.1.1]{bonchi2023deconstructing}} \\
        &= ( \Br{f}_{UW} \piu \Br{f}_{US'} ) ; \Idl{\PiuL[][n]{U}}{W}{S'} \tag{Table~\ref{tab:producttape}, Ind. hyp.} \\
        &= ( \Br{f}_{UW} \piu \Br{f}_{US'} ) ; \Idl{\PiuL[][n]{\uno}}{UW}{US'} \tag{*} \\
        &= \Br{f}_{UW \piu US'} \tag{Table~\ref{table:def syn sugar}} \\
        &= \Br{f}_{U(W \piu S')}
    \end{align*} 
    where (*) is proved by a straightforward induction on $n$ and using the inductive definition of $\delta^l$ in Table~\ref{table:def syn sugar}.
\end{proof}

\begin{proof}[Proof of Theorem \ref{thm:free signametape}]
    By construction, $(\CatTapeT[\Gamma][\algt], \piu, \zero)$ is a symmetric monoidal category. Moreover, for every $P \in (\sort^\star)^\star$, there are natural commutative monoids $(\codiag{P}, \cobang{P})$ defined as in Table~\ref{table:def syn sugar}, making it a fc category. 
    

    The proof that $\t_1 \per \t_2 \defeq \LW{P}{\t_2} ; \RW{S}{\t_1} $ is a monoidal product involves showing that the whiskerings satisfy certain laws, reported in Table~\ref{table:whisk} in Appendix~\ref{app:tape}. With these laws, it is easy to see that $(\CatTapeT[\Gamma][\algt], \per, \uno)$ is a strict symmetric monoidal category, as shown in~\cite[Theorem 5.10]{bonchi2023deconstructing}.
    Similary, the proof that $\CatTapeT[\Gamma][\algt]$ is a $\sort$-sesquistrict rig category is taken verbatim from that of \cite[Theorem 5.10]{bonchi2023deconstructing}.

    Now, obeserve that there is a morphism of fc categories $i \colon \opcat{\lawveret} \to \CatTapeT[\Gamma][\algt]$ defined inductively on objects as 
    \[ i(\zero) \defeq \zero \qquad \text{ and } \qquad i(1 + n) \defeq 1 \piu i(n), \]
    while on the arrows it maps the monoid structure of $\opcat{\lawveret}$ to the monoid structure of $\CatTapeT[\Gamma][\algt]$, i.e.
    \[ i(\codiag{1}) \defeq \codiag{1} \qquad\text{ and }\qquad i(\cobang{1}) \defeq \cobang{1} \]
    and each $f \colon n \to 1$ in $\sign$ to the corresponding tape diagram, i.e. $i(f) \defeq \Br{f}_\uno$. One can easily check that $i(-)$ satisfies the condition in Definition~\ref{def:Trig cat}, and thus $\CatTapeT[\Gamma][\algt]$ is a sesquistrict $\algt$-rig category. We conclude the proof by showing that it is also the free one.

    Given a monoidal signature $(\sort,\Gamma)$, a sesquistrict $\algt$-rig category $(\Cat{M} \to \Cat{C}, j \colon \opcat{\lawveret} \to \Cat{C})$ and an interpretation $\interpretation=(\alpha_\sort \colon \sort \to \ob{\Cat{M}}, \alpha_\Gamma \colon \Gamma \to Ar(\Cat{C}))$ of $(\sort,\Gamma)$ in $\Cat{M} \to \Cat{C}$, 
    the inductive extension of $\interpretation$, hereafter referred to as $\CBdsem{-}_\interpretation \colon \CatTapeT[\Gamma][\algt] \to \Cat{C}$, is defined as follows:

    \noindent{\small
    \renewcommand{\arraystretch}{1.5}
    \!\!\!\begin{tabular}{lllll}
    $\CBdsem{s}_{\interpretation} = \alpha_{\Gamma}(s) $& $\CBdsem{\, \tapeFunct{c} \,}_{\interpretation}= \CBdsem{c}_{\interpretation} $ & $\CBdsem{\codiag{U}}_{\interpretation}= \codiag{\alpha^\sharp_\sort(U)} $&$ \CBdsem{\cobang{U}}_{\interpretation} = \cobang{\alpha^\sharp_\sort(U)}$   &  $\CBdsem{\Br{f}_{U}}_{\interpretation} = j(f) \per \id{\alpha^\sharp_\sort(U)}$\\
    $\CBdsem{\id{A}}_{\interpretation}= \id{\alpha_\sort(A)} $&$ \CBdsem{\id{1}}_{\interpretation}= \id{1} $&$ \CBdsem{\symmt{A}{B}}_{\interpretation} = \symmt{\alpha_\sort(A)}{\alpha_\sort(B)}  $&$ \CBdsem{c;d}_{\interpretation} = \CBdsem{c}_{\interpretation}; \CBdsem{d}_{\interpretation}  $ & $ \CBdsem{c\per d}_{\interpretation} = \CBdsem{c}_{\interpretation} \per \CBdsem{d}_{\interpretation}$\\
    $\CBdsem{\id{U}}_{\interpretation}= \id{\alpha^\sharp_\sort(U)} $&$ \CBdsem{\id{\zero}}_{\interpretation} = \id{\zero}  $&$ \CBdsem{\symmp{U}{V}}_{\interpretation}= \symmp{\alpha^\sharp_\sort(U)}{\alpha^\sharp_\sort(V)} $&$ \CBdsem{\s;\t}_{\interpretation} = \CBdsem{\s}_{\interpretation} ; \CBdsem{\t}_{\interpretation} $&$ \CBdsem{\s \piu \t}_{\interpretation} = \CBdsem{\s}_{\interpretation} \piu \CBdsem{\t}_{\interpretation} $
    \end{tabular}
    }

    Observe that there is a trivial interpretation of $(\sort, \Gamma)$ in $(\sort \to \CatTapeT[\Gamma][\algt], i \colon \opcat{\lawveret} \to \CatTapeT[\Gamma][\algt])$ given by  $\interpretation' = (\alpha'_\sort \colon \sort \to \ob{\sort}, \alpha'_\Gamma \colon \Gamma \to Ar(\CatTapeT[\Gamma][\algt]))$, where $\alpha'_\sort(A) \defeq A$ and $\alpha'_\Gamma(\gen) \defeq \tapeFunct{\gen}$, for all $A \in \sort, \gen \in \Gamma$.

    We aim to find a unique sesquistrict $\algt$-rig functor $(\alpha \colon \sort \to \Cat{M}, \beta \colon \CatTapeT[\Gamma][\algt] \to \Cat{C})$ such that $\alpha'_\sort ; \alpha = \alpha_\sort$ and $\alpha'_\Gamma ; \beta = \alpha_\Gamma$. 

    First observe that for every $A \in \sort$, $\alpha'_\sort(A) = A$, and thus $\alpha(A)$ is forced to be $\alpha_\sort(A)$. Secondly, since we ask for $\beta$ to be a morphism of fc rig categories, it must preserve the two symmetric monoidal structures and the commutative monoid structure. The fact that $\beta$ preserves the rest of the rig structure can be proved using the inductive definitions in Table~\ref{table:def syn sugar}.

    The only degree of freedom is on the generators $\gen \in \Gamma$ and the $\Br{f}_U$'s.  However, observe that $\alpha'_\Gamma(\gen) = \tapeFunct{\gen}$, therefore $\beta$ must map $\tapeFunct{\gen}$, and thus $\gen$, in $\alpha_\Gamma(\gen)$.

    Finally, since we want $\beta$ to be a morphism of $(\sign, E)$-rig categories, we must require that $i ; \beta = j$ and so it must be such that $\beta(f_\uno) = j(f)$. Since $\beta$ preserves the rig structure and $\Br{f}_U = \Br{f}_1 \per \id{U}$, it is also the case that $\beta(\Br{f}_U) = \beta(\Br{f}_1 \per \id{U}) = \beta(\Br{f}_1) \per \beta(\id{U}) = j(f) \per \id{\alpha^\sharp_\sort(U)}$. This proves that the only such $\beta$ is $\CBdsem{-}_\interpretation$.\qedhere




    
\end{proof}

\begin{proof}[Proof of Theorem \ref{th:tapes-free-T-cd}]
To prove that $\CatTapeTCD[\Gamma][\algt]$ is a sesquistrict $\algt$-rig category, we proceed exactly as in Theorem~\ref{thm:free signametape}. To establish that it is also a copy-discard category, and therefore a $\algt$-cd rig, we must show that the operations $\copier{P}$ and $\discharger{P}$, defined in \eqref{eq:copierind}, satisfy the axioms of cocommutative comonoids for every polynomial $P \in (\sort^\star)^\star$. This follows from~\cite[Theorem 7.3]{bonchi2023deconstructing}.

Finally, the fact that $\CatTapeTCD[\Gamma][\algt]$ is the free sesquistrict $\algt$-cd rig follows from the same argument as in Theorem~\ref{thm:free signametape}. 
The key observation is that an interpretation $\interpretation=(\alpha_\sort \colon \sort \to \ob{\Cat{M}}, \alpha_\Gamma \colon \Gamma \to Ar(\Cat{C}))$ of $(\sort,\Gamma)$ in a sesquistrict $\algt$-cd rig category $(\Cat{M} \to \Cat{C}, j \colon \opcat{\lawveret} \to \Cat{C})$, gives rise uniquely to a morphism $\CBdsem{-}_\interpretation \colon \CatTapeTCD[\Gamma][\algt] \to \Cat{C}$, defined inductively as follows:

\begin{equation*}
    \renewcommand{\arraystretch}{1.5}
\!\!\!\begin{tabular}{l@{\;\;\;\;}l@{\;\;\;\;}l@{\;\;\;\;}l@{\;\;\;\;}l}
$\CBdsem{s}_{\interpretation} \defeq \alpha_{\Gamma}(s) $& $\CBdsem{\, \tapeFunct{c} \,}_{\interpretation}\defeq \CBdsem{c}_{\interpretation} $ & $\CBdsem{\codiag{U}}_{\interpretation}\defeq \codiag{\alpha^\sharp_\sort(U)} $&$ \CBdsem{\cobang{U}}_{\interpretation} \defeq \cobang{\alpha^\sharp_\sort(U)}$   &  $\CBdsem{\Br{f}_{U}}_{\interpretation} \defeq j(f) \per \id{\alpha^\sharp_\sort(U)}$\\
$\CBdsem{\id{A}}_{\interpretation}\defeq \id{\alpha_\sort(A)} $&$ \CBdsem{\id{1}}_{\interpretation}\defeq \id{1} $&$ \CBdsem{\symmt{A}{B}}_{\interpretation} \defeq \symmt{\alpha_\sort(A)}{\alpha_\sort(B)}  $&$ \CBdsem{c;d}_{\interpretation} \defeq \CBdsem{c}_{\interpretation}; \CBdsem{d}_{\interpretation}  $ & $ \CBdsem{c\per d}_{\interpretation} \defeq \CBdsem{c}_{\interpretation} \per \CBdsem{d}_{\interpretation}$\\
$\CBdsem{\id{U}}_{\interpretation}\defeq \id{\alpha^\sharp_\sort(U)} $&$ \CBdsem{\id{\zero}}_{\interpretation} \defeq \id{\zero}  $&$ \CBdsem{\symmp{U}{V}}_{\interpretation}\defeq \symmp{\alpha^\sharp_\sort(U)}{\alpha^\sharp_\sort(V)} $&$ \CBdsem{\s;\t}_{\interpretation} \defeq \CBdsem{\s}_{\interpretation} ; \CBdsem{\t}_{\interpretation} $&$ \CBdsem{\s \piu \t}_{\interpretation} \defeq \CBdsem{\s}_{\interpretation} \piu \CBdsem{\t}_{\interpretation} $ \\
 & & $\CBdsem{\copier{A}} \defeq \copier{\alpha_\sort(A)}$ & $\CBdsem{\discharger{A}} \defeq \discharger{\alpha_\sort(A)}$ & 
\end{tabular}
\end{equation*}

\end{proof}